\documentclass[11pt]{article}
\usepackage{amsmath}
\usepackage{amsfonts}
\usepackage{amssymb}
\usepackage{amsthm}
\usepackage{hyperref}
\usepackage{enumitem}
\usepackage{xspace}

\newlength{\dinwidth}
\newlength{\dinmargin}
\newcommand{\bmco}{B}
\newcommand{\ha}{\hat a}
\newcommand{\tOme}{\widehat{\Ome}}
\newcommand{\tvac}{\langle 0 \vert}
\newcommand{\Span}{\mathrm{Span}}
\newcommand{\tOmeR}{\tOme_{\mcR}}
\newcommand{\OmeR}{\Ome_{\mcR}}

\renewcommand{\Im}{\mathrm{Im\,}}
\newcommand{\Rs}{R}
\newcommand{\haDe}{\hat\De}
\newcommand{\haS}{\hat S}
\newcommand{\hah}{\hat h}
\newcommand{\haPsi}{\hat\Psi}
\newcommand{\cGaz}{\cGa^{(0)}}
\newcommand{\cGao}{\cGa^{(1)} }
\newcommand{\mcD}{\mathcal D}
\newcommand{\mcG}{\mathcal G}
\newcommand{\tP}{\tilde P}
\newcommand{\epsz}{\eps_0}
\newcommand{\mcoz}{\mco_0}
\newcommand{\ta}{\ti a}
\newcommand{\Nz}{N_0}
\newcommand{\Int}{\mathcal I}

\newcommand{\cdGa}{\mathrm{d}\check{\Ga}}
\newcommand{\Imm}{\mathrm{Im}}
\newcommand{\diag}{\mathrm{diag}}
\newcommand{\Omeg}{\Omega}
\newcommand{\bD}{\mathbf{D}}

\newcommand{\Ome}{\Omega}

\newcommand{\mcR}{\mathcal R}
\newcommand{\mcJ}{\mathcal J}
\newcommand{\mcJz}{\mcJ_{0}}
\newcommand{\cm}{c_{\mathrm{m}} }
\newcommand{\fin}{\mathrm{fin}}
\newcommand{\mcC}{\mathcal C}
\newcommand{\sic}{\mathrm{sc}}
\newcommand{\ac}{\mathrm{ac}}
\newcommand{\pp}{\mathrm{pp}}

\newcommand{\mcE}{\mathcal E}
\newcommand{\mcT}{\mathcal T}

\newcommand{\mcXo}{\mcX_1}
\newcommand{\mcS}{\mathcal S}
\newcommand{\iso}{\mathrm{iso} }
\newcommand{\mrO}{\mathrm O}
\newcommand{\LLP}{\mathrm{LLP}}
\newcommand{\mcA}{\mathcal A}
\newcommand{\mcX}{\mathcal X}
\renewcommand{\i}{\mathrm i}

\newcommand{\R}{t}
\newcommand{\Siz}{\Si}
\newcommand{\K}{\mathcal K}

\newcommand{\qt}{q^t}
\newcommand{\jt}{j^t}
\newcommand{\bnd}{\mathrm{bnd}}
\newcommand{\Hext}{H^{\ex}}

\newcommand{\qi}{j_{\infty}} 
\newcommand{\qz}{j_{0}}  
\newcommand{\qd}{q}

\newcommand{\slim}{\mathrm{s-}\lim}
\newcommand{\jj}{j}

\newcommand{\Exc}{\mathrm{Exc}}
\newcommand{\dcGa}{\mathrm{d}\cGa }
\newcommand{\ex}{\mathrm{ex}}
\newcommand{\dGa}{\mathrm{d}\Ga }
\newcommand{\g}{}

\newcommand{\wt}{\widetilde}

\newcommand{\su}{\substack}

\newcommand{\ti}{\tilde}

\newcommand{\un}{\underline}
\newcommand{\cGa}{\check{\Ga}}

\newcommand{\vac}{\vert 0\ran}
\newcommand{\Om}{\vac}

\newcommand{\Si}{\Sigma_0}

\newcommand{\ad}{\mathrm{ad}}

\newcommand{\be}{\beta}

\newcommand{\pa}{\partial}

\newcommand{\Ran}{\mathrm{Ran}}
\newcommand{\ov}{\overline}

\newcommand{\mfh}{\mathfrak{h}}

\newcommand{\uk}{\underline{k}}

\newcommand{\eps}{\varepsilon}
\newcommand{\de}{\delta}

\newcommand{\De}{\Delta}
\newcommand{\e}{\mathrm{e}}

\newcommand{\pho}{\mathrm{ph}}
\newcommand{\Fock}{\mathcal{F}}

\newcommand{\si}{\sigma}

\newcommand{\h}{\fr{1}{2}}
\newcommand{\nat}{\mathbb{N}}
\newcommand{\hil}{\mathcal{H}}
\newcommand{\om}{\omega}

\newcommand{\mco}{\mathcal{O}}

\newcommand{\supp}{\mathrm{supp}}
\newcommand{\fr}[2]{\frac{#1}{#2}}
\newcommand{\al}{\alpha}
\newcommand{\real}{\mathbb{R}}
\newcommand{\complex}{\mathbb{C}}
\newcommand{\la}{\lambda}
\newcommand{\non}{\nonumber}
\newcommand{\Ga}{\Gamma}

\newcommand{\lan}{\langle}
\newcommand{\ran}{\rangle}

\theoremstyle{plain}

\newtheorem{theoreme}{Theorem } [section]
\newtheorem{proposition}[theoreme]{Proposition}
\newtheorem{lemma}[theoreme]{Lemma}
\newtheorem{corollary}[theoreme]{Corollary}

\theoremstyle{definition}
\newtheorem{definition}[theoreme]{Definition}
\newtheorem{criterion}[theoreme]{Criterion}
\newtheorem{condition}{Condition}

\newtheorem{remark}[theoreme]{Remark}
\newtheorem{example}[theoreme]{Example}

\newcommand{\beco}{\begin{condition}}
\newcommand{\eeco}{\end{condition}}
\newcommand{\beq}{\begin{equation}}
\newcommand{\eeq}{\end{equation}}
\newcommand{\beqa}{\begin{eqnarray}}
\newcommand{\eeqa}{\end{eqnarray}}
\newcommand{\ben}{\begin{arabicenumerate}}
\newcommand{\een}{\end{arabicenumerate}}
\newcommand{\bex}{\begin{example}}
\newcommand{\eex}{\end{example}}
\newcommand{\ber}{\begin{remark}}
\newcommand{\eer}{\end{remark}}
\newcommand{\bec}{\begin{corollary}}
\newcommand{\eec}{\end{corollary}}
\newcommand{\bep}{\begin{proposition}}
\newcommand{\eep}{\end{proposition}}
\newcommand{\becr}{\begin{criterion}}
\newcommand{\eecr}{\end{criterion}}

\newcommand{\localized}{localized\xspace}
\newcommand{\Localized}{Localized\xspace}

\newcommand{\set}[2]{\{#1\,|\, #2\}}
\newcommand{\bigset}[2]{\bigl\{#1\,\big|\, #2\bigr\}}
\newcommand{\Bigset}[2]{\Bigl\{#1\,\Big|\, #2\Bigr\}}
\newcommand{\sico}{\mathrm{sc}}

\def\bel{\begin{lemma}}
\def\eel{\end{lemma}}
\def\bet{\begin{theoreme}}
\def\eet{\end{theoreme}}
\def\bed{\begin{definition}}
\def\eed{\end{definition}}

\begin{document}

\bibliographystyle{amsplain}

\title{The translation invariant massive Nelson model: III. Asymptotic completeness below the two-boson threshold}

\author{
Wojciech Dybalski \\
Zentrum Mathematik, TU M\"unchen \\
       and \\ 
Department of Mathematics, University of Paris-Sud\\
email: dybalski@ma.tum.de \\
\\ and \\ \\
Jacob Schach M{\o}ller \\ Department of Mathematics,  Aarhus University\\
email: jacob@imf.au.dk}

\maketitle

\begin{abstract}
We show asymptotic completeness of two-body scattering 
for a class of translation invariant models describing a single quantum
particle (the electron) linearly coupled to a massive
scalar field (bosons).  Our proof is based on a recently established
Mourre estimate for these models. In contrast to previous approaches, it
requires no number cutoff, no restriction on the particle-field coupling strength, 
and no restriction on the magnitude of total momentum.  Energy, however, is restricted 
by the two-boson threshold,  admitting only scattering of a dressed electron and a single
asymptotic boson.  The class of models we consider include the UV-cutoff Nelson and polaron models. 
\end{abstract}

\setcounter{page}{0}
\thispagestyle{empty}
\newpage

\tableofcontents
\section{Introduction}

The last two decades witnessed substantial progress in our understanding of asymptotic completeness (AC) in
 Quantum Field Theory (QFT). 
On the relativistic side first examples of massive  and massless 
theories with complete particle interpretation have been constructed in \cite{Le08, DT10}. On the side of
non-relativistic QFT,  far-reaching insights 
have been obtained by application of methods from many-body quantum mechanics \cite{En78,D,Gr,SiSo, DGBook,GeLaBook}.
AC of systems describing a \emph{confined} quantum-mechanical particle (the electron) interacting with second-quantized Bose fields
is well under control in the case of massive field quanta (bosons) \cite{HuSp1,HuSp2,DGe1,DGe2,FGSch2, Am} and there is 
rapid progress on the massless side \cite{Sp97,Ge1, DK11, FS12, FS12b}.
However,  the case of \emph{translation invariant} quantum-mechanical systems coupled to quantum fields
is far from being fully understood, even if the bosons are massive. The main difficulty here is
the phenomenon of the electron mass renormalization, familiar from relativistic QFT. In the existing works 
this difficulty is overcome only  at a cost of  technical assumptions on the coupling strength,
total momentum of the system and dispersion relations of the  electron and bosons \cite{FGSch3,FGSch4} 
or by means of a  number cutoff \cite{GMR}. In the present paper we 
show that all these restrictions can be eliminated, at least at the level of two-body scattering:
We show AC below the two-boson threshold in a class of translation invariant massive QFT under very general assumptions, 
including the massive Nelson model \cite{Ne} and the Fr\"ohlich polaron model \cite{FrH} with 
physical (infrared-singular) coupling function. We stress that in the case of the polaron model, with constant
dispersion relation of bosons, the physical picture of propagating particles is not self-evident, not to speak of AC.
 It comes to light only after taking the electron mass renormalization properly into account and extracting the effective dynamics of
the electron-boson system. This is achieved for the first time in  the present work.

We consider a class of models describing a free quantum-mechanical particle, e.g. a non-relativistic  electron, linearly coupled
to a UV-cutoff massive scalar field, e.g. longitudinal optical  phonons or massive relativistic bosons. 
The isolated energy-momentum spectrum, i.e. the region below the one-boson threshold, 
is under our assumptions an analytic variety. It consists  of the ground state mass shell, which is non-degenerate for all total momenta,
and possibly excited isolated mass shells that may cross each other. 
To each mass shell one can associate a distinct dressed electron species. They have different dispersion relations, hence different masses, and some may even have group velocity in a direction opposite to momentum (non-increasing dispersion).
 Incoming and outgoing states are of the  form  $\Psi\otimes \eta$, where $\Psi$ is a dressed electron state
 (or superpositions thereof) and  $\eta$ is a vector in Fock space describing a collection of free asymptotic bosons.
We note in passing that  during a scattering process the outgoing dressed electron may differ from the incoming  dressed electron
i.e. the dressed electron species may not be conserved by collisions with  bosons.
The central objects of our investigation are the (conventional) wave operators, defined in (\ref{canonical-wave-operator}) below, which map incoming/outgoing states to states in the physical Hilbert space. In particular $\Psi\otimes \vac$, where $\vac$ is the vacuum vector, is mapped into
the dressed electron state $\Psi$. For general $\eta$, vectors from the ranges of the wave operators describe scattering states of dressed electrons and bosons. As usual, AC is defined as unitarity of the wave operators, which means
that all states of physical interest belong to their ranges.

Existence of the wave operators is known for the Nelson model \cite{DGe1}, but not for the polaron model. 
In the present paper we construct the wave operators and prove AC under rather natural assumptions which cover both the Nelson and polaron case: We employ no number cutoff, hence a dressed electron consists of a bare electron accompanied by an infinite virtual boson cloud.  There are no restrictions on the electron-field coupling strength and no limitations  on the magnitude of total momentum.  The energy is only restricted by the  (total momentum dependent) two-boson threshold which defines the largest spectral subspace on which only single-boson scattering processes take place.  Above this threshold, we are not -- yet -- able to handle the plethora of scattering channels available.

To explain the novel strategy of our proof  of AC, we recall several standard concepts, which will be defined precisely in Section~\ref{Preliminaries} and Appendix~\ref{Fock-combinatorics}.
We use the $\Gamma$-functor notation of Segal for constructions of spaces and operators
in the context of second quantization.
The   Hilbert spaces of incoming and outgoing  configurations are given by $\hil_{\pm}:=\hil_{\bnd}\otimes \Fock$, where $\hil_{\bnd}$ contains
the dressed electron   states and
$\Fock$ is the bosonic Fock space over the single-boson space $\mfh$. The extended Hamiltonian and momentum operators are defined as
\begin{equation}
 H^{\ex}:=H\otimes 1+1\otimes \dGa(\om)
\qquad \textup{and} \qquad 
 P^{\ex}:=P\otimes 1+1\otimes \dGa(k),
\end{equation}
where $\om$ is the dispersion relation of the bosons and $(P,H)$ 
denote the total energy-momentum operators of our system. We recall that $H$ acts on $\hil_{\bnd}$ as a direct sum of  multiplication 
operators, one for each dressed electron species.
For any pair of bounded operators $q_0,q_{\infty}$ on $\mfh$ we define the map  $\cGa(q_0,q_{\infty})^*$, 
from a domain in $\hil_{\pm}$ to $\hil$, by the relation:
\beqa
\cGa(q_0,q_{\infty})^*(\Psi\otimes a^*(h_1)\ldots a^*(h_n)\vac)=a^*(q_{\infty}h_1)\ldots a^*(q_{\infty}h_n)\Ga(q_0)\Psi. \label{dGa-check-intro}
\eeqa
The goal of our investigation is to establish the existence and unitarity of the wave operators
\beqa
\Ome^{\pm}=\slim_{t\to\pm\infty} \e^{\i tH}\cGa(1,1)^*\e^{-\i tH^{\ex}}, \label{canonical-wave-operator}
\eeqa
below the two-boson threshold (in the joint spectrum of $(H^{\ex},P^{\ex})$). 
For reasons which will become clear below, we
divide this region of the spectrum into small subsets $\mco\subset \real^\nu\times\real$.  For each $\mco$ we construct a \localized right inverse
of $\Ome^\pm$ on the corresponding spectral subspace of $(P,H)$.  As noted in \cite{DGe1}, a natural candidate has the form
\beqa
W^{\pm*}_{\mco}=\slim_{t\to\pm\infty} \e^{\i tH^{\ex}}\cGa(q_0^t,q_{\infty}^t)\e^{-\i tH} \label{inverse-wave-operator}
\eeqa
where $q_0^t$, $q_{\infty}^t$ are some time-dependent families of operators s.t.  $q_0^t+q_{\infty}^t=1$ so that  one can exploit the relation $\cGa(1,1)^*\cGa(q_0^t,q_{\infty}^t)=1$. One important difference
between our approach and previous work on asymptotic completeness in QFT consists in the construction of the operators $q_0^t,q_{\infty}^t$.

Before we explain this construction, we recall that the Hamiltonian $H$ has a direct integral decomposition into fiber Hamiltonians
$H(\xi)$ at fixed momentum $\xi$. 
As shown in \cite{MR12}, and stated precisely in  Theorem~\ref{Mourre-estimate-theorem} below, 
if $\mco$ is sufficiently small (and localized outside of some sets of measure zero) we can choose
$(\xi_0,\la_0)\in \mco$, a neighbourhood $\mcJz$ of $\lambda_0$, and $\cm>0$ s.t.
\begin{align}
&\mathbf{1}_{\mcJz} (H(\xi)) \i [H(\xi),  \dGa(a_{\xi_0})] \mathbf{1}_{\mcJz}(H(\xi))
\geq \cm \mathbf{1}_{\mcJz}(H(\xi)), \label{intro-positive-commutator}\\
& \mathbf{1}_{\mcJz}(H^{(1)}(\xi))\i [H^{(1)}(\xi),  1\otimes a_{\xi_0}]\mathbf{1}_{\mcJz}(H^{(1)}(\xi))
\geq \cm \mathbf{1}_{\mcJz}(H^{(1)}(\xi)), \label{intro-Mourre-est-one}
\end{align}
where $H^{(1)}:=H\otimes 1+1\otimes\om$ acts on $\hil^{(1)}=\hil\otimes\mfh$.
The estimates hold true for $\xi$ belonging to a small neighbourhood of $\xi_0$,
such that the Cartesian product of this neighbourhood with $\mcJz$
 contains  $\mco$. The operator $a_{\xi_0}$ has the form
\beq\label{Intro-a}
a_{\xi_0} = \h\bigl\{ v_{\xi_0}\cdot \i\nabla_k + \i\nabla_k\cdot v_{\xi_0}\bigr\},
\eeq
where  $\i\nabla_k$ is  the boson position operator and $v_{\xi_0}$ is a vector field in momentum space, which carries 
information about the dispersion relations of incoming/outgoing dressed electrons present in the
energy-momentum region $\mco$.  Now we define  $\ti a_{\xi_0}:=\h\bigl\{ v_{\xi_0}\cdot z + z\cdot v_{\xi_0}\bigr\}$,
where $z=\i \nabla_k-y$ is the relative distance between the electron and the boson, and set 
\begin{align}
q_0^t:=q_0(\ti a_{\xi_0}/t) \qquad \textup{and}\qquad  q_\infty^t:=q_\infty(\ti a_{\xi_0}/t), \label{q-definitions}
\end{align}
where $q_0$, $q_{\infty}$ are smooth approximate characteristic functions  of $(-\infty, c_0]$, $[c_{0},\infty)$, 
$c_0>0$ is smaller than $\cm$, and $q_0+q_{\infty}=1$. With such a choice of   $q_0^t$, $q_{\infty}^t$,
closely tied to  Mourre theory, strong convergence in (\ref{inverse-wave-operator}) can be established using  
the positive commutator  estimates~(\ref{intro-positive-commutator}),(\ref{intro-Mourre-est-one}).
We note that this convergence result holds only in  the spectral subspace of $\mco$. Indeed, only in this subspace
estimate (\ref{intro-positive-commutator}) holds  with the operator $v_{\xi_0}$, which entered into 
the definitions~(\ref{q-definitions}).  The fact that    $W^{\pm*}_{\mco}$ has to be defined for each region $\mco$ separately 
is, however, not an obstacle, since we use  this operator only as a tool to show the existence and unitarity of the
wave operators   $\Ome^\pm$, which do not contain any information about the (non-canonical) operators $v_{\xi_0}$. 

A large part of our paper is devoted to the proof of strong convergence of the localized inverse of the wave
operator  in (\ref{inverse-wave-operator}) with the help of the Mourre estimates. An important intermediate step here is a
novel minimal-velocity propagation estimate (See Proposition~\ref{main-propagation-estimate} below). As our proof
of this propagation estimate differs significantly from the arguments available in the literature, let us state here its special
case and outline the proof:
 Let $j_0$, $j_{\infty}$ be smooth approximate characteristic functions  of $(-\eps,\eps)$,  $\real\backslash (-\eps,\eps)$ s.t. 
$j_0^2+j_{\infty}^2=1$ and let  $j^t:=(j_0(a_{\xi_0}/t), j_{\infty}(a_{\xi_0}/t))$. 
Then there exists $c>0$ such that for  all $\Psi\in\Fock$:
\beqa
\int_1^{\infty}dt\, \fr{1}{t}\lan\Psi_t,\cGa(j^t)^*\chi^{(1)}(1\otimes q'(a_{\xi_0}/t) )\chi^{(1)}\cGa(j^t)\Psi_t\ran\leq c\|\Psi\|^2,
\label{multi-particle-prop-est-intro}
\eeqa
where  $q'$ is a smooth approximate characteristic function of $\Int:=[-R,-\eps]\cup [\eps,c_0]$, $\Psi_t:=\e^{-\i t H(\xi)  }\Psi$, $\chi^{(1)}:=\chi(H^{(1)}(\xi))$ and $\chi\in C_0^{\infty}(\real)$ is supported below the two-boson threshold.  
Proceeding to  the proof of  (\ref{multi-particle-prop-est-intro}),  let us consider a propagation observable $\Phi(t):=\chi\dGa(q^t)\chi$ where
 $\chi:=\chi(H(\xi))$, $q(\la):=\int_0^{\la} q'(s) ds$  and $q^t:=q( a_{\xi_0}/t)$. In the standard
proofs of propagation estimates in non-relativistic QFT \cite{DGe1,FGSch3} one computes to the leading order in $t$ the Heisenberg derivative
\beqa
\bD\Phi(t)=\pa_t \Phi(t)+\i[H(\xi),\Phi(t)]
\eeqa
making use of the concrete expression (\ref{fiber-Hamiltonian}) for the  Hamiltonian $H(\xi)$. In the presence of the electron mass renormalization
this strategy breaks down for large coupling strength, because it introduces into the analysis the bare dispersion relation $\Omeg$ of the electron,
appearing in~(\ref{fiber-Hamiltonian}). To extract the correct physical dynamics of the electron-boson system we proceed differently:
Making use of the fact that $\cGa(j^t)^*\cGa(j^t)=1$, we write
\beqa
\bD\Phi(t)=\cGa(j^t)^*\cGa(j^t) \chi\bD \dGa(q^t)\chi=\cGa(j^t)^*\chi^{(1)} \bD^{(1)} (1\otimes q^t)\chi^{(1)}\cGa(j^t)+O(t^{-2}),\label{intro-proof}
\eeqa
where $\bD^{(1)}$ is the Heisenberg derivative w.r.t. the Hamiltonian $H^{(1)}(\xi)$ and  $O(t^{-2})$ denotes a term bounded in norm by $ct^{-2}$.  The last step in (\ref{intro-proof}), justified  in Proposition~\ref{a-dGa-Heisenberg}, consists in commuting $\cGa(j^t)$ to the right and showing that the
resulting rest-terms are of order $O(t^{-2})$. Here we only indicate how to exchange $\cGa(j^t)$ with $\chi$, since it contains the essence of the 
argument: First, we make use of the fact that $\cGa(j^t) \chi(H(\xi))=\chi(H^{\ex}(\xi))\cGa(j^t)+O(t^{-1})$ (Lemma~\ref{Helffer-dGa}). Next, we  
exploit that $\chi$ is localized below the two-boson threshold  to write $\chi(H^{\ex}(\xi))=\chi(H(\xi))\oplus \chi(H^{(1)}(\xi))$ 
(Lemma~\ref{expansion-truncation}).
Finally, we show that the first term in this direct sum gives rise to  expressions of order $O(t^{-2})$ if $j_0$ is supported outside of $\Int$.
Given expression (\ref{intro-proof}), we estimate the commutator $\i[H^{(1)}(\xi), (1\otimes q^t)]$ from below, using the 
Mourre estimate~(\ref{intro-Mourre-est-one}) and, by integrating both sides of the resulting expression along the time evolution, we obtain 
the propagation estimate (\ref{multi-particle-prop-est-intro}).

It is clear from the above discussion that our proof of AC is very different from the standard arguments used in the absence
of the electron mass renormalization \cite{Gr,DGe1} or in the weak coupling regime  \cite{FGSch3}. 
In particular, our argument does not rely on the phase-space propagation estimate, which is problematic in the presence of level
crossings in the isolated spectrum.  By our methods  we can handle a large class of electron and boson dispersion relations and, due to the fact that  $v_{\xi_0}$ can be chosen to vanish for small momenta, we can  cover the infrared-singular physical coupling of the polaron model.
In addition, no smallness conditions on the coupling strength are involved.  Thus, similarly to the classical results
on asymptotic completeness in quantum mechanics \cite{ D, Gr, SiSo},  our result applies to a very large  class of models  
which contains  experimentally realizable physical systems (e.g. the polaron). We are convinced that our analysis provides a
solid fundation for future developments of scattering theory in QFT.

Going beyond the two-boson threshold for the models studied here will be a challenging task requiring more involved constructions of propagation observables, due to the more complicated channel structure. While we do have some ideas as to how to proceed, there are technical obstructions requiring new insights to overcome. Another promising direction of future research  
concerns the spectral and scattering theory of many-body dispersive systems.
The methods developed in this paper, combined  with those of \cite{MR12}, can be viewed from a broader perspective as a general strategy to deal with 
such systems. We hope -- in fact expect -- that one can study many body Schr{\"o}dinger operators, with relativistic kinetic energy, as well as spin-wave scattering, i.e. the magnon model, with the aid of the techniques developed here. See \cite{Ge98, GS97, Ya92}, where both of these long-standing open problems are discussed. Finally, we would like to point out that  collision theory of dispersive systems is an important intermediate step towards
the problem of asymptotic completeness in local relativistic QFT, as for example the $P(\phi)_2$ models. This observation has recently been exploited in \cite{DG12} to show the existence of certain asymptotic observables in these theories. Thus an  application of the methods of the present paper in the local relativistic setting is another promising -- and tractable --  research direction. We recall that partial results on asymptotic completeness in $P(\phi)_2$ models can be found in \cite{CD82,SZ76}.  For recent progress on relativistic scattering theory we refer  to \cite{Dy05,DT10,Le08}.

This paper is organized as follows: In Section~\ref{Preliminaries} we define the class of models under study,  summarize the
known facts concerning their spectrum, including Mourre theory,  and state the main results of this paper. In Section~\ref{Heisenberg} we
 derive convenient representations for the Heisenberg derivatives of certain propagation observables which are then combined with Mourre estimates 
in Section~\ref{propagation-estimates} to derive minimal velocity propagation estimates. These propagation estimates are the key input
to the proofs of existence of the relevant asymptotic observables in Section~\ref{time-convergence-proofs}, including the 
\localized inverses of the wave operators of the form~(\ref{inverse-wave-operator}). In Section~\ref{Geometric-WO} we establish
properties of these operators which are then used in Section~\ref{AC-section} to prove the existence and unitarity of the (conventional) 
wave operators~(\ref{canonical-wave-operator}). More technical steps of our investigation are postponed to appendices.

\vspace{0.5cm}

\noindent{\bf Acknowledgment:} This project started in collaboration with Morten Grud Rasmussen, 
who contributed to a proof of AC for the  polaron model with a short-range condition. 
This different proof,  which preceded the present argument, will be published in a separate paper by the present authors and Morten Grud Rasmussen.

The authors thank Jan Derezi\'nski, Christian G\'erard, and Herbert Spohn for useful discussions we had during the course of this work. We acknowledge financial support of the Danish Council for Independent Research grant no. 09-065927 "Mathematical Physics", and hospitality of the Hausdorff Research Institute for   Mathematics, Bonn. Moreover, W.D. is grateful for the support of the Lundbeck Foundation and the  German Research Foundation (DFG), the latter within the grant SP181/25--2 and stipend DY107/1--1.

\section{Preliminaries and Results} \label{Preliminaries}
\setcounter{equation}{0}
\subsection{Hamiltonian}
Let $\K=L^2(\real^{\nu}_y)$ be the Hilbert space of a quantum mechanical particle moving in $\real^{\nu}$,
whose position is denoted by $y$ and  momentum by $D_y:=-\i\nabla_y$. Let $\mfh=L^2(\real^{\nu}_k)$ be the Hilbert space of a 
single boson, whose dispersion relation will be denoted $\om(k)$. The Hilbert space for the Bose field is the Fock space
\beqa
\Fock=\Ga(\mathfrak h)=\bigoplus_{n=0}^{\infty}\Fock^{(n)},
\eeqa
 where $\Fock^{(n)}=\Ga^{(n)}(\mfh)=\mfh^{\otimes_s n}$ is the symmetric tensor product of the single-boson spaces and the vacuum vector will be denoted by $\vac$. The boson creation and annihilation
operators are denoted by $a^*(k)$, $a(k)$ and satisfy the canonical commutation relations $[a(k),a^*(k')]=\de(k-k')$ and $[a(k),a(k')]=[a^*(k),a^*(k')]=0$. 
The total energy and momentum operators of the bosons are given by
\begin{align}
H_{\pho}&:= \dGa(\om)=\int_{\real^{\nu}} dk\, \om(k)a^*(k)a(k),\\
P_{\pho}&:= \dGa(k)=\int_{\real^{\nu}} dk\, k a^*(k)a(k).
\end{align}
The Hilbert space of the system consisting of the electron and the bosons is $\hil=\K\otimes \Fock$. The dynamics is governed by the Hamiltonian
\beq
H=\Omeg(D_y)\otimes 1+1\otimes H_{\pho}+\g\phi(G_y), \label{The-model-hamiltonian}
\eeq
where the interaction term is given by
\beqa
\phi(G_y):=\int_{\real^{\nu}}dk\, \big(\e^{-\i ky} G(k) 1\otimes a^*(k)+\e^{\i ky}\overline{G(k)} 1\otimes a(k) \big). \label{interaction-definition}
\eeqa
Under the minimal conditions on $\Omeg$, $\om$ and $G$, specified below following \cite{MR12}, this Hamiltonian is essentially self-adjoint on 
$C_0^{\infty}(\real^{\nu})\otimes \mcC$, where  $\mcC:=\Ga_{\fin}(C_0^{\infty}(\real^{\nu}))$ is defined in Appendix~\ref{Fock-combinatorics}. 
\beco \bf(Minimal Conditions)\rm. \label{Condition-one}  There exists $s_{\Omeg}\in [0,2]$ and  $C>0$ s.t. the dispersion relation $\om$ and the coupling function $G$ satisfy:
\begin{enumerate}[label=\textbf{\textup{(MC\arabic*)}},ref=\textup{(MC\arabic*)},leftmargin = *]
\item\label{MC-G-decay} $\om\in C(\real^{\nu})$, $\Omeg\in C^{2}(\real^{\nu})$,  $\lan k \ran^{6}G\in L^2(\real^{\nu})$,  
where $\lan k \ran=\sqrt{k^2+1}$.
\item\label{MC-Boson-Mass} $m:=\inf_{k\in\real^{\nu}}\om(k)>0$.
\item \label{MC-three}$\forall k\in\real^{\nu}$ we have $\om(k)\leq C\lan k\ran$, $\Omeg(k)\geq C^{-1}\lan k\ran^{s_{\Omeg}}-C$.
\item $|\pa_{k}^{\al}\Omeg(k)|\leq C\lan k\ran^{s_{\Omeg}-|\al|}$, for all multiindices $\al$ with $0\leq |\al|\leq 2$.
\item\label{MC-Subadd} $\forall k_1,k_2\in \real^{\nu}$ we have $\om(k_1+k_2)<\om(k_1)+\om(k_2)$.
\item Either $\lim_{|k|\to \infty}\om(k)=\infty$ or: $\sup_{k\in\real^{\nu}}\om(k)<\infty$ and $\lim_{|k|\to\infty}\Omeg(k)=\infty$.
\end{enumerate}
We note that \ref{MC-G-decay} is stronger than in \cite{MR12}. 
\eeco
We recall that the Hamiltonian (\ref{The-model-hamiltonian}) commutes with the total momentum  operators given by
\beqa
P=D_y\otimes 1+1\otimes P_{\pho},
\eeqa
thus it has a fiber decomposition. More precisely, using the unitary transform of Lee-Low-Pines \cite{LLP}
\beqa
I_{\LLP}:=(F\otimes 1)\circ \Ga(\e^{\i k\cdot y}),
\eeqa
where $F$ is the Fourier transform in the electron position variable and $\Ga$ the second quantization functor (cf. Appendix~\ref{Fock-combinatorics}), we obtain
\beqa
I_{\LLP}H I_{\LLP}^*=\int^{\oplus}_{\real^{\nu}} d\xi\, H(\xi).
\eeqa
The fiber Hamiltonians have the form
\beq
H(\xi)=\Omeg(\xi-P_{\pho})+H_{\pho}+\g\phi(G), \label{fiber-Hamiltonian}
\eeq
where $\phi(G):=\phi(G_y)|_{y=0}$,  and are essentially self-adjoint on $\mcC$. 
The joint spectrum of the family of commuting self-adjoint  operators $(P,H)$ is given by 
\beqa
\Sigma=\bigset{ (\xi,\la)\in\real^{\nu+1} }{ \la\in \si(H(\xi))}.
\eeqa
It can be decomposed into the pure-point,  absolutely continuous and singular continuous parts
\beqa
\Sigma=\Sigma_{\pp}\cup\Sigma_{\ac}\cup\Sigma_{\sic}
\eeqa 
defined as $\Sigma_i=\{\, (\xi,\la)\in\real^{\nu}\times \real\,|\, \la\in \si_i (H(\xi))\,\}$, where $i\in\{\pp,\ac,\sic\}$.
We denote the bottom of the spectrum of the fiber Hamiltonians by 
\beq
\Siz(\xi):=\inf\,\si(H(\xi))
\eeq
and the bottom of the spectrum of the full operator by
$\Si:=\inf_{\xi\in\real^{\nu}}\,\Si(\xi)$. 
Moreover, we introduce
\beq
\Siz^{(n)}(\xi,\un k):=\Siz(\xi-\sum_{j=1}^n k_j)+\sum_{j=1}^n\om(k_j)
\eeq
and define the $n$-boson thresholds
\beq
\Siz^{(n)}(\xi):=\inf_{\un k\in \real^{n\nu}}\Siz(\xi,\un k).
\eeq
By the HVZ Theorem \cite{FrJ2,MAHP,MRMP,Sp2}, 
\beq
\si_{\mathrm{ess}}(H(\xi))=\bigl[\Siz^{(1)}(\xi),\infty \bigr). \label{sigma-essential}
\eeq
and below $\Siz^{(1)}(\xi)$ the spectrum consists of locally finitely many eigenvalues of finite multiplicity,
which can only accumulate at $\Siz^{(1)}(\xi)$.  Due to the subadditivity assumption \ref{MC-Boson-Mass} on $\om$, we
have
\beqa
\Sigma_0^{(n)}(\xi)\geq \Sigma_0^{(m)}(\xi)
\eeqa
for any $n>m$. The inequality is strict if $\lim_{|k|\to\infty}\om(k)=\infty$. If $M=\sup_{k\in\real^{\nu}}\om(k)<\infty$,
then the inequality is also strict if $2\liminf_{|k|\to\infty}\om(k)>M$, which is satisfied by the constant polaron relation \cite{MRMP}. 
In these cases the region $\mcE^{(1)}$, where
\begin{equation}
\begin{aligned}
\mcE^{(1)}&=  \bigset{(\xi,\la)\in \real^{\nu+1}}{\la\in  \mcE^{(1)}(\xi)},\\
\mcE^{(1)}(\xi)&= \bigset{\la\in \real}{ \Sigma_0^{(1)}(\xi)\leq \la< \Sigma_0^{(2)}(\xi)}, \label{mcE-one}
\end{aligned}
\end{equation}
is non-empty.

\subsection{Extended Hamiltonian}\label{Extended-Hamiltonian}

The formalism of extended Hilbert space, which we  present in this section and in Appendix~\ref{Extended-appendix}, was introduced in \cite{DGe1} and used later on in \cite{Am,DGe2,FGSch2,MAHP,MRMP} in the context of spectral and scattering theory. 
Let us define the extended Fock space and the extended physical Hilbert space as follows
\begin{align}
\Fock^{\ex}&= \Fock\otimes\Fock=\Fock\oplus \Bigl(\bigoplus_{\ell=1}^{\infty}\Fock\otimes \Fock^{(\ell)}\Bigr)
\simeq \Fock\oplus  \Bigl(\bigoplus_{\ell=1}^{\infty}L^2_{\mathrm{sym}}(\real^{\ell\nu};\Fock)     \Bigr)      ,\\
\hil^{\ex}&= \hil\otimes \Fock=\hil\oplus \Bigl(\bigoplus_{\ell=1}^{\infty}\hil\otimes \Fock^{(\ell)}\Bigr),
\end{align}
where we made use of the  identification $ \Fock\otimes \Fock^{(\ell)} \simeq L^2_{\mathrm{sym}}(\real^{\ell\nu};\Fock)$.
The extended Hamiltonian and extended total momentum operators are given by  
\begin{align}\label{Hex-Def-And-Decomp}
H^{\ex}&= H\otimes 1+1\otimes \dGa(\om)=H\oplus \Bigl( \bigoplus_{\ell=1}^{\infty} H^{(\ell)}   \Bigr),\\
P^{\ex}&= P\otimes 1+1\otimes \dGa(k)=P\oplus \Bigl( \bigoplus_{\ell=1}^{\infty} P^{(\ell)}   \Bigr).
\end{align}
Here
\begin{equation}\label{Hl-And-Pl}
H^{(\ell)} = H\otimes 1 + 1\otimes\dGa^{(\ell)}(\omega), \qquad P^{(\ell)} = P\otimes 1 + 1\otimes \dGa^{(\ell)}(k).
\end{equation}
The operators
$(H^{\ex}, P^{\ex})$ are  essentially self-adjoint on $C_0^{\infty}(\real^{\nu})\otimes \mcC^{\ex}$, where $\mcC^{\ex}:=\mcC\otimes\mcC$. 
Similarly, $(H^{(\ell)}, P^{(\ell)})$ are essentially self-adjoint on $C_0^{\infty}(\real^{\nu})\otimes \mcC^{(\ell)}$, where 
$\mcC^{(\ell)}:=\mcC\otimes C_0^{\infty}(\real^{\nu})^{\otimes_s \ell}$.  
Since  $(H^{\ex}, P^{\ex})$ as well as $(P^{(\ell)},H^{(\ell)})$, for $\ell\in\nat$,  form commuting families of self-adjoint operators,  we can introduce their joint spectral resolutions $E^{\ex}(\,\cdot\,)$ and  $E^{(\ell)}(\,\cdot\,)$.
We use extended Lee-Low-Pines transformations to perform fiber decompositions
of $H^{\ex}$ and $H^{(\ell)}$ w.r.t. the total momentum.
They have the form  
\beqa
I_{\LLP}^{\ex}:=(F\otimes 1)\circ \Ga^{\ex}(\e^{\i k\cdot y})=I_{\LLP}\oplus\Bigl(\bigoplus_{\ell=1}^{\infty} I_{\LLP}^{(\ell)} \Bigr),
\eeqa
where  $F$ is the Fourier transform in the electron position variable, $\Ga^{\ex}( \e^{\i k\cdot y})$ is defined as explained in Section~1.2 of \cite{MR12}
and  $I_{\LLP}^{(\ell)}:=(I_{\LLP}^{\ex})_{|\hil\otimes \Fock^{\ell}}$.
There holds
\begin{align}\label{FirstExFib}
H^{\ex} =I_{\LLP}^{\ex *}\Bigl( \int_{\real^{\nu}}^{\oplus} d\xi\,  H^{\ex}(\xi) \Bigr)I_{\LLP}^{\ex},\qquad 
H^{(\ell)}  = I_{\LLP}^{(\ell) *}\Bigl( \int_{\real^{\nu}}^{\oplus}  d\xi\, H^{(\ell)}(\xi) \Bigr)I_{\LLP}^{(\ell)}.
\end{align}
The fiber Hamiltonians $H^\ex(\xi)$ are essentially self-adjoint on $ \mcC^{\ex}$  and  have the  form
\beqa
H^{\ex}(\xi)=\Omeg(\xi-\dGa^{\ex}(k))+\dGa^{\ex}(\om)+\phi(G)\otimes 1, \label{extended-fiber-hamiltonian}
\eeqa
where $\dGa^{\ex}(\,\cdot\,)$ is defined in  Appendix~\ref{Extended-appendix}. The extended fiber Hamiltonians $H^\ex(\xi)$ 
can be decomposed just as for $H^\ex$, cf.~\eqref{Hex-Def-And-Decomp}, and we get as expected
\begin{equation}
H^{\ex}(\xi)= H(\xi)\oplus\Bigl(\bigoplus_{\ell=1}^{\infty}H^{(\ell)}(\xi)\Bigr).\label{Hamiltonian-decomp}
\end{equation}
Since there is no interaction in the second tensor component of $H^{(\ell)}(\xi)$,
which is simply a multiplication operator, we can decompose further into a direct integral over momenta from $\real^{\ell\nu}$:
\begin{align}\label{MomentumFib}
H^{(\ell)}(\xi)&= \int_{ \real^{\ell\nu} }^\oplus dk\, H^{(\ell)}(\xi,\un k),\\
\label{DoubleFib}
H^{(\ell)}(\xi;\un k)&= H(\xi-\sum_{j=1}^\ell k_j)+(\sum_{j=1}^\ell\om(k_j))1.
\end{align}
In our investigation we will often make use of the following simple fact:
\bel\label{expansion-truncation} Let $\chi\colon\real\to\real$ be a bounded Borel function, with essential support
in the set $(-\infty, \Sigma_0^{(n)}(\xi))$. Then
\beqa
\chi(H^{\ex}(\xi))=\chi(H(\xi))\oplus \Bigl( \bigoplus_{\ell=1}^{n-1}\chi(H^{(\ell)}(\xi))\Bigr).  
\eeqa 
\eel  
\begin{proof} Let $\ell\geq n$. We recall that
\beqa
\Siz^{(\ell)}(\xi)=\inf_{\un k\in\real^{\ell\nu}}\Bigl(    \Siz(\xi-\sum_{j=1}^\ell k_j)+\sum_{j=1}^\ell\om(k_j)\Bigr).
\eeqa
Consequently, 
\beqa
H^{(\ell)}(\xi)=\int^{\oplus}_{\real^{\ell\nu}}dk\,  \bigl(H(\xi-\sum_{j=1}^\ell k_j)+(\sum_{j=1}^\ell\om(k_j))1 \bigr) \geq \Sigma_0^{(\ell)}(\xi)1.
\eeqa
Since $\Sigma_0^{(\ell)}(\xi)\geq \Sigma_{0}^{(n)}(\xi)$, and $\chi$ is supported below $\Sigma_0^{(n)}(\xi)$, only the first $n-1$
terms of the expansion
\begin{equation}
\chi(H^{\ex}(\xi))=\chi(H(\xi))\oplus\Bigl(\bigoplus_{\ell=1}^{\infty}\chi(H^{(\ell)}(\xi))\Bigr)
\end{equation}
are non-zero. \end{proof}

\subsection{Structure of the spectrum}\label{isolated-spectrum}

To continue our discussion of the spectrum of $H$ we need more restrictive assumptions.
Following \cite{MR12}, we state:
\beco \bf (Spectral Theory). \rm \label{Condition-two} We impose:  
\begin{enumerate}[label = \textbf{\textup{(ST\arabic*)}}, ref =
  \textup{(ST\arabic*)},leftmargin = *]
\item $\Omeg$ and $\om$ are  real analytic functions. 
\item\label{ST-DistDer} $G$ admits $2$ distributional derivatives with $\pa_k^{\al}G\in L^2_{\mathrm{loc}}(\real^{\nu}\backslash \{0\})$, 
for all $1\leq |\al|\leq 2$. 
\item For all orthogonal matrices $O\in \mrO(\nu)$ and all $k\in\real^{\nu}$ we have $\om(Ok)=\om(k)$, $\Omeg(Ok)=\Omeg(k)$,
and $G(Ok)=G(k)$ almost everywhere. 
\item\label{ST-HigherOrder} $\sup_{k\in\real^{\nu}}|\pa^{\al}_k\om(k)|<\infty$ for all $|\al|\geq 1$ and  
$|\pa_{k}^{\be}\Omeg(k)|\leq C_{\be}\lan k \ran^{s_{\Omeg}-|\be|}$  for $|\be|\geq 2$. $s_{\Omeg}\in [0,2]$ appeared in Condition~\ref{Condition-one}. 
\end{enumerate}

 We note that \ref{ST-DistDer} coincides with the corresponding condition from  \cite{MR12} for $n_0=2$.  
 \ref{ST-HigherOrder} is stronger than in  \cite{MR12}.
\eeco

Making use of Kato's analytic perturbation theory \cite{Ka} we obtain a description of the
isolated part of the spectrum (cf. (\ref{sigma-essential}) above):
\beqa
\Sigma_{\iso}=\bigset{(\xi,E)\in \Sigma }{ E<\Sigma_{0}^{(1)}(\xi) }.
\eeqa
This spectrum consists of analytic mass shells and level crossings. The set of level crossings
is defined as
\beqa
\mcX:=\bigset{(\xi,E)\in \Sigma_{\iso} }{ \forall n\in\nat:\,
  \Sigma_{\iso}\cap B_{1/n}((\xi,E)) 
\textup{ is not a graph }}. 
\eeqa
The connected components of $\mcX$ are $S^{\nu-1}$-spheres. They have the
form $\pa B(0; R)\times \{ E \}$, or, in the degenerate case, $\{0\}\times \{E\}$.  They can accumulate either at infinity
or at the bottom of the essential spectrum. The level crossings  are connected in $\Sigma_{\iso}$ by shells which are real-analytic manifolds.
Each shell is a pair $(\mcA, S)$,
where $\mcA=\set{\xi\in\real^{\nu} }{ r<|\xi|<R }$, $0\leq r<R$,  is an open annulus or an open ball centred at zero. The
function $S\colon \mcA\to \real$ is real analytic and rotation invariant.

The structure of the continuous spectrum in $\mcE^{(1)}$, cf. (\ref{mcE-one}), was studied in \cite{MR12} with the help of
Mourre theory. As these results are very relevant for the present investigation  we summarize them 
here. For any $\xi\in\real^{\nu}$ the conjugate operator has the form $A_{\xi}=\dGa(a_{\xi})$, where
\beqa
a_{\xi}=\h\bigl\{v_{\xi}\cdot \i\nabla_k+\i\nabla_{k}\cdot v_{\xi} \bigr\},
\eeqa
and $v_{\xi}\in C_0^{\infty}(\real^{\nu}\backslash\{0\};\real^{\nu})$ is a suitable vector field constructed in \cite{MR12}.						
It is easily seen that  $a_{\xi}$  is essentially self-adjoint on $C_0^{\infty}(\real^{\nu})$ and $A_{\xi}$ is essentially self-adjoint
on $\mcC$. In \cite{MR12} one can find a construction  of the threshold sets $\mcT^{(1)}(\xi)\subset \real$, $\xi\in\real^\nu$,
which carry information about the structure of the isolated spectrum,  and exceptional sets 
\beqa
\Exc(\xi)=(0,\om(0))+\Sigma_{\iso}(\xi),\quad \xi\in\real^{\nu}, \label{exceptional-set}
\eeqa
 which account for a possible singularity of the coupling function $G$ at zero. (We recall that in \cite{MR12}, formula (1.35),
$\Exc(\xi)$ was defined to be empty for $G$ regular at zero. Here it is always given by (\ref{exceptional-set})).
The main result of \cite{MR12} can be summarized as follows:
\bet\cite{MR12}\label{Mourre-estimate-theorem} Assume Conditions~\ref{Condition-one} and~\ref{Condition-two}. Let $\xi\in\real^{\nu}$ unless stated otherwise. Then the following  properties hold true:
\begin{enumerate}[label = \textup{(\alph*)}, ref =\textup{(\alph*)},leftmargin=*]
\item  The sets $\mcE^{(1)}(\xi)\cap \mcT^{(1)}(\xi)$
and $\mcE^{(1)}(\xi)\cap \Exc(\xi)$ are locally finite with possible accumulation points only at $\Sigma_0^{(2)}(\xi)$.
\item All eigenvalues in $\si_{\pp}(H(\xi))\cap\mcE^{(1)}(\xi)\backslash(\mcT^{(1)}(\xi)\cup\Exc(\xi))$ have finite multiplicity.
\item The set $\si_{\pp}(H(\xi))\cap\mcE^{(1)}(\xi)$ is at most countable, with accumulation points at most
in $\mcT^{(1)}(\xi)\cup\Exc(\xi)\cup\{\Sigma_0^{(2)}(\xi)\}$.
\item Let $(\xi_0,\la_0)\in \mcE^{(1)}$ be s.t.  
$\la_0\in  \mcE^{(1)}(\xi_0) \backslash (\mcT^{(1)}(\xi_0)\cup\Exc(\xi_0)\cup \si_{\pp}(H(\xi_0))$.
Then there exist a neighbourhood $N_0$ of $\xi_0$, a neighbourhood $\mcJz$ of $\la_0$, and a constant $\cm>0$ 
s.t. for any $\xi\in N_0$:
\begin{align}
& \mathbf{1}_{\mcJz}(H(\xi))\i [H(\xi),  A_{\xi_0}]\mathbf{1}_{\mcJz}(H(\xi))\geq \cm \mathbf{1}_{\mcJz}(H(\xi)), \label{Mourre-est} \\
& \mathbf{1}_{\mcJz}(H^{(1)}(\xi))\i [H^{(1)}(\xi),  1\otimes a_{\xi_0}]\mathbf{1}_{\mcJz}(H^{(1)}(\xi))
\geq \cm \mathbf{1}_{\mcJz}(H^{(1)}(\xi)). \label{Mourre-est-one}
\end{align}
\item The fiber Hamiltonians have no singular continuous spectrum below the two-boson threshold:
\beqa
 \si_{\sico}(H(\xi))\cap (-\infty, \Sigma_0^{(2)}(\xi))=\emptyset.
\eeqa
\end{enumerate}
\eet

\subsection{Results}\label{results-subsection}

We begin by introducing some notation. First of all, the space of
bound states $\hil_\bnd$ of the system is the closure of the span of all states of
the form $I_{\LLP}^*\int^\oplus \Psi_\xi d\xi$, where $\real^\nu \ni \xi \to \Psi_\xi$ is
compactly supported, measurable and $\Psi_\xi$ is an eigenvector for $H(\xi)$ for a.e. $\xi$.
Expressed concisely in terms of the joint spectral resolution $E$ for the vector of commuting operators $(P,H)$ this amounts to
\beqa
\hil_\bnd  = E(\Sigma_\pp)\hil. 
\eeqa
Incoming scattering states prepared at $t= -\infty$, as well as
outgoing scattering states at $t \to +\infty$, consist of a
superposition of interacting dressed electrons and a collection of free bosons. 
That is, the incoming and outgoing spaces are
\beqa
\hil_{\pm} = \hil_\bnd\otimes \Fock. 
\eeqa
The asymptotic dynamics  on incoming and outgoing spaces are generated
by the restriction of $\Hext$ to $\hil_{\pm}$. In the light of our discussion in the preceding two subsections, $\Hext_{|_{\hil_{\pm}}}$ is a direct
sum of operators of the form
\beqa
I_{\LLP}^{(\ell) *}\Bigl( \int_{\real^{\nu}}^{\oplus}  d\xi\, \int_{ \real^{\ell\nu} }^\oplus dk\, S(\xi-k_1-\cdots -k_{\ell}) + 
\omega(k_1)+\cdots +\omega(k_{\ell})  \Bigr)I_{\LLP}^{(\ell)},
\eeqa
where $\ell\in\nat_0$ and $(\mcA, S)$ are  shells in $\Sigma_{\iso}$. 
Moreover, it is an easy consequence of the HVZ
theorem, cf.~\cite[Theorem~2.1]{MRMP}, that $(P,H)$ and $(P^\ex, \Hext)_{|\hil_{\pm} }$ have
identical energy-momentum spectra.

Let us recall that the asymptotic creation  operators of bosons  are usually defined as follows:
\beqa
a_\pm^*(h)\Psi:=\lim_{t\to\pm\infty}\e^{\i tH}a^*(\e^{-\i \om t}h)\e^{-\i tH}\Psi, \label{asymptotic-fields}
\eeqa 
where $h\in\mfh$ and $\Psi$  belongs to the dense domain $\mcD$  of vectors of bounded energy 
(i.e. $\mcD:=\bigcup_{K\subset \real^{\nu+1}}\Ran\, E(K)$, where the union extends over all compact sets).
It is well known \cite{HK68}  and easy to see that the limit
exists in the case of the massive Nelson model   (i.e. $G\in S(\real^{\nu})$ and $\om(k)=\sqrt{k^2+m^2}$, where $m>0$).
As a consequence, in this case there exist mappings $\widetilde{\Ome}^\pm$ 
defined on $\Psi'\in \mcD\otimes \Ga_{\fin}(\mfh)$ by
\beqa
\widetilde{\Ome}^\pm\Psi'=\lim_{t\to\pm\infty} \e^{\i tH}\cGa(1,1)^*\e^{-\i tH^{\ex}}\Psi', \label{wave-operator-intro-def}
\eeqa
where the scattering identification map $\cGa(1,1)^*$ is defined in (\ref{dGa-check-intro}) and in Appendix~\ref{Extended-appendix}.
The restrictions of $\widetilde{\Ome}^\pm$ to $\hil_{\pm}$, denoted by $\Ome^{\pm}$, are usually called the (conventional) wave operators.
They were introduced first in \cite{HuSp2}. The associated
(conventional) scattering operator $S\colon \hil_-\to \hil_+$ is
then given by $S = (\Ome^-)^*\Ome^+$.
Observe that restricted to the subspace $ \hil_\bnd\otimes \complex\subset\hil_\pm$,
the wave operators trivially exist and act as injections
\begin{equation}
\forall \Psi\in\hil_\bnd:\qquad \Omega^\pm(\Psi\otimes\vac) = \Psi.
\end{equation}

To serve as an acceptable wave operator, $\Ome^\pm$ should be isometric.  
At small coupling strength such a result  seems to  be within reach of  methods
present in the literature \cite{FGSch3,AMZ}. However, at arbitrary couplings, in the possible presence of eigenvalues
embedded in the continuous spectrum, nothing is known to date about this problem. Not to speak of the problem
of asymptotic completeness in the massive Nelson model, which is the
question of  isometry of the adjoints of the wave operators.

We note that in the case of the polaron model 
(i.e. $G=\ti G(k)/|k|$, $\ti G\in S(\real)$ and $\om(k)=m$, where $m>0$),
which is also covered by our assumptions, problems start already 
at the level of  existence of the wave operators. 
Since the boson dispersion relation gives only a phase factor, it might even
seem that the wave operators (\ref{wave-operator-intro-def}) do not exist!

It turns out that the situation is much better than outlined above, at
least in the energy-momentum regime below the two-boson threshold i.e.
in the region
$\mcR:=\set{(\xi,\la)\in\real^{\nu+1} }{\la<\Si^{(2)}(\xi)}$. Indeed, in
this region our main result  resolves all the problems
mentioned in the two paragraphs above:
\bet\label{main-theorem} Assume Conditions~\ref{Condition-one} and~\ref{Condition-two}. The wave operators $\OmeR^+\colon E^\ex(\mcR)
\hil_{\pm}\to \hil$ exists in the sense of the strong limits
\beqa
\OmeR^\pm:=\slim_{t\to \pm\infty} \e^{\i tH}\cGa(1,1)^*\e^{-\i tH^{\ex}}, \label{intro-wave-operator}
\eeqa
where $\cGa(1,1)$ is defined in Appendix~\ref{Extended-appendix}.
The operators $\OmeR^\pm$ are unitary as maps from  $E^{\ex}(\mcR)\hil_{\pm}$ to
 $E(\mcR)\hil$. More precisely: 
\begin{equation}
\OmeR^{\pm*}\OmeR^\pm = E^{\ex}(\mcR)_{|\hil_{\pm} }\qquad \textup{and} \qquad
\OmeR^\pm\OmeR^{\pm*} = E(\mcR).
\end{equation}
Finally, the scattering operator $S_{\mcR} = (\OmeR^{-})^*\OmeR^+\colon
E^\ex(\mcR)\hil_-\to E^\ex(\mcR)\hil_+$ is unitary. 
\eet

In the energy-momentum regime $\mcR$, scattering only happens between bound states associated to the isolated part $\Sigma_\iso$ of $\Sigma_\pp$. For this reason a special role is played by the bound states pertaining to isolated mass shells, for which we use the notation
\begin{equation}\label{Hiso}
\hil_\iso = E(\Sigma_\iso)\hil.
\end{equation}

Let us introduce the terminology that a state
$\Psi \in \hil_\bnd$ and a smearing function $h\in \mfh$
are $\mcR$-compatible if there  exists a Borel set $\mcS\subset \real^\nu\times\real$
such that $\Psi \in E(\mcS)\hil_\bnd$ and
$\set{(\xi+k,E+\omega(k))}{(\xi,E)\in \mcS, k\in \supp\, h}\subset \mcR$.
By $\supp\, h$, we understand $h$'s essential support.
Note that by energy-momentum considerations we can always choose $\mcS \subset \Sigma_\iso$, such that in fact $\Psi\in\hil_\iso$.

With the terminology just introduced,  $E^\ex(\mcR)\hil_{\pm}$ is the direct sum of states of the
form $\Psi\otimes\vac$, with $\Psi\in E(\mcR)\hil_\bnd$, 
and states from the closure of the span of
states of the form $\Psi\otimes a^*(h)\vac$, where $\Psi\in\hil_\iso$ and  
$h$ are $\mcR$-compatible. See Lemma~\ref{energetic-considerations}) for a proof.

For any $\Psi\in \hil_\iso$ and $h\in \mfh$ which are $\mcR$-compatible
we define the corresponding scattering state as follows 
\beq
a_+^*(h)\Psi:=\OmeR^+(\Psi\otimes a^*(h)\vac). 
\eeq
Theorem~\ref{main-theorem} has the following corollary:
\bec\label{HR-corollary}  Let $a_+^*(h)\Psi$, $a_+^*(h')\Psi'$ be scattering states and $\Psi''\in E(\mcR)\hil_\bnd$.
There hold the following properties:
\begin{enumerate}[label=\textup{(\alph*)},ref=\textup{(\alph*)},leftmargin = *]
\item\label{HR-Cor-a} Tensor product structure:
\beqa
\lan a_+^*(h)\Psi,a_+^*(h')\Psi'\ran=\lan h, h'\ran\lan \Psi,\Psi'\ran
\quad\textup{and}\quad
\lan a_+^*(h)\Psi,\Psi''\ran=0.
\eeqa
\item\label{HR-Cor-b} Asymptotic completeness: 
\beqa
E(\mcR)\hil=\ov{\Span\bigset{ a_+^*(h)\Psi, \Psi'' }{\Psi,  h \textup{ are } \mcR\textup{--compatible}, \Psi''\in E(\mcR)\hil_\bnd}}.
\eeqa
\end{enumerate}
\eec

Note that for the particular case of the polaron model, the notion of
$\Psi$ and $h$ being $\mcR$-compatible is completely trivial. Here
$\mcR =  \set{(\xi,E)\in \real^{\nu+1} }{E<\Si + 2m}$, where $m$ is the phonon
mass, cf.~\ref{MC-Boson-Mass}, and $\Si$ is the bottom of the spectrum
of $H$. That is, $\mcR$ is just  a half-space. Being $\mcR$-compatible
thus reduces to $\Psi\in E(\real^\nu\times(-\infty,\Sigma_0+m))\hil_\bnd
= \hil_\iso$, with no condition on $h$. Hence, in this the polaron case we have:
\beqa
E(\mcR)\hil=\ov{\Span\bigset{ a_+^*(h)\Psi, \Psi'' }{\Psi\in\hil_\iso,  h\in\mfh \textup{ and } \Psi''\in E(\mcR)\hil_\bnd}}.
\eeqa

\section{Heisenberg derivatives}\label{Heisenberg}
\setcounter{equation}{0}

As usual in investigations of the problem of asymptotic completeness, we are interested
in the  existence of  asymptotic observables, which are strong limits as $t\to\infty$  of time dependent families of observables 
 of the form
\beqa
\real\ni t\to \e^{\i tH(\xi)}\Phi(t)\e^{-\i tH(\xi)},
\eeqa
where the propagation observable $\real\ni t\to \Phi(t)\in B(\Fock)$ is uniformly bounded in time.
Since we are going to proceed via Cook's method, we are interested in the Heisenberg derivatives of propagation
observables, defined a priori in the sense of forms on $D(H(\xi))$ as
\beqa
\bD\Phi(t)=\pa_t\Phi(t) +\i[H(\xi),\Phi(t)].
\eeqa
In Propositions~\ref{a-dGa-Heisenberg} and \ref{second-technical} below we will express such derivatives by Heisenberg derivatives of some
propagation observables $\real\ni t\to \Phi^{(1)}(t)\in B(\Fock\otimes\Fock^{(1)})$, given by
\beqa
\bD^{(1)} \Phi^{(1)}(t)=\pa_t \Phi^{(1)}(t)+\i[H^{(1)}(\xi), \Phi^{(1)}(t)  ].
\eeqa
Before  we  state and prove these propositions, which  provide the  technical basis for  our investigation,  we need the following definition:
\bed\label{j-definition-a} Let $j_{0},j_{\infty}\in C^{\infty}(\real)$ be s.t. $j_0',j_{\infty}'\in C_0^{\infty}(\real)$, $0\leq j_{0},j_{\infty}\leq 1$,   $j_0=1$ in a neighbourhood of zero.   We set $\jt_{0}:=j_{0}(a/t)$, $\jt_{\infty}:=j_{\infty}(a/t)$, and $\jt:=(j_0^t,j_{\infty}^t)$ as a map $\mfh\to\mfh\oplus\mfh$
defined by $\jt h:=(j_0^th,j_{\infty}^th)$.
\eed
\begin{remark}
In Section~\ref{Heisenberg} and in Appendices~\ref{commutator-bounds-appendix}-\ref{appendix-second-technical} $a:=\h\bigl\{v\cdot \i\nabla_k+\i\nabla_{k}\cdot v\bigr\}$, where 
$v\in  C_0^{\infty}(\real^{\nu}\backslash\{0\};\real^{\nu})$ is an arbitrary vector field. Unless stated otherwise, in the remaining 
part of the paper $a:=a_{\xi_0}=\h\bigl\{v_{\xi_0}\cdot \i\nabla_k+\i\nabla_{k}\cdot v_{\xi_0}\bigr\}$ is the observable appearing
in Theorem~\ref{Mourre-estimate-theorem}, associated with some neighbourhoods $N_0$ and $\mcJz$.
\end{remark}
\bep\label{a-dGa-Heisenberg} Let $\xi\in\real^{\nu}$ and $\chi\in C_0^{\infty}(\real)_{\real}$ be supported in $(-\infty, \Sigma_0^{(2)}(\xi))$.
Let $q\in C^{\infty}(\real)_{\real}$ be  s.t.  $0\not\in\supp\, q$
 and  $q'\in C_0^{\infty}(\real)$ (in particular $q$ is bounded). 
Let $j_0,j_{\infty}$ be as specified in Definition~\ref{j-definition-a} and s.t. $j_0^2+j_{\infty}^2=1$.   
Then
\begin{align}
\bD(\chi\dGa(q^t)\chi) & = \Ga(j_0^t)\bD(\chi\dGa(q^t)\chi)\Ga(j_0^t)\non\\
& \quad + \cGao(j^t)^*\chi^{(1)}\bD^{(1)}(1\otimes q^t)\chi^{(1)}\cGao(j^t)+O(t^{-2}),
\label{first-Heisenberg-a}
\end{align} 
where we set $\chi:=\chi(H(\xi))$, $\chi^{(\ell)}:=\chi(H^{(\ell)}(\xi))$ and $q^t:=q(a/t)$.
Moreover, for $\supp\, j_0\cap \supp\, q=\emptyset$ we have\footnote{One can weaken this assumption to  $\supp\, j_0\cap \supp\, q'=\emptyset$ at a cost of additional complications in Appendix~\ref{First-proposition-appendix}. This is, however, not needed in the following. }
\beqa
\Ga(j_0^t)\chi\bD\dGa(q^t)\chi\Ga(j_0^t)=O(t^{-2}).
\eeqa
\eep
\begin{proof}  We write $j:=j^t$, $q:=q^t$. The Heisenberg derivative of the asymptotic observable $\Phi(t):=\chi\dGa(q)\chi$ is given by
\beq
\mathbf{D}\Phi(t)=\chi\dGa(\pa_t q)\chi+\i\chi[H(\xi), \dGa(q)]\chi. \label{Heisenberg-derivative}
\eeq
We consider the first term  on the r.h.s.  above:
\begin{equation}
\chi\dGa(\pa_t q)\chi=\cGa(j)^*\chi^{\ex} \dGa^{\ex}(\pa_t q)\chi^{\ex}\cGa(j)+O(t^{-2}), 
\end{equation}
where we applied Proposition~\ref{derivative-term} and Lemma~\ref{concrete-admissible-obs}, and set $\chi^{\ex}:=\chi(H^{\ex}(\xi))$.  We use the decomposition~(\ref{Hamiltonian-decomp}) 
of $H^{\ex}(\xi)$ 
and, by Lemma~\ref{expansion-truncation}, it suffices to consider the terms $\ell=0$ and $\ell=1$. The $\ell=0$ term has the following form
\beqa
\cGaz(j)^*\chi^{(0)}  \dGa^{(0)}(\pa_t q) \chi^{(0)}\cGaz(j)=\Ga(\jj_0)\chi\dGa(\pa_t q)\chi\Ga(\jj_0).
\eeqa
If $j_0$ is supported outside of the support of $q$, this contribution is of order $O(t^{-2})$ by Proposition~\ref{l=0}. 
Otherwise it contributes to the first term on the r.h.s. of (\ref{first-Heisenberg-a}).

The $\ell=1$ term has the form
\beqa
\cGao(j)^*\chi^{(1)}\bigl(\dGa(\pa_t q)\otimes 1+1\otimes (\pa_t q)  \bigr)\chi^{(1)}\cGao(j).
\eeqa
By Corollary~\ref{resolvents}, we obtain
\beqa
\cGao(j)^*\chi^{(1)}\big(\dGa(\pa_t q)\otimes 1 \big)\chi^{(1)}\cGao(j)=O(t^{-2}).
\eeqa
So we are left with
\beq
\cGao(j)^*\chi^{(1)}\big( 1\otimes \pa_t q  \big)\chi^{(1)} \cGao(j), \label{time-derivative-contribution}
\eeq
which contributes to the expression on the r.h.s. of (\ref{first-Heisenberg-a}).

Now we proceed to the second term on the r.h.s. of (\ref{Heisenberg-derivative}). From Proposition~\ref{gamma-check}, we obtain
\beqa
\chi[H(\xi),\dGa(q)]\chi 
=\cGa(j)^*\chi^{\ex}[H^{\ex}(\xi), \dGa^{\ex}(q)]\chi^{\ex}\cGa(j)+O(t^{-2}).
\eeqa
Making use of the decomposition~(\ref{Hamiltonian-decomp}), we get 
\beqa
\cGa(j)^*\chi^{\ex}[H^{\ex}(\xi), \dGa^{\ex}(q)]\chi^{\ex}\cGa(j)
=\cGa(j)^*\Bigl(\bigoplus_{\ell=0}^{\infty}\chi^{(\ell)}[H^{(\ell)}(\xi), \dGa^{(\ell)}(q)]\chi^{(\ell)}\Bigr) \cGa(j).  
\eeqa
By Lemma~\ref{expansion-truncation}, it suffices to consider $\ell=0$ and $\ell=1$ terms:
The $\ell=0$ contribution is the following:
\beqa
\cGaz(j)^*\chi^{(0)} [H^{(0)}(\xi),\dGa^{(0)}(q)]\chi^{(0)} \cGaz(j)
=\Ga(\jj_0)\chi [H(\xi),\dGa(q)]\chi \Ga(\jj_0).\label{zero-term}
\eeqa
If $j_0$ is supported outside of the support of $q$, 
this contribution is of order $O(t^{-2})$ by Proposition~\ref{l=0}. Otherwise it contributes to the first term on the r.h.s. of (\ref{first-Heisenberg-a}).

Let us now consider the contribution with $\ell=1$:
\beqa
\cGao(j)^*\chi^{(1)}[H^{(1)}(\xi), \dGa^{(1)}(q)]\chi^{(1)}\cGao(j).
\eeqa
We recall that $\dGa^{(1)}(q)=\dGa(q)\otimes 1+1\otimes q$
and obtain 
\begin{align}
& \cGao(j)^*\chi^{(1)}[H^{(1)}(\xi), \dGa^{(1)}(q)]\chi^{(1)}\cGao(j)\non\\
& \quad =\cGao(j)^*\chi^{(1)}[H^{(1)}(\xi), \dGa(q)\otimes 1]\chi^{(1)}\cGao(j)
 +\cGao(j)^*\chi^{(1)}[H^{(1)}(\xi), 1\otimes q]\chi^{(1)}\cGao(j).
\end{align}
The first term on the r.h.s. above is of order $O(t^{-2})$ by Lemma~\ref{decay-of-dGa-term}. The second term contributes to the expression from the 
statement of the proposition.

Thus, together with  (\ref{time-derivative-contribution}), we get
\begin{align}
\mathbf{D}\Phi(t)& = \Ga(j_0)\mathbf{D}\Phi(t)\Ga(j_0)\non\\
& \quad +\cGao(j)^*\chi^{(1)}\bigl( 1\otimes \pa_t q+\i[H^{(1)}(\xi), 1\otimes q]  \bigr)\chi^{(1)}\cGao(j)+O(t^{-2}),
\end{align}
and the first term on the r.h.s. contributes to $O(t^{-2})$ for $j_0$ supported outside of the support of $q$. This concludes the proof. 
\end{proof}

\bep\label{second-technical} Let $\xi\in\real^{\nu}$ and $\chi\in C_0^{\infty}(\real)_{\real}$ be supported in $(-\infty, \Sigma_0^{(2)}(\xi))$.  
Let $ q\in C^{\infty}(\real)$ be s.t. $q'\in C_0^{\infty}(\real)$, $0 \leq q\leq 1$, $q=1$ on a neighbourhood $\De $ of zero. 
Let $j_0, j_{\infty}$ be as specified in Definition~\ref{j-definition-a}, s.t. $j_0^2+j_{\infty}^2=1$ and $\jj_0$ 
is supported in $\De$.  Then
\beq
\chi\mathbf{D}\Ga(q^t)\chi=\cGao(j^t)^*\chi^{(1)}(\Ga(q^t)\otimes 1)\bD^{(1)}(1\otimes q^t)\chi^{(1)}\cGao(j^t)+O(t^{-2}),
\label{second-technical-formula}
\eeq
where we set $\chi:=\chi(H(\xi))$, $\chi^{(\ell)}:=\chi(H^{(\ell)}(\xi))$ and $q^t:=q(a/t)$. Consequently,
\beqa
\chi\mathbf{D}\Ga(q^t)\chi=\fr{1}{t}\cGao(j^t)^*\chi^{(1)}C_t(1\otimes (q')^t)\chi^{(1)}\cGao(j^t)+O(t^{-2}), \label{Heisenberg-derivative-expression}
\eeqa
where $\{C_t\}_{t\in\real}$ is a family of bounded operators on $\Fock\otimes\Fock^{(1)}$ which satisfies
\begin{equation}
  C_t(N+1)=O(1) 
  \qquad \textup{and} \qquad
  [C_t, 1\otimes p^t]=O(t^{-1}), \label{C-second-property}
\end{equation}
for any $p\in C^{\infty}(\real)_{\real}$ with $p'\in C_0^{\infty}(\real)$.
\eep
\begin{proof} We set $q:=q^t$, $j:=j^t$ and  compute the Heisenberg derivative:
\beqa
\chi\mathbf{D}\Ga(q)\chi&=&\chi(\dGa(\qd,\pa_t \qd)+\i[H(\xi), \Ga(\qd)])\chi.
\eeqa
Making use of Proposition~\ref{Ga-Heisenberg}, we obtain
\begin{equation}
\chi \dGa(\qd,\pa_t \qd) \chi = \cGa(j)^*\chi^{\ex}\dGa^{\ex}(\qd, \pa_t \qd) \chi^{\ex}\cGa(j)+O(t^{-2}),
\end{equation}
where we set $\chi^{\ex}:=\chi(H^{\ex}(\xi))$ and $\dGa^{\ex}(\,\cdot\,,\,\cdot\,)$ is defined by formula~(\ref{dGa-q-p-def}).
Inserting decomposition~(\ref{Hamiltonian-decomp}) of $H^{\ex}(\xi)$, we get
\beqa
\cGa(j)^*\chi^{\ex}\dGa^{\ex}(\qd, \pa_t \qd) \chi^{\ex}\cGa(j)
=\cGa(j)^*\Bigl(\bigoplus_{\ell=0}^{\infty}\chi^{(\ell)}  \dGa^{(\ell)}(\qd, \pa_t \qd)  \chi^{(\ell)}\Bigr)\cGa(j).  
\eeqa
By Lemma~\ref{expansion-truncation}, it suffices to consider $\ell=0$ and $\ell=1$ terms. For $\ell=0$, we get 
\beqa
\Ga(j_0)\chi \dGa(\qd, \pa_t \qd ) \chi\Ga(j_0)=O(t^{-2}) 
\eeqa
by Proposition~\ref{l=0-Ga-lemma}.  The  $\ell=1$ term is given by
\beqa
\cGao(j)^*\chi^{(1)} \bigl( \dGa(\qd, \pa_t \qd )\otimes \qd+\Ga(\qd) \otimes \pa_t \qd   \bigr) \chi^{(1)} \cGao(j).  
\eeqa
 We note that, by Corollary~\ref{resolvents}, 
\beqa
\cGao(j)^*\chi^{(1)}\bigl(\dGa(\qd,\pa_t \qd)\otimes \qd \bigr)\chi^{(1)}\cGao(j)=O(t^{-2}).
\eeqa
 So we are left with
\begin{equation}
\cGao(j)^*\chi^{(1)}\bigl(\Ga(\qd)   \otimes \pa_t \qd  \bigr)\chi^{(1)}\cGao(j),
\end{equation}
which contributes to the r.h.s. of (\ref{second-technical-formula}). Next, we choose $\ti\chi\in C_0^{\infty}(\real)_{\real}$ s.t. $\chi\ti\chi=\chi$ and make
 use of Lemma~\ref{Helffer-second-component-lemma} to write
\begin{align}
& \cGao(j)^*\chi^{(1)}\bigl(\Ga(\qd)   \otimes \pa_t \qd  \bigr)\chi^{(1)}\cGao(j)\non\\
& \qquad =\cGao(j)^*\chi^{(1)}(\Ga(\qd)   \otimes  \ti q )\ti \chi^{(1)}(1\otimes  \qd') \chi^{(1)}\cGao(j)+O(t^{-2}),
\end{align}
where $\ti q=-(a/t)f(a/t)$, $f\in C_0^{\infty}(\real)$ is equal to one on the support of $q'$ and vanishes outside of a slightly larger set. 
The operator $C_{1,t}:=(\Ga(\qd)   \otimes \ti q)\ti \chi^{(1)}$ is the first contribution to $C_t$ appearing in (\ref{Heisenberg-derivative-expression}).
It is obvious that $C_{1,t}$ satisfies the first property in~(\ref{C-second-property}), and the second property in~(\ref{C-second-property}) follows from Lemma~\ref{Helffer-second-component-lemma}.

Let us now consider the second contribution to the Heisenberg derivative. From Proposition~\ref{Ga-Heisenberg} we obtain
\begin{equation}
\chi [H(\xi), \Ga(q)]\chi  = \cGa(j)^*\chi^{\ex}[H^{\ex}(\xi), \Ga^{\ex}(\qd)]\chi^{\ex}\cGa(j)+O(t^{-2}).
\end{equation}
By inserting the decomposition~(\ref{Hamiltonian-decomp}), we get:
\beqa
\cGa(j)^*\chi^{\ex}[H^{\ex}(\xi), \Ga^{\ex}(\qd)]\chi^{\ex}\cGa(j)
=\cGa(j)^*\Bigl(\bigoplus_{\ell=0}^{\infty}\chi^{(\ell)}
[H^{(\ell)}(\xi), \Ga^{(\ell)}(\qd)]\chi^{(\ell)} \Bigr)\cGa(j).  
\eeqa
As before, it is enough to consider $\ell=0$ and $\ell=1$ terms. As for the $\ell=0$ term, 
\beqa
\Ga(j_0)\chi [H(\xi), \Ga(\qd)]\chi\Ga(j_0)=O(t^{-2}),
\eeqa
by Proposition~\ref{l=0-Ga-lemma}. The $\ell=1$ term is given by
\begin{align}
& \cGao(j)^*\chi^{(1)} [H^{(1)}(\xi), \Ga(\qd)\otimes \qd]\chi^{(1)} \cGao(j)\non\\
& \qquad =\cGao(j)^*\chi^{(1)} [H^{(1)}(\xi),\Ga(\qd)\otimes 1](1\otimes \qd)\chi^{(1)} \cGao(j)\non\\
& \qquad \quad +\cGao(j)^*\chi^{(1)}(\Ga(\qd)\otimes  1) [ H^{(1)}(\xi), 1\otimes \qd ] \chi^{(1)}\cGao(j).
\end{align}
The first term on the r.h.s. above is $O(t^{-2})$ by Proposition~\ref{new-comm-prop} and the second term contributes to (\ref{second-technical-formula}).
This concludes the proof of (\ref{second-technical-formula}). 

Let us now  complete the proof of (\ref{Heisenberg-derivative-expression}): First,  we note that by Proposition~\ref{comm-H-dGa}
\beqa
 [ H^{(1)}(\xi), 1\otimes \qd ] \chi^{(1)}=O(t^{-1}).
\eeqa
Making use of this fact and of   Lemma~\ref{j_0Lemma}, we can write
\begin{align}
& \cGao(j)^*\chi^{(1)}\ti\chi^{(1)}(\Ga(\qd)\otimes  1) [ H^{(1)}(\xi), 1\otimes \qd ] \chi^{(1)} \cGao(j)\non\\
&\qquad =\cGao(j)^*\chi^{(1)}(\Ga(\qd)\otimes  1)\ti \chi^{(1)} [ H^{(1)}(\xi), 1\otimes \qd ] \chi^{(1)} \cGao(j)+O(t^{-2}).
\end{align}
Next, we note that by Proposition~\ref{q-a-commutator}
\begin{align}
\ti \chi^{(1)} \i[ H^{(1)}(\xi), 1\otimes \qd ] \chi^{(1)}&= \fr{1}{t}\ti\chi^{(1)}C (1\otimes (q')^t)\chi^{(1)}+O(t^{-2})\non\\
&= \fr{1}{t}\ti\chi^{(1)}C \ti\chi^{(1)}(1\otimes (q')^t)\chi^{(1)}+O(t^{-2}), \label{I-proposition-formula}
\end{align}
where $C$ is a bounded operator on $\Fock\otimes\Fock^{(1)}$, which satisfies $[C,1\otimes p^t]=O(t^{-1})$ 
for any $p\in C^{\infty}(\real)_{\real}$ s.t. $p'\in C_0^{\infty}(\real)$ and in the second step in (\ref{I-proposition-formula})  we made use of Lemma~\ref{Helffer-second-component-lemma}.
The second contribution to $C_t$ is thus given by
\beqa
C_{2,t}:= (\Ga(\qd)\otimes  1)\ti\chi^{(1)}C\ti\chi^{(1)}. 
\eeqa
 Again, it is obvious that $C_{2,t}$ satisfies the first property in~(\ref{C-second-property}), and the second property in~(\ref{C-second-property}) follows from $[C,1\otimes p^t]=O(t^{-1})$  and  Lemma~\ref{Helffer-second-component-lemma}.  
 \end{proof}

\section{Propagation estimates}\label{propagation-estimates}
\setcounter{equation}{0}

In this section we use the expressions for Heisenberg derivatives of  propagation observables,
established in Section~\ref{Heisenberg}, to prove suitable minimal velocity propagation estimates.  
We will use these estimates in Section~\ref{time-convergence-proofs} to verify the existence of the relevant asymptotic
observables. 
\bep\label{main-propagation-estimate} Let $\chi\in C_0^{\infty}(\real)_{\real}$ be supported in $\mcJz$ and $\xi\in N_0$.  
 Fix $0<\eps<c_0<\cm$, where $\cm$ appeared in the Mourre estimate (\ref{Mourre-est-one}), and $R>\eps$.  
\begin{enumerate}
[label = \textup{(\alph*)}, ref =\textup{(\alph*)},leftmargin=*]
\item\label{MainPropEstFullA} Let $\Int_0=[-R,c_0]$. Then there exists $c>0$ such that for all $\Psi^{(1)}\in\Fock\otimes\Fock^{(1)}$:
\beqa
 \int_1^{\infty}dt\, \fr{1}{t}\lan\Psi^{(1)}_t,\chi^{(1)}(1\otimes \mathbf{1}_{ \Int_0 }(a_{\xi_0}/t) )\chi^{(1)}\Psi^{(1)}_t\ran\leq c\|\Psi^{(1)}\|^2,
\label{single-particle-prop-est}
\eeqa
where $\Psi_t^{(1)}:=\e^{-\i t H^{(1)}(\xi)  }\Psi^{(1)}$ and $\chi^{(1)}:=\chi(H^{(1)}(\xi))$. 
\item\label{MainPropEstRightA} Let $j_0,j_{\infty}$ be as specified in Definition~\ref{j-definition-a} and s.t. $j_0^2+j_{\infty}^2=1$. Let $\Int=[-R,-\eps]\cup [\eps,c_0]$.
Then there exists $c>0$ such that for  all $\Psi\in\Fock$:
\beqa
\int_1^{\infty}dt\, \fr{1}{t}\lan\Psi_t,\cGao(j^t)^*\chi^{(1)}(1\otimes \mathbf{1}_{\Int}(a_{\xi_0}/t) )\chi^{(1)}\cGao(j^t)\Psi_t\ran\leq c\|\Psi\|^2,
\label{multi-particle-prop-est}
\eeqa
where $\Psi_t:=\e^{-\i t H(\xi)  }\Psi$. 
\end{enumerate}
\eep
\begin{proof} We set $a:=a_{\xi_0}$ and  start with a brief consideration which is relevant for both parts of the proposition. 
Let $q\in C^{\infty}(\real)$ be s.t. $q'\in C_0^{\infty}(\real)$, $q'\geq 0$, $\sqrt{q'}\in C_0^{\infty}(\real)$ and $\supp\, q'\subset [-R-1, c_0']$
for some $c_0<c_0'<\cm$.
Let us consider the propagation observable
\beqa
\Phi^{(1)}(t)=\chi^{(1)}(1\otimes q )\chi^{(1)},
\eeqa
where we set $q:=q(a/t)$. Its Heisenberg derivative gives
\beqa
\bD^{(1)}\Phi^{(1)}(t)=\chi^{(1)}\bigl(-\fr{1}{t} 1\otimes (a/t)q'+\i\ti\chi^{(1)}[H^{(1)}(\xi), 1\otimes q]\ti\chi^{(1)}\bigr)\chi^{(1)},
\eeqa 
where we chose some function $\ti\chi\in C_0^{\infty}(\real)_{\real}$, supported in $\mcJz$, s.t. $\chi\ti \chi=\chi$. Next, making
use of Proposition~\ref{q-a-commutator}, we can write
\begin{align}
\ti\chi^{(1)}[H^{(1)}(\xi), 1\otimes q]\ti\chi^{(1)} 
& = \fr{1}{t}(1\otimes\sqrt{q'})\ti\chi^{(1)}
[H^{(1)}(\xi), 1\otimes a ]^\circ\ti\chi^{(1)}(1\otimes\sqrt{q'}) +O(t^{-2})\non\\
&  \geq\fr{\cm}{t} (1\otimes\sqrt{q'})(\ti\chi^{(1)})^2 (1\otimes\sqrt{q'}) +O(t^{-2})\non\\
&  =\fr{\cm}{t} \ti\chi^{(1)} (1\otimes q') \ti\chi^{(1)}+O(t^{-2}), \label{commutator-lower-bound}
\end{align}
where in the second step we made use of the Mourre estimate~(\ref{Mourre-est-one}) and in the last step of  
Lemma~\ref{Helffer-second-component-lemma}. (The notation $[\,\,\, , \,\,\, ]^\circ$ is explained in Appendix~\ref{commutator-bounds-appendix}). On the other hand
\beqa
-\fr{1}{t} 1\otimes (a/t)q'\geq -\fr{c_0'}{t} 1\otimes  q'. \label{time-derivative-lower-bound}
\eeqa
Thus we obtain from (\ref{commutator-lower-bound}) and (\ref{time-derivative-lower-bound}) that
\beqa
\bD^{(1)}\Phi^{(1)}(t)&\geq &\fr{c}{t} \chi^{(1)} (1\otimes q') \chi^{(1)}+O(t^{-2}), \label{general-bound-prop-est}
\eeqa
where $c:=\cm-c_0'>0$. 

Now we are ready to prove part \ref{MainPropEstFullA} of the proposition. By choosing   $q$ s.t. $q'=1$ on 
$\Int_0=[-R,c_0]$, we obtain from~(\ref{general-bound-prop-est}) that
\beqa
\bD^{(1)}\Phi^{(1)}(t) \geq \fr{c}{t} \chi^{(1)} \bigl(1\otimes \mathbf{1}_{ \Int_0 }(a/t) \bigr) \chi^{(1)}+O(t^{-2}), \label{single-particle-lower-bound}
\eeqa
 By integrating this expression along the time evolution we obtain  (\ref{single-particle-prop-est}).

Proceeding to part \ref{MainPropEstRightA} of the proposition we choose $q'$  s.t.  $\supp\, q'\subset [-R-1,-\eps/2  ]\cup [\eps/2, c_0']$ for $c_0<c_0'<\cm$ 
and $q'=1$ on $[-R,-\eps] \cup [\eps,c_0]$. We also require that
\beqa
q(\la)= \int_{0}^{\la}q'(s)ds
\eeqa
to ensure that $q$ vanishes in a neighbourhood of zero. We consider the propagation observable
\beqa
\Phi(t)=\chi\dGa(q)\chi.
\eeqa
Proposition~\ref{a-dGa-Heisenberg} gives that
\begin{equation}
\bD\Phi(t) = \Ga(j_0)\mathbf{D}\Phi(t)\Ga(j_0)+\cGao(j)^* \bD^{(1)} \Phi^{(1)}(t)  \cGao(j)+O(t^{-2}),
\label{Heisenberg-derivative-statement-one}
\end{equation}
where we set $j:=j^t$. As for the second  term on the r.h.s. above, we obtain
\begin{align}
& \cGao(j)^*\chi^{(1)} \bD^{(1)} \Phi^{(1)}(t)  \chi^{(1)}\cGao(j)\non\\
& \qquad \geq\fr{c}{t}\cGao(j)^*\chi^{(1)} \bigl(1\otimes \mathbf{1}_{ \Int }(a/t) \bigr) \chi^{(1)}\cGao(j)+O(t^{-2}),
\label{lower-bound-on-DPhi}
\end{align}
where we made use of (\ref{general-bound-prop-est}). 
Let us now estimate the   first term on the r.h.s. of (\ref{Heisenberg-derivative-statement-one}). We choose  $\ti j$ as specified in Definition~\ref{j-definition-a}, s.t. $\ti j_0^2+\ti j_{\infty}^2=1$ and  $\supp\,\ti j_0$ does not intersect with the support of $q$. Then, making use again of Proposition~\ref{a-dGa-Heisenberg} and of formula~(\ref{lower-bound-on-DPhi}), we obtain
\beqa
\Ga(j_0)\mathbf{D}\Phi(t)\Ga(j_0)=\Ga(j_0)\cGao(\ti j)^*\bD^{(1)}\Phi^{(1)}(t) \cGao(\ti j) \Ga(j_0)+O(t^{-2})\geq O(t^{-2}),  \label{a-rest-term}
\eeqa
i.e. this term is bounded from below by an integrable contribution.

Making use of  (\ref{lower-bound-on-DPhi}) and   (\ref{a-rest-term}),  we obtain
\beqa
\bD\Phi(t)\geq \fr{c}{t}\cGao(j)^*\chi^{(1)} \bigl(1\otimes \mathbf{1}_{\Int}(a/t) \bigr)\chi^{(1)}\cGao(j)+O(t^{-2}), 
\eeqa
where $c>0$.
By integrating both sides of this inequality along the time evolution and making use of the fact that $\Phi(t)$ is bounded, uniformly in time, we conclude the proof. \end{proof}


\bep\label{minimal-velocity-estimate}  Let $\chi\in C_0^{\infty}(\real)_{\real}$ be supported in $\mcJz$   and $\xi\in N_0$. 
Then there exist $c>0$ and $0<\epsz<\cm/2$, where $\cm$ appeared in  the Mourre estimate (\ref{Mourre-est-one}), s.t. for any $R>0$ and $\Psi\in\Fock$:
\begin{equation}
\int_1^{\infty}\bigl\|\Ga\bigl(\mathbf{1}_{[-R,\epsz]}(a_{\xi_0}/t)\bigr)\chi(H(\xi))\Psi_t\bigr\|^2\fr{dt}{t}\leq c\|\Psi\|^2,
\end{equation}
where $\Psi_t=\e^{-\i tH(\xi)}\Psi$. 
\eep
\begin{proof} We set $a:=a_{\xi_0}$ and $A:=\dGa(a_{\xi_0})$.
Let $q\in C_0^{\infty}(\real)$ be  s.t. $0\leq q\leq 1$. Suppose that $q$ is  supported in  $[-R-1,2\epsz]$ and  $q=1$ on  $[-R, \epsz]$
for some $\epsz>0$ to be specified later. Moreover, suppose that $q'=q_+-q_-$, where $q_{\pm}\geq 0$, $\sqrt{q_{\pm}}\in C_0^{\infty}(\real)$,
$\supp\,q_{+}\subset [-R-1,-R]$, $\supp\, q_-\subset [\epsz,2\epsz]$.

We set $q^t:=q(a/t)$ and introduce the propagation observable
\begin{equation}
\Phi_{\xi}(t) = \chi(H(\xi))\Ga(\qt) \fr{A}{t}\Ga(\qt)\chi(H(\xi)). \label{direct-int-observable}
\end{equation}
Note that by Corollary~\ref{Cor-DomInv-Gamma} we have $\Gamma(q^t) \chi(H(\xi)) \Fock\subset D(H(\xi))\cap D(A)$, such that the computation above -- as well as the one to follow -- is meaningful.
It can easily be shown that $\Phi_{\xi}$ is bounded uniformly in time. Let us now study the Heisenberg derivative of $\Phi_{\xi}$: 
We set $q:=q^t$, $\chi:=\chi(H(\xi))$ and write:
\beqa
\mathbf{D}\Phi_{\xi}(t)=
\chi  \Ga(q) \mathbf{D}\bigl(\fr{A}{t}\bigr) \Ga(q)\chi
+\chi \mathbf{D}(\Ga(q))\fr{A}{t} \Ga(q)\chi
+\chi \Ga(q)\fr{A}{t} \mathbf{D}(\Ga(q))\chi. \label{minimal-heisenberg-derivative}
\eeqa
As for the first term on the r.h.s. above, we obtain:
\beqa
\chi \Ga(q) \mathbf{D}\bigl(\fr{A}{t}\bigr) \Ga(q)\chi=-\fr{1}{t}\chi \Ga(q) \fr{A}{t}   \Ga(q) \chi+\fr{1}{t}\chi \Ga(q)\i[H(\xi),A]^\circ\Ga(q)\chi.
\label{minimal-leading-term}
\eeqa
Concerning the first term on the r.h.s. of (\ref{minimal-leading-term}), we note the bound
\begin{align}
\fr{1}{t}\chi \Ga(q) \fr{A}{t}   \Ga(q) \chi & =
\fr{1}{t}\chi \Ga(q)\ti\chi \dGa( q_1,  (a/t) q_1)   \ti\chi   \Ga(q) \chi+O(t^{-2})\non\\
& \leq \fr{1}{t}\chi \Ga(q)\ti\chi \dGa( q_1,(a/t) q_2  )  \ti\chi   \Ga(q) \chi+O(t^{-2})\non\\  
& \leq  c\epsz\fr{1}{t} \chi \Ga(q)^2\chi+O(t^{-2}), \label{minimal-add-leading}
\end{align}
where $c$ is independent of $\epsz$ and $R$. Here $\ti\chi\in C_0^{\infty}(\real)_{\real}$ is s.t. $\ti\chi\chi=\chi$ and $\ti\chi$ is supported 
in $\mcJz$. We also chose functions $ q_1,q_2\in C_0^{\infty}(\real)$, $0\leq  q_1,q_2\leq 1$ s.t. $q q_1=q$, $\supp\, q_1\subset [-R-2,3\epsz ]$, 
$\supp\, q_{2}\subset [-3\epsz,3\epsz]$ and $q_1(s)s\leq q_2(s)s$ for all $s\in \real$.
In (\ref{minimal-add-leading}) we made use of Lemma~\ref{j_0Lemma}, the fact that $\Ga(q_1)\dGa(a)=\dGa( q_1, q_1 a)$, and
\beq
\|\dGa( q_1,  (a/t) q_2 )(1+N)^{-1}\|\leq \|(a/t)q_{2} \|\leq 3\epsz.
\eeq
In view of this bound it is clear that the constant  $c=3\|(1+N)\ti\chi\|\|\ti\chi\|$, appearing in (\ref{minimal-add-leading}), is independent of $R$,
$\epsz$.  As for the second term on the r.h.s. of (\ref{minimal-leading-term}), we write
\begin{align}
\fr{1}{t}\chi \Ga(q)\i[H(\xi),A]^\circ\Ga(q)\chi & = \fr{1}{t}\ti \chi\Ga(q) \chi   \i[H(\xi),A]^\circ\chi\Ga(q)\ti\chi+O(t^{-2})\non\\
& \geq \cm\fr{1}{t}\ti \chi\Ga(q) \chi^2\Ga(q)\ti\chi+O(t^{-2})\non\\
& = \cm\fr{1}{t}\chi\Ga(q)^2 \chi+O(t^{-2}). \label{application-of-Mourre}
\end{align}
The first step above follows from Lemma~\ref{j_0Lemma} and from the fact that
\beqa
\|[H(\xi),A]^\circ\Ga(q)\chi\|<\infty\quad \textup{and}\quad   
\|[H(\xi),A]^\circ\chi\|<\infty.
\eeqa
These bounds are  simple consequences of Corollary~\ref{Cor-DomInv-Gamma} and Lemma~\ref{minimal-lemma}, after rewriting 
\begin{align}
\non [H(\xi),A]^\circ\Ga(q)\chi & = \bigl\{[H(\xi),A]^\circ (N+1)^{-1}
\Ga(q)(H_0(\xi)+1)^{-3}\bigr\}\\
& \quad \times  \bigl\{ (N+1)(H_0(\xi) +1)^{-1}\bigr\}\bigl\{ (H_0(\xi)+1)^4\chi\bigr\} 
\end{align}
and recalling Corollary~\ref{Cor-technical2}.
In the second step of (\ref{application-of-Mourre}) we used (\ref{Mourre-est}) and in the last step once more Lemma~\ref{j_0Lemma}. 
Summing up, we got
\beqa
\chi  \Ga(q) \mathbf{D}\bigl(\fr{A}{t}\bigr) \Ga(q)\chi\geq (\cm-c\epsz)\fr{1}{t}\chi\Ga(q)^2 \chi+O(t^{-2}).
\eeqa

Let us now consider the remaining two terms on the r.h.s. of (\ref{minimal-heisenberg-derivative}). First, we note that  
\begin{align}
\chi \mathbf{D}(\Ga(q))\fr{A}{t} \Ga(q)\chi &= \chi \mathbf{D}(\Ga(q))\ti\chi \fr{A}{t} \Ga(q)\chi
+\chi\mathbf{D}(\Ga(q)) [\dGa(q,(a/t)q ),\ti\chi]\chi\non\\
&= \chi\ti\chi\mathbf{D}(\Ga(q))\ti\chi \fr{A}{t} \Ga(q)\chi+O(t^{-2}).
\end{align}
Here in the second step we applied Lemma~\ref{minimal-auxiliary} and  Proposition~\ref{comm-H-Ga}, which ensures that
$\chi\mathbf{D}(\Ga(q))=O(t^{-1})$.
As for the first term on the r.h.s. above, we obtain from Proposition~\ref{second-technical}:
\beqa
\ti\chi\mathbf{D}(\Ga(q))\ti\chi=\mathbf{D}(\ti\chi\Ga(q)\ti\chi)=\fr{1}{t}\cGao(j)^*\ti\chi^{(1)}C_t(1\otimes q')\ti\chi^{(1)} \cGao(j)+O(t^{-2}),
\eeqa
where we set $j:=j^t$ and $q':=(q')^t$. Thus, recalling that $q'=q_+-q_-$ and $\sqrt{q_{\pm}}\in C_0^{\infty}(\real)$, we can write
\beqa
\mathbf{D}(\ti\chi\Ga(q)\ti\chi)=\sum_{\si\in \{\pm\}}\si\fr{1}{t}\cGao(j)^*\ti\chi^{(1)} (1\otimes \sqrt{q_{\si} } )   C_t(1\otimes \sqrt{q_{\si}} )\ti\chi^{(1)} \cGao(j)+O(t^{-2}),
\eeqa
where we exploited the second property in~(\ref{C-second-property}). Thus we get
\begin{align}
& \chi\ti\chi \mathbf{D}(\Ga(q))\ti\chi \fr{A}{t} \Ga(q)\chi\non\\
& \quad =\sum_{\si\in \{\pm\}}\si\fr{1}{t}\chi\cGao(j)^*\ti\chi^{(1)} (1\otimes \sqrt{q_{\si} })C_t\dGa^{(1)}(q,(a/t)q) (1\otimes \sqrt{q_{\si}})\ti\chi^{(1)} \cGao(j)\chi\non\\
& \quad \qquad +O(t^{-2}),
\end{align}
where we made use of the fact that $\cGao(j)A\Ga(q)=\dGa^{(1)}(q,(a/t)q)\cGao(j)$ and then of Lemma~\ref{minimal-auxiliary} to exchange $\dGa^{(1)}(q,(a/t)q)$ with $\ti\chi^{(1)}$. 
Since $C_t\dGa^{(1)}(q,(a/t)q)=O(1)$, by the first part of property~(\ref{C-second-property}), we obtain for any $\Psi\in \Fock$:
\begin{align}
& \bigl|\bigl\lan \Psi_{t},\chi\ti\chi \mathbf{D}(\Ga(q))\ti\chi \fr{A }{t} \Ga(q)\chi\Psi_{t}\bigr\ran\bigr|\non\\
& \qquad \leq\sum_{\si\in \{\pm \}} \fr{c}{t} \|(1\otimes \sqrt{q_{\si}})\ti\chi^{(1)} \cGao(j)\chi\Psi_{t}\|^2
+O(t^{-2})\|\Psi\|^2.
\end{align}
This expression is integrable, uniformly in $\Psi$ from the unit ball in $\Fock$, by the Cauchy-Schwarz inequality
and Proposition~\ref{main-propagation-estimate}. (To apply this latter proposition we assume that $2\epsz<\cm$). 
The last term on the r.h.s of (\ref{minimal-heisenberg-derivative}) is treated
analogously.

Altogether, we have obtained that
\beq
\mathbf{D}\Phi_{\xi}(t)\geq  (\cm-c\epsz)\fr{1}{t}\chi\Ga(q)^2 \chi+B(t)+O(t^{-2}),
\eeq
where $c$ is independent of $\epsz$ and $R$, and $B(t)$ is integrable along the time-evolution provided that $2\epsz<\cm$.
By choosing $\epsz$ sufficiently small, we conclude the proof. 
\end{proof}

\section{Existence of some asymptotic observables}\label{time-convergence-proofs}
\setcounter{equation}{0}

As usually in the time-dependent approach to the problem of asymptotic completeness, the
central question is the existence of suitable asymptotic observables as strong limits, as time goes to infinity, 
of their approximating sequences.  In this section we
answer this question with the help of the propagation estimates established in Section~\ref{propagation-estimates}.
With this information at hand, the proof of asymptotic completeness, completed in Sections~\ref{Geometric-WO} and \ref{AC-section},
is relatively straightforward.
\bet\label{asymptotic-observable} Let $\chi\in C_0^{\infty}(\real)_{\real}$ be supported in $\mcJz$   and $\xi\in N_0$.  
Let $q\in C^{\infty}(\real)$ be s.t. $0\leq q\leq 1$, $q'\in C_0^{\infty}(\real)$ and $\supp\, q'\subset (-\infty,\cm)\backslash [-\eps,\eps]$, 
for some $0<\eps<\cm$, where   $\cm$ appeared in (\ref{Mourre-est}).  
Then the following strong limit exists
\begin{equation}
Q^+(H(\xi))\chi := \slim_{t\to \infty} \e^{\i tH(\xi)} \Ga(q^t)\e^{-\i tH(\xi)}\chi ,\label{standard-asymp-const}
\end{equation}
and commutes with bounded Borel functions of $H(\xi)$. (Here we set $\chi:=\chi(H(\xi))$). Moreover, if $\supp\, q \subset (-\infty,\epsz)$, where
$\epsz$ appeared in Proposition~\ref{minimal-velocity-estimate}, then $Q^+(H(\xi))\chi=0$.
\eet
\begin{proof}  We set $q:=q^t$  and
define $\Phi(t)=\chi\Ga(q)\chi$. Making use of Proposition~\ref{second-technical}, we obtain 
\begin{equation}
\mathbf{D}\Phi(t) = \fr{1}{t}\cGao(j)^*\ti\chi^{(1)}C_t(1\otimes q')\ti\chi^{(1)} \cGao(j)+O(t^{-2}),
\end{equation}
where we set $j:=j^t$. Let $\ti q \in C_0^{\infty}(\real)$ be supported in  $(-\infty,\cm)\backslash [-\eps,\eps]$
and equal to one on the support of $q'$. Then, making use of the second property 
in~(\ref{C-second-property}), we can write
\beqa
\mathbf{D}\Phi(t)=\fr{1}{t}\cGao(j)^*\ti\chi^{(1)}(1\otimes \ti q)C_t(1\otimes q')\ti\chi^{(1)} \cGao(j)+O(t^{-2}).
\eeqa
Since $C_t=O(1)$, we obtain 
\begin{align}
|\lan \Psi_{1,t},\mathbf{D}\Phi(t) \Psi_{2,t}\ran| & \leq  \fr{c}{t} \|(1\otimes \ti q)\ti\chi^{(1)}\cGao(j)\Psi_{1,t}\|\,\|(1\otimes q')\ti\chi^{(1)} \cGao(j)\Psi_{2,t}\|\non\\
&\quad  +O(t^{-2})\|\Psi_1\|\|\Psi_2\|, \label{A-const-H-deriv}
\end{align}
where $\Psi_i\in \Fock$ and $\Psi_{i,t}=\e^{-\i t H(\xi)}\Psi_i$, $i\in \{1,2\}$. By integrating both sides of this inequality
over some time interval, applying the Cauchy-Schwarz inequality to the integral of the first term on the r.h.s. 
of (\ref{A-const-H-deriv}), taking supremum over $\Psi_1$ s.t. $\|\Psi_1\|\leq 1$ and exploiting
 Proposition~\ref{main-propagation-estimate}, we obtain  strong
convergence in (\ref{standard-asymp-const}) by the Cook method. 
Now we choose $\ti\chi\in C_0^{\infty}(\real)_{\real}$, supported in $\mcJz$ and s.t. $\ti\chi\chi=\chi$. Lemma~\ref{j_0Lemma} gives
\beqa
\e^{\i tH(\xi)}\Ga(q)\e^{-\i tH(\xi)}\chi=\e^{\i tH(\xi)}\ti \chi \Ga(q) \ti \chi\e^{-\i tH(\xi)}\chi+O(t^{-1}). \label{shifting-chi}
\eeqa
 The second term on the r.h.s. above  converges strongly by the above consideration. 
By a computation analogous to (\ref{shifting-chi}) one shows that $Q^+(H(\xi))\chi$  commutes with $H(\xi)$. This concludes the proof of 
(\ref{standard-asymp-const}). 

Let us now show the last statement of the theorem, i.e.  that for $q$ s.t.  $\supp\, q \subset (-\infty,\epsz)$ there holds
\beqa
Q^+(H(\xi))\chi=0. \label{vanishing-of-asymptotic-obs}
\eeqa
Let $q_{R}\in C_0^{\infty}(\real)$, $0\leq q_R\leq 1$, be s.t.  $q_{R}(s)=q(s)$ for $s\in (-R,\infty)$ and $q_{R}=0$ for $s<-R-1$, for some  $R>\cm$.
Then, coming back to the explicit notation $q^t=q(a/t)$ and $q_R^t=q_R(a/t)$,  we obtain from Proposition~\ref{minimal-velocity-estimate} and from (\ref{standard-asymp-const}) that
\beqa
\slim_{t\to \infty} \e^{\i tH(\xi)} \Ga(q^t_R)\e^{-\i tH(\xi)}\chi=0.
\eeqa
On the other hand, Lemma~\ref{negative-spectrum-of-a} gives that
\beqa
\| \big(\Ga(q^t_R)-\Ga(q^t)\big)\e^{-\i tH(\xi)}\chi \Psi\|=O(R^{-1})
\eeqa
uniformly in $t$ for $\Psi$ from some dense domain in $\Fock$. This concludes the proof of (\ref{vanishing-of-asymptotic-obs}). 
\end{proof}

\bet\label{asymptotic-observable-two} Let $\chi\in C_0^{\infty}(\real)_{\real}$ be supported in $\mcJz$   and $\xi\in N_0$.  
Let $p\in C^{\infty}(\real)$ be s.t. $0\leq p\leq 1$, $p'\in C_0^{\infty}(\real)$ and $\supp\, p'\subset (-\infty,\cm)$, 
where   $\cm$ appeared in (\ref{Mourre-est}).  
Then the following strong limit exists
\beqa
Q^+(H^{\ex}(\xi))_{\infty}\chi^{\ex}:=
\slim_{t\to \infty} \e^{\i tH^{\ex}(\xi)}\big( 1\otimes \Ga(p^t)\big)\e^{-\i tH^{\ex}(\xi)}\chi^{\ex},\label{infty-asympt-const}
\eeqa
and commutes with bounded Borel functions of  $H^{\ex}(\xi)$. (Here we set $\chi^{\ex}:=\chi(H^{\ex}(\xi))$).

If, in addition, $p=1$ on $[\cm,\infty)$, then 
\beq
Q^+(H^{\ex}(\xi))_{\infty}\chi^{\ex}= \chi^{\ex}. \label{Q-infty-projections}
\eeq
\eet
\begin{proof} Concerning the proof of (\ref{infty-asympt-const}), we set $p:=p^t$, choose $\ti\chi\in C_0^{\infty}(\real)_{\real}$, supported in $\mcJz$ and 
 s.t. $\chi\ti\chi=\chi$.  We  note the relation
\beq
[\ti\chi^{\ex}, 1\otimes \Ga(p)] \chi^{\ex}=[\ti\chi^{(1)}, 1\otimes p] \chi^{(1)}=O(t^{-1}), \label{Gap-comm}
\eeq
which is a consequence of the decomposition~(\ref{Hamiltonian-decomp}), Lemma~\ref{expansion-truncation}, (which ensures that only $\ell=0$ and $\ell=1$ terms survive in this expansion), and of Lemma~\ref{Helffer-second-component-lemma}. Thus it suffices to
prove strong convergence of $t\to \e^{\i tH^{\ex}(\xi)}\chi^{\ex}(1\otimes \Ga(p))\chi^{\ex}\e^{-\i t H^{\ex}(\xi)}$ for any 
$\chi\in C_0^{\infty}(\real)_{\real}$ supported in $\mcJz$.  We apply  decomposition~(\ref{Hamiltonian-decomp}) and Lemma~\ref{expansion-truncation} to this expression. The $\ell=0$ component gives $\chi(H(\xi))^2$ which is time-independent. The $\ell=1$ component has the form
\beq
t\to \chi^{(1)}\e^{\i tH^{(1)}(\xi)} (1\otimes  p  )\e^{-\i tH^{(1)}(\xi)}\chi^{(1)}. \label{l=1-infty-prop-obs}
\eeq
We consider the propagation observable $\Phi_{\infty}(t):=\chi^{(1)}(1\otimes  p  )\chi^{(1)}$. To prove the strong convergence
of (\ref{l=1-infty-prop-obs}) we will show  integrability of the Heisenberg derivative 
\beqa
\bD^{(1)}\Phi_{\infty}(t)=\chi^{(1)}\bigl(-\fr{1}{t} 1\otimes (a/t)p'+\i[H^{(1)}(\xi), 1\otimes p]\bigr)\chi^{(1)}. \label{structure-of-derivative-beginning}
\eeqa
By Proposition~\ref{q-a-commutator} 
\beqa
\chi^{(1)}[H^{(1)}(\xi), 1\otimes p]\chi^{(1)}=\fr{1}{t}\chi^{(1)}C (1\otimes p')\chi^{(1)}+O(t^{-2}),
\eeqa
where $C$ is a bounded operator on $\Fock\otimes\Fock^{(1)}$, which satisfies 
\beqa
[C,1\otimes p_1^t]=O(t^{-1}) \label{commutator-vanishing}
\eeqa
for any $p_1\in C^{\infty}(\real)_{\real}$ s.t. $p'_1\in C_0^{\infty}(\real)$. Let $\ti p\in C_0^{\infty}(\real)_{\real}$ be supported in  $(-\infty,\cm)$ 
and be equal to one on the support of $p'$. Then, due to  (\ref{commutator-vanishing}), we obtain
\beqa
\bD^{(1)}\Phi_{\infty}(t)=\fr{1}{t}\chi^{(1)}(1\otimes \ti p)\ti C_t (1\otimes p')\chi^{(1)}+O(t^{-2}), \label{structure-of-derivative-end}
\eeqa
where $\ti C_t=-1\otimes (a/t)\ti p+C$ is a family of operators which is uniformly bounded in $t$. Thus we can write
\begin{align}
 |\lan \Psi_{1,t},\bD^{(1)}\Phi_{\infty}(t) \Psi_{2,t}\ran| & \leq  \fr{c}{t} \|(1\otimes \ti p)\chi^{(1)}\Psi_{1,t}\|   \|(1\otimes p')\chi^{(1)}\Psi_{2,t}\| +O(t^{-2})\|\Psi_{1}\| \|\Psi_{2}\|,
\end{align}
where $\Psi_i\in \Fock\otimes \Fock^{(1)}$, $\Psi_{i,t}=\e^{-\i tH^{(1)}(\xi)}\Psi_i$, $i\in\{1,2\}$.
With the help of this bound, the Cauchy-Schwarz inequality and Proposition~\ref{main-propagation-estimate}, we obtain  strong
convergence of (\ref{l=1-infty-prop-obs}) by the Cook method. This completes the proof of~(\ref{infty-asympt-const}). To show that  
the limit commutes with bounded functions of the Hamiltonian, one proceeds analogously as in relation~(\ref{Gap-comm}) above.

Let us now proceed to the proof of (\ref{Q-infty-projections}). We come back to the explicit notation  $p^t=p(a/t)$. As we have shown above
(cf. formula~(\ref{l=1-infty-prop-obs}))
\begin{align}
&  \e^{\i tH^{\ex}(\xi)}\bigl( 1\otimes \Ga(p^t)\bigr)\e^{-\i t H^{\ex}(\xi)}\chi^{\ex}\non\\
& \qquad  = \chi(H(\xi)) \oplus \ti\chi^{(1)}\e^{\i t H^{(1)}(\xi)}( 1\otimes p^t)\e^{-\i t H^{(1)}(\xi)}\ti\chi^{(1)} \chi^{(1)}+O(t^{-1}).
\end{align}
Thus it suffices to show that for $\Psi\in\Fock\otimes\Fock^{(1)}$  there holds
\beqa
\lim_{t\to\infty}\lan\Psi_t,\chi^{(1)}(1\otimes p(a/t))\chi^{(1)}\Psi_t\ran=\lan \Psi, (\chi^{(1)})^2 \Psi\ran. \label{projection-property}
\eeqa
Setting $q:=1-p$, this is equivalent to
\beqa
\lim_{t\to\infty}\lan\Psi_t,\chi^{(1)}(1\otimes q(a/t))\chi^{(1)}\Psi_t\ran=0. \label{decay-of-q}
\eeqa 
We note that $\supp\,q\subset (-\infty,\cm)$.
Let us choose a function $q_R\in C_0^{\infty}(\real)$, $0\leq q_R\leq 1$, which coincides with 
$q$ on $(-R, \infty)$, but is equal to zero on $(-\infty, -R-1]$ for some  $R>1$. We obtain from (\ref{single-particle-prop-est})  that
\beqa
\lim_{t\to\infty} \lan\Psi_{t},\chi^{(1)}(1\otimes q_R(a/t))\chi^{(1)}\Psi_{t}\ran=0,
\eeqa 
where we exploited the first part of this proposition to obtain convergence.

 Now let $\Psi$ be an element of the domain of $1\otimes a$. Then
$\Psi$ belongs to the domain of $(1\otimes a)\chi^{(1)}$, since $H^{(1)}(\xi)$ is of class $C^{1}(1\otimes a)$. 
Cf. \cite[Lemma 2.2, Proposition~2.8]{MR12}.
Furthermore,  the operator representing the commutator form $\i [H^{(1)}(\xi),1\otimes a]$ is
given by $\i[H^{(1)}(\xi),1\otimes a]^\circ = -\nabla\Omega(\xi-\dGa^{(1)}(k))\cdot (1\otimes v) 
+ 1\otimes \nabla\om\cdot v$,
which is $H^{(1)}(\xi)$-bounded. Consequently,
the group  $\e^{-\i t H^{(1)}(\xi)}$ preserves $D(1\otimes a)\cap D(H^{(1)}(\xi))$. Now we set $\ti q_R:=q-q_R$ and compute
\begin{align}
&  \bigl\lan\Psi_{t},\chi^{(1)}(1\otimes \ti q_R(a/t))\chi^{(1)}\Psi_{t}\bigr\ran\non\\
&\qquad =\bigl\lan \Psi, \chi^{(1)}\e^{\i tH^{(1)}(\xi) }\Bigl(1\otimes \fr{\ti q_R(a/t)}{(a/t) }\Bigr) \e^{-\i tH^{(1)}(\xi)}  
\e^{\i tH^{(1)}(\xi)}\Bigl(1\otimes \fr{a}{t}\Bigr)\e^{-\i tH^{(1)}(\xi) }\chi^{(1)}\Psi \bigr\ran\non\\
&\qquad =  \bigl\lan \Psi,\chi^{(1)} \e^{\i tH^{(1)}(\xi) }\Bigl(1\otimes \fr{\ti q_R(a/t)}{(a/t) }\Bigr) \e^{-\i tH^{(1)}(\xi)}\non\\ 
& \qquad\quad \times \fr{1}{t}\int_0^tdt'  \e^{\i t'H^{(1)}(\xi)}\i[H^{(1)}(\xi), 1\otimes a]^\circ\e^{-\i t'H^{(1)}(\xi) }\chi^{(1)}\Psi \bigr\ran\non\\ 
& \qquad \quad + \bigl\lan \Psi, \chi^{(1)}\e^{\i tH^{(1)}(\xi) }\Bigl(1\otimes \fr{\ti q_R(a/t)}{(a/t) }\Bigr) \e^{-\i tH^{(1)}(\xi)} \fr{1}{t} (1\otimes a) \chi^{(1)}\Psi \bigr\ran. 
\end{align}
Hence, making use of the fact that $\|[H^{(1)}(\xi), 1\otimes a]^\circ\chi^{(1)}\|<\infty$, (cf. Lemma~\ref{minimal-lemma}), we obtain
\beqa
|\lan\Psi_{t},(1\otimes \ti q_R(a/t))\Psi_{t}\ran|\leq \fr{c}{R}\|\Psi\|^2+\fr{c}{Rt} \|\Psi\|\, \|(1\otimes a)\chi^{(1)}\Psi\|.
\eeqa
Since this expression can be made arbitrarily small, uniformly in $t$, by choosing $R$ sufficiently large, we have
proven (\ref{decay-of-q}) for $\Psi$ in the domain of $(1\otimes a)$, which is dense. This concludes the proof. \end{proof}

\bet\label{wave-operators-theorem}  Let $\chi\in C_0^{\infty}(\real)_{\real}$ be supported in $\mcJz$   and $\xi\in N_0$.
Let $j_0,j_{\infty}$ be as specified in Definition~\ref{j-definition-a}, s.t. $j_0^2+j_{\infty}^2=1$, $\supp\, j_0'\subset (-\infty,\cm)$ and hence
$\supp\ j_{\infty}' \subset (-\infty,\cm)$, where $\cm$ appeared in (\ref{Mourre-est}).  Let $q=(q_0,q_{\infty}):=(j_0^2,j_{\infty}^2)$ 
(in particular $q_0+q_{\infty}=1$). Then the following strong limits exist: 
\begin{align}
W^{+}(q^t)(\xi)\chi^{\ex} & := \slim_{t\to\infty}\e^{\i tH(\xi)}\cGa(q^t)^*\e^{-\i tH^{\ex}(\xi) }\chi^{\ex},\\
W^{+}(q^t)(\xi)^*\chi & := \slim_{t\to\infty}\e^{\i tH^{\ex}(\xi)}\cGa(q^t)\e^{-\i tH(\xi) }\chi,
\end{align}
where we set $\chi:=\chi(H(\xi))$ and $\chi^{\ex}:=\chi(H^{\ex}(\xi))$.
These operators intertwine (bounded Borel functions of) $H(\xi)$ and $H^{\ex}(\xi)$.
\eet
\begin{proof}  We set $q:=q^t$, $j:=j^t$  and  consider the asymptotic observable $\Phi(t)=\chi^{\ex}\cGa(q)\chi$. 
Its non-symmetric Heisenberg derivative is given by
\beqa
\wt{\mathbf{D}}\Phi(t)=\chi^{\ex}\big(\dcGa(q,\pa_t q)+\i H^{\ex}(\xi)\cGa(q)-\i\cGa(q)H(\xi)   \big)\chi. 
\label{non-symmetric-derivatives}
\eeqa
The first term on the r.h.s. above can be rearranged as follows
\begin{align}
\chi^{\ex}\dcGa(q,\pa_t q) & = 2\chi^{\ex}\dGa^{\ex}(\un j, \un {\pa_t j})\cGa(j)\non\\
& = 2\chi^{\ex}\big(\dGa(\qz,\pa_t \qz)\otimes \Ga(\qi)+\Ga(\qz)\otimes \dGa(\qi,\pa_t \qi)\big)\cGa(j),
\end{align}
 where $\un j=\diag(j_0,j_{\infty})$, $\un{ \pa_tj}:=\diag(\pa_t j_0,\pa_t j_{\infty})$ are propagation observables  on $\mfh\oplus\mfh$ and
in the last step  we made use of Lemma~{\ref{more-fock-combinatorics}}. As for the remaining terms on the r.h.s. of (\ref{non-symmetric-derivatives}), 
we obtain from Lemma~\ref{commutator-exchange}  that
\beqa
\chi^{\ex}\big(H^{\ex}(\xi)\cGa(q)-\cGa(q)H(\xi)\big)\chi
=2\chi^{\ex}[H^{\ex}(\xi),\Ga^{\ex}(\un j)]\cGa(j)\chi+O(t^{-2}).
\eeqa
Thus, altogether, we get
\begin{align}
\wt{\mathbf{D}}\Phi(t) & = 2\chi^{\ex}\bigl(\dGa^{\ex}(\un j, \un {\pa_t j})+\i[H^{\ex}(\xi),\Ga^{\ex}(\un j)] \bigr)  \cGa(j)\chi+O(t^{-2}) \label{non-symm-deriv-one} \\
& = 2\chi^{\ex}\bigl(\dGa^{\ex}(\un j, \un {\pa_t j})+\i[H^{\ex}(\xi),\Ga^{\ex}(\un j)] \bigr)\chi^{\ex}  \cGa(j)\chi+O(t^{-2}), \label{non-symm-deriv-two} 
\end{align}
where in the last step we chose $\ti\chi\in C_0^{\infty}(\real)_{\real}$, supported in $\mcJz$ and s.t.
 $\ti\chi\chi=\chi$. To exchange $\cGa(j)\ti\chi$ with  $\ti\chi^{\ex}  \cGa(j)$, we made use of Lemma~\ref{Helffer} and of
the fact that
\beq
\chi^{\ex}[H^{\ex}(\xi),\Ga^{\ex}(\un j)]=O(t^{-1}),
\eeq
which follows from  Proposition~\ref{comm-H-Ga}.

Now we apply decomposition~(\ref{Hamiltonian-decomp}) of  $H^{\ex}(\xi)$. As for the   $\ell=0$ component, 
we obtain from (\ref{non-symm-deriv-one})
\begin{align}
\wt{\mathbf{D}}\Phi^{(0)}(t) & = \chi\bigl(2\dGa(j_0,\pa_t j_0)+2\i[H(\xi),\Ga(j_0)]\bigr)\Ga(j_0)\chi+O(t^{-2})\non\\
& = \chi\bigl(\dGa(q_0,\pa_t q_0)+\i[H(\xi),\Ga(q_0)]\bigr)\chi+O(t^{-2}). \label{asymptotic-const-reappearance}
\end{align}
To justify the second step above we make use of  the relations
\begin{align}
& \dGa(q_0,\pa_t q_0)=2\dGa(j_0,\pa_tj_0)\Ga(j_0),\\
& \chi [H(\xi),\Ga(q_0)]\chi=2\chi [H(\xi),\Ga(j_0)]\Ga(j_0)\chi+\chi [\Ga(j_0),[H(\xi),\Ga(j_0)]]\chi, \label{HGachi}
\end{align}
and of the fact that  the last term on the r.h.s. of (\ref{HGachi}) is $O(t^{-2})$ by Lemma~\ref{Ga-H-Ga}. 
We note that the first term on the r.h.s. of (\ref{asymptotic-const-reappearance}) is the Heisenberg derivative of $\Phi_0(t):=\chi\Ga(q_0)\chi$. 
We recall that $q_0'=(j_0^2)'=2j_0j_0'$ and, by Definition~\ref{j-definition-a}, $j_0$ is equal to one in some interval $[-\eps,\eps]$, $0<\eps<\cm$.
Thus $\supp\,q_0'\subset (-\infty, \cm)\backslash [-\eps,\eps]$ and the Heisenberg derivative of $\Phi_0$ can be shown to be integrable along
the time evolution as in the proof of Proposition~\ref{asymptotic-observable}.

Let us proceed to the $\ell=1$ component: Let $\chi^{(1)}:=\chi(H^{(1)}(\xi))$. From~(\ref{non-symm-deriv-two}) we obtain 
\begin{align}
\wt{\mathbf{D}}\Phi^{(1)}(t) & = 2\chi^{(1)}\bigl(\dGa(\qz,\pa_t\qz)\otimes \qi+\Ga(\qz)\otimes \pa_t\qi\non\\
& \quad +\i[H^{(1)}(\xi),\Ga(\qz)\otimes \qi]\bigr)\ti\chi^{(1)}\cGao(j)\chi+O(t^{-2}). \label{non-symm-deriv}
\end{align}
We note  that
\begin{align}
& \chi^{(1)}(\dGa(\qz,\pa_t\qz)\otimes \qi)\ti\chi^{(1)}=O(t^{-2}), \label{wave-op-decay-one}\\
& \chi^{(1)}[H^{(1)}(\xi),\Ga(\qz)\otimes 1](1\otimes \qi)\ti\chi^{(1)}=O(t^{-2}), \label{wave-op-decay-two} 
\end{align}
where (\ref{wave-op-decay-one}) follows from Corollary~\ref{resolvents} and (\ref{wave-op-decay-two}) is a consequence of 
Proposition~\ref{new-comm-prop}. Thus we obtain from~(\ref{non-symm-deriv}) that
\begin{align}
\wt{\mathbf{D}}\Phi^{(1)}(t) & = 2\chi^{(1)}(\Ga(\qz)\otimes 1) \bD^{(1)}( 1\otimes\qi)\ti\chi^{(1)}\cGao(j)\chi+O(t^{-2})\non\\
& = 2\chi^{(1)}(\Ga(\qz)\otimes 1)\ti \chi^{(1)} \bD^{(1)}( 1\otimes\qi)\ti\chi^{(1)} \cGao(j)\chi+O(t^{-2}), \label{a-derivative}
\end{align}
where in the last step we made use of the fact that $[\ti\chi^{(1)}, \Ga(\qz)\otimes 1)]=O(t^{-1})$, which follows from Lemma~\ref{j_0Lemma}
and of the estimate $ \bD^{(1)}( 1\otimes\qi)\ti\chi^{(1)}=O(t^{-1})$, which
is a consequence of Proposition~\ref{comm-H-dGa}. Proceeding as in (\ref{structure-of-derivative-beginning})--(\ref{structure-of-derivative-end}) 
above, we obtain that
\beqa
\ti\chi^{(1)} \bD^{(1)}( 1\otimes\qi)\ti\chi^{(1)}=\fr{1}{t}\ti\chi^{(1)}(1\otimes \ti j_{\infty})\ti C_t (1\otimes j_{\infty}')\ti\chi^{(1)}+O(t^{-2}),
\eeqa
where  $\ti j_{\infty}\in C_0^{\infty}(\real)_{\real}$ is supported in  $(-\infty,\cm)$ and is equal to one on the support of $j_{\infty}'$, and
$t\to\ti C_t$ is a family of operators, which is uniformly bounded in $t$. Thus we get
\begin{align} \label{wave-operators-finite}
& |\lan \Psi_{1,t}, \wt{\mathbf{D}}\Phi^{(1)}(t) \Psi_{2,t}\ran| \\
\non &\qquad \leq\fr{c}{t}  \| (1\otimes  |\ti j_{\infty} |) \ti\chi^{(1)}  \Psi_{1,t}\|\,\|(1\otimes  |j_{\infty}'|)\ti\chi^{(1)} \cGao(j)\chi \Psi_{2,t}\|+O(t^{-2})\|\Psi_1\|\,\|\Psi_2\|,
\end{align}
where $\Psi_{1,t}=\e^{-\i tH^{(1)}(\xi)}\Psi_1$ and $\Psi_{2,t}=\e^{-\i t H(\xi)}\Psi_2$ for some arbitrary vectors $\Psi_1\in\Fock\otimes \Fock^{(1)}$,
$\Psi_2\in\Fock$.  Due to the
support properties of $\ti j_{\infty}$ and the fact that $\supp\, j_{\infty}'\subset (-\infty, \cm)\backslash [-\eps,\eps]$, (since $j_0=1$ on $[-\eps,\eps]$
and $j_0^2+j_{\infty}^2=1$), we can apply 
Proposition~\ref{main-propagation-estimate} to show integrability of (\ref{wave-operators-finite}).

Thus we obtained that both $t\to \Phi(t)$ and $t\to \Phi(t)^*$ converge strongly. Now the result follows
by an application of Lemma~\ref{Helffer}, which also gives the intertwining property.  \end{proof}

\section{\Localized wave operators}\label{Geometric-WO}

\setcounter{equation}{0}
In this section we construct \localized wave operators, defined on a small neighbourhood $\mco$ of any point 
$(\xi_0,\la_0)\in \mcE^{(1)} \backslash (\mcT^{(1)}\cup\Exc\cup \Sigma_{\pp} )$. 
The adjective `\localized',   used to  describe the wave operators constructed in this section, requires a brief
clarification: On the one hand  it alludes  to their
construction in an energy-momentum spectral subspace of the small set  $\mco$. On the other hand it refers  to the 
Sigal-Soffer type localization
onto  a spectral subspace, constructed using the one-body propagation observable $\ta_{\xi_0}$ (cf. expression~(\ref{tilde-a}) below) and describing classically permitted scattering configurations. The fact that these \localized wave operators turn out to coincide with the conventional wave operators is due to the Mourre estimate preventing scattering states from occupying classically forbidden configurations in the large time limit.
\bed\label{small-neighbourhood} We set $\mcoz=\Nz\times\mcJz$, where $\Nz$ and $\mcJz$ appeared in Theorem~\ref{Mourre-estimate-theorem}
and choose an open bounded neighbourhood $\mco$ of $(\xi_0,\la_0)$, whose closure is contained in the interior of $\mcoz$.
\eed
\noindent We recall from Theorem~\ref{Mourre-estimate-theorem} that with the set $\mcoz$ we can associate 
the observable $a_{\xi_0}=\h\{v_{\xi_0}\cdot \i\nabla_k+\i\nabla_{k}\cdot v_{\xi_0} \}$ which enters into the Mourre estimates. 
We define the following counterpart of this observable 
\beqa
\ta_{\xi_0}:=\h\{(1\otimes v_{\xi_0})\cdot z+z\cdot (1\otimes v_{\xi_0}) \}, \label{tilde-a}
\eeqa
where $z:=1\otimes x-y \otimes 1$ on $\K\otimes\mfh$ is the relative distance between the electron and the boson and we set $x:=\i\nabla_k$. 
In the remaining part of the section we will set $v:=v_{\xi_0}$, $a:=a_{\xi_0}$ and $\ta:=\ta_{\xi_0}$,
unless stated otherwise.

\begin{remark}\label{Rem-ExtSecQuant} We will make use of an 
extension of the expression $\Gamma(q)$ to contractions $q$ on 
$\K\otimes \mfh$, which was discussed in \cite[Remark~1.1]{MR12}.
 We leave it to the reader to check that this remark applies, whenever we meet second quantization in the extended sense discussed here.
 Furthermore we will also need to work with such operators $q$ viewed as acting in $\K\otimes\Fock\otimes\mfh$, but skipping over the middle $\Fock$-component. This is what is meant $q_\infty(\ta/t)$
in Theorem~\ref{final-wave-operators}.

Finally, we warn the reader that we will be abusing notation, in particular in  Proposition~\ref{more-asymptotic-obs}, by writing  $1 \otimes \Ga(p_{\de}(\ta/t))$,
for the operator $\mcE (\Gamma(p_\de(\ta/t)) \otimes 1) \mcE$, where $\mcE\colon \hil^\ex \to\hil^\ex$ is the exchange operator defined on simple tensors by $\mcE(f\otimes\eta\otimes \eta') = f \otimes\eta'\otimes\eta$. 
\end{remark}

Before we proceed to the construction of asymptotic objects in $\mco$,  we need one more definition:
\bed\label{delta-definition} Let $0<c_0<\epsz$, where $\epsz$ appeared in Proposition~\ref{minimal-velocity-estimate}. 
Let $q\in C^{\infty}(\real)$ be s.t.  $0\leq q\leq 1$, $q(s)=1$ for $s\leq c_0/2$ and $q(s)=0$ for 
$s>c_0$. Furthermore, $q$ is a non-increasing function.
We write $q_{\de}(s)=q(s/\de)$ and $q_{\de}^{t}(s)=q_{\de}(s/t)$ for $0<\de\leq 1$.
\eed
\bep\label{asymp-obs-del-lim} Let $q$ be  as specified in Definition~\ref{delta-definition}.
Then the following strong limit exists
\begin{align}
Q_{\de}^+(H)& := \slim_{t\to \infty} \e^{\i tH} \Ga(q_{\de}(\ta/t))\e^{-\i tH}E(\mco\cup \Sigma_{\iso}), \label{asymptotic-observable-approx}
\end{align}
and equals $E(\Sigma_{\iso})$. In particular, $Q^+(H):=Q_{\de}^+(H)$ is independent of $\de$. 
\eep

\begin{proof} Let $\chi\in C_0^{\infty}(\real^{\nu+1})_{\real}$ be equal to one on $\mco$ and be supported in $\mcoz$. 
Then $E(\mco)=\chi(P,H)E(\mco)$ and we can write
\begin{align}
Q_{\de}^t(H)E(\mco)& = \e^{\i tH} \Ga(q_{\de}(\ta/t))\e^{-\i tH}\chi(P,H)E(\mco)\non\\
& = I_{\mathrm{LLP}}^*\Bigl(\int^{\oplus}_{\real^{\nu}} d\xi\, \e^{\i tH(\xi)} \Ga(q_{\de}(a/t))\e^{-\i tH(\xi)}\chi(\xi,H(\xi))\Bigr) I_{\mathrm{LLP}}E(\mco),
\label{time-approx-asympt-obs}
\end{align}
where we denoted by $Q_{\de}^t(H)$ the approximants on the r.h.s. of (\ref{asymptotic-observable-approx}).
It follows from Theorem~\ref{asymptotic-observable}, the dominated convergence theorem and the properties 
of $q$ specified in Definition~\ref{delta-definition}, that this expression converges to zero as $t\to\infty$.

Let us consider now $Q_{\de}^t(H)E(\Sigma_{\iso})$. 
We recall from Section~\ref{isolated-spectrum} that $\Sigma_{\iso}$ is a union  
of graphs of at most countably many analytic functions $p\colon N\to\real$, where
$N\subset\real^\nu$ are open sets. Let $\mathcal G$ be a graph of
one of these functions. Then we obtain
\beqa
Q_{\de}^t(H)E(\mathcal G)\Psi=I_{LLP}^*\int_{N}^{\oplus} d\xi\, \e^{\i t(H(\xi)-p(\xi))}\Ga(q_{\de}(a/t))
 \Psi_{\xi},
\eeqa
where $\real^{\nu}\ni \xi\to \Psi_{\xi}\in\Fock$ is a square-integrable Borel function representing $\Psi$. Now by 
the dominated convergence theorem $\lim_{t\to\infty}Q_{\de}^t(H) E(\mathcal G)\Psi=E(\mathcal G)\Psi$. \end{proof}

\bep\label{more-asymptotic-obs} Let $q$ and $1-p$ be as specified in Definition~\ref{delta-definition}. 
Then the following strong limits exist 
\begin{align}
Q_{\de}^+(H^{\ex})_0 & := \slim_{t\to\infty}\e^{\i tH^{\ex}}\bigl(\Ga(q_{\de}(\ta/t))\otimes 1)\bigr)\e^{-\i tH^{\ex}}E^{\ex}(\mco),\label{left-obs} \\
Q_{\de}^+(H^{\ex})_{\infty} & := \slim_{t\to\infty}\e^{\i tH^{\ex}}\bigl( 1\otimes \Ga(p_{\de}(\ta/t))\bigr)\e^{-\i tH^{\ex}}E^{\ex}(\mco), \label{right-obs}\\
Q_{\de}^+(H^{\ex}) & := \slim_{t\to\infty}\e^{\i tH^{\ex}}\bigl(\Ga(q_{\de}(\ta/t))\otimes \Ga(p_{\de}(\ta/t))\bigr)\e^{-\i tH^{\ex}}E^{\ex}(\mco),
\label{extended-obs}
\end{align}
and are independent of $\de$ (thus we can omit the subscript $\de$). Moreover there holds
\begin{align}
Q^+(H^{\ex})_0& = (E(\Sigma_{\iso})\otimes 1)E^{\ex}(\mco),\label{left-relation}\\
Q^+(H^{\ex})_{\infty}& = E^{\ex}(\mco), \label{right-relation} \\
Q^+(H^{\ex})& = (E(\Sigma_{\iso})\otimes 1)E^{\ex}(\mco). \label{extended-relation}
\end{align}
\eep
\begin{proof}  Let $\chi\in C_0^{\infty}(\real^{\nu+1})_{\real}$ be equal to one on $\mco$ and be supported in $\mcoz$. 
Then $E^{\ex}(\mco)=\chi(P^{\ex},H^{\ex})E^{\ex}(\mco)$ and we can write
\begin{align}
Q_{\de}^+(H^{\ex})_0 & = \slim_{t\to\infty}\e^{\i tH^{\ex}}\bigl(\Ga(q_{\de}(\ta /t))\otimes 1 \bigr)\e^{-\i tH^{\ex}}E^{\ex}(\mco)\non\\
& = \slim_{t\to\infty} \bigl(\e^{\i tH }\Ga(q_{\de}(\ta /t))\e^{-\i tH }E(\Sigma_{\iso}) \otimes 1\bigr)E^{\ex}(\mco)\non\\
& \quad +\slim_{t\to\infty} \bigl(\e^{\i tH }\Ga(q_{\de}(\ta /t))\e^{-\i tH }E(\mco) \otimes \vac\tvac \bigr)E^{\ex}(\mco)\non\\
& = \bigl(E(\Sigma_{\iso}) \otimes 1\bigr)E^{\ex}(\mco),
\end{align}
where in the first step we made use of  Lemma~\ref{energetic-considerations} and in the second step of Proposition~\ref{asymp-obs-del-lim}
to obtain the existence of the limit. This proves (\ref{left-obs}) and (\ref{left-relation}).

Making use of  Theorem~\ref{asymptotic-observable-two}, we obtain that there exists the limit
\begin{align}\label{infinite-component}
& Q_{\de}^+(H^{\ex})_\infty=\slim_{t\to\infty}\e^{\i tH^{\ex}}\bigl(1 \otimes \Ga(p_{\de}(\ta/t)) \bigr)\e^{-\i tH^{\ex}}E^{\ex}(\mco)\\
& \quad  = \slim_{t\to\infty}I^{\ex *}_{\mathrm{LLP}}\int^{\oplus}_{\real^{\nu}} d\xi\, \e^{\i tH^{\ex}(\xi) }\bigl(1\otimes \Ga(p_{\de}(a/t))\bigr)
\e^{-\i tH^{\ex}(\xi) }\chi(\xi,H^{\ex}(\xi)) I^{\ex}_{\mathrm{LLP}}   E^{\ex}(\mco)\non
\end{align}
which equals $E^{\ex}(\mco)$.
This proves (\ref{right-obs}) and (\ref{right-relation}).

Existence of the limit (\ref{extended-obs}) and relation~(\ref{extended-relation}) are obvious consequences of the facts proven above. 
\end{proof}
In the following theorem we construct the \localized wave operator $W^{+}_{\mco}$, associated with the
region $\mco$ specified in Definition~\ref{small-neighbourhood}. We also show that its adjoint is a strong
limit of its approximating sequence.
\bet\label{final-wave-operators} Let $j_0,j_{\infty}$ be as specified in Definition~\ref{j-definition-a},  s.t. $j_0^2+j_{\infty}^2=1$ and, in addition, 
let  $j_0$ and $1-j_{\infty}$ satisfy the conditions from Definition~\ref{delta-definition}. Let $q=(q_0,q_{\infty}):=(j_0^2,j_{\infty}^2)$ 
(in particular $q_0+q_{\infty}=1$). Then the   following strong limits exist: 
\begin{align}
W^{+}_{\mco,\de}& = \slim_{t\to\infty}\e^{\i tH}\cGa(q_{\de}(\ta/t))^*\e^{-\i tH^{\ex} } E^{\ex}(\mco),\\
W^{+*}_{\mco,\de}& = \slim_{t\to\infty}\e^{\i tH^{\ex}}\cGa(q_{\de}(\ta/t))\e^{-\i tH} E(\mco),
\end{align}
and intertwine $\chi(P,H)$ with $\chi(P^{\ex}, H^{\ex})$ for any $\chi\in C_0^{\infty}(\real^{\nu+1})_{\real}$. (Consequently,
$W^{+*}_{\mco,\de}$ is the adjoint of $W^{+}_{\mco,\de}$).
Moreover, these limits are independent of $\de$, for sufficiently small $\de$, thus we can omit the
subscript $\de$.
\eet
\begin{proof}  Let $\chi\in C_0^{\infty}(\real^{\nu+1})_{\real}$ be equal to one on $\mco$ and be supported in $\mcoz$. We write
\begin{align}
W^{+}_{\mco,\de}& =\slim_{t\to\infty}I_{\mathrm{LLP}}^*\int^{\oplus}_{\real^{\nu}} d\xi\, \e^{\i tH(\xi)}\cGa(q_{\de}(a/t))^*\e^{-\i tH^{\ex}(\xi) }
\chi(\xi,H^{\ex}(\xi)) I^{\ex}_{\mathrm{LLP}} E^{\ex}(\mco),\\
W^{+*}_{\mco,\de}& =\slim_{t\to\infty}I^{\ex *}_{\mathrm{LLP}}\int^{\oplus}_{\real^{\nu}} d\xi\, \e^{\i tH(\xi)}\cGa(q_{\de}(a/t))\e^{-\i tH^{\ex}(\xi) }
\chi(\xi,H(\xi)) I_{\mathrm{LLP}} E(\mco).
\end{align}
The existence of  these limits and the intertwining property follows from Theorem~\ref{wave-operators-theorem} by the dominated convergence theorem.

Let us now show that $W^{+}_{\mco,\de}$ is independent of $\de$ for sufficiently small $\de$. First, we note that 
\beqa
W^+_{\mco,\de}=W^+_{\mco,\de}(Q_{\eps}^+(H)\otimes 1) \label{WQ-relation}
\eeqa
for $\eps/2> \de$. However, by Proposition~\ref{asymp-obs-del-lim}, $Q_{\eps}^+(H)=E(\Sigma_{\iso})$ is independent of
$\eps$, hence the relation holds also for $\eps/2<\de$. Let us make $\eps$ even smaller, so as to ensure that $\eps/2<\de/4$. Then we can write
\beqa
W^+_{\mco,\de}(Q_{\eps}^+(H)\otimes 1)=W^+_{\mco,\de}Q_{\eps}^+(H^{\ex})=W^+_{\mco,\eps}Q_{\de}^+(H^{\ex}). \label{WQQex-relation}
\eeqa
But the r.h.s. above is independent of $\de$ by Proposition~\ref{more-asymptotic-obs}. Thus both $W^+_{\mco,\de}$ and $W^{+*}_{\mco,\de}$
are independent of $\de$. 
\end{proof}
Now we proceed to  the proof of an isometry property of  $W^{+*}_{\mco}$. An important role in the
proof is played by the map $\cGa(1,1)$ whose adjoint is the  scattering identification operator  from \cite{DGe1,HuSp2}.  
\bet\label{isometric-property} The \localized wave operators, defined as in Theorem~\ref{final-wave-operators}, satisfy
\beqa
W^{+}_{\mco}W^{+*}_{\mco}=E(\mco).
\eeqa
\eet
\begin{proof} Let $q_{\de}:=(q_{0,\de}(\ta/t), q_{\infty,\de}(\ta/t))$ be as specified in Theorem~\ref{final-wave-operators} and abbreviate
$\un q_{\de}:=\diag(q_{0,\de}(\ta/t), q_{\infty,\de}(\ta/t))$, the corresponding family of observables on $\K\otimes(\mfh\oplus\mfh)$.  
We set $W^+:=W^+_{\mco}$ and write
\begin{align}
W^{+}W^{+*}=W^{+}_{\de}W^{+*}& = \slim_{t\to\infty}\e^{\i tH}\cGa(q_{\de})^*\e^{-\i tH^{\ex}}E^{\ex}(\mco)W^{+*}\non\\
&= \slim_{t\to\infty}\e^{\i tH}\cGa(q_{\de})^*\e^{-\i tH^{\ex}}W^{+*}\non\\
&= \slim_{t\to\infty}\e^{\i tH}\cGa(1,1)^*\Ga^{\ex}( \un q_{\de} )\e^{-\i tH^{\ex}}W^{+*}.
\end{align}
Making use of the fact that $\mco$ is localized below the two-boson threshold, we obtain that $W^{+*}=\tP W^{+*}$, 
where $\tP:=1\otimes (P_{0}+P_{1})$ acts on $\hil^{\ex}=\K\otimes \Fock^{\ex}$ and $P_{n}: \Fock^{\ex}\to \Fock\otimes \Fock^{(n)}$
are natural restriction maps.
Let  $\Psi$ be a vector in the range of
$E(\mco)$. Then, setting $\Rs:=(\i+H)^{-1}$ and $\Rs^{\ex}:=(\i+H^{\ex})^{-1}$ we write
\begin{align}
& \e^{\i tH}\cGa(1,1)^* \Ga^{\ex}( \un q_{\de} )\e^{-\i tH^{\ex}}W^{+*}\Psi\non\\
& \qquad=\e^{\i tH}\cGa(1,1)^*\tP \Ga^{\ex}( \un q_{\de} )   (\Rs^{\ex})^{2}\e^{-\i tH^{\ex}}   W^{+*}\Rs^{-2}\Psi\non\\
& \qquad=\e^{\i tH}\cGa(1,1)^*\tP\Rs^{\ex}[H^{\ex}, \Ga^{\ex}( \un q_{\de} )] (\Rs^{\ex})^{3} \e^{-\i tH^{\ex}}   W^{+*}\Rs^{-3}\Psi\non\\
& \qquad\quad+\e^{\i tH}\cGa(1,1)^*\tP\Rs^{\ex}\e^{-\i tH^{\ex}} \e^{\i tH^{\ex}}  \Ga^{\ex}( \un q_{\de} ) \e^{-\i tH^{\ex}}   W^{+*}\Rs^{-1}\Psi.
\label{local-asymp-comp-formula}
\end{align}
By Proposition~\ref{comm-H-Ga} and property~\ref{Gamma-one-one},  the  term involving the commutator above is of order $O(t^{-1})$.
As for the second term, we note that by Proposition~\ref{more-asymptotic-obs}, the fact that $Q^+(H)=E(\Sigma_{\iso})$ 
and property~(\ref{Gamma-one-one}), 
\begin{equation}
\slim_{t\to\infty}\e^{\i tH}\cGa(1,1)^*\tP\Rs^{\ex}\e^{-\i tH^{\ex}} 
\bigl(\e^{\i tH^{\ex}} \Ga^{\ex}( \un q_{\de} )\e^{-\i tH^{\ex}}-Q^+(H)\otimes 1\bigr)W^{+*}\Rs^{-1}\Psi=0.
\end{equation}
Since $(Q^+(H)\otimes 1\big)W^{+*}=W^{+*}$, (cf.~formula~(\ref{WQ-relation}) above), we  obtain that the r.h.s. of (\ref{local-asymp-comp-formula}) equals (up to terms that tend to zero in the limit $t\to\infty$)
\beqa
\e^{\i tH}\cGa(1,1)^*\tP\Rs^{\ex}\e^{-\i tH^{\ex}}W^{+*} \Rs^{-1}\Psi.
\eeqa
This is asymptotic in the limit $t\to\infty$  to the expression 
\begin{align}
& \e^{\i tH}\cGa(1,1)^*\tP\Rs^{\ex} \cGa(q_{\de})\e^{-\i tH} \Rs^{-1}\Psi\non\\
& \qquad=\e^{\i tH}\cGa(1,1)^*\tP \Rs^{\ex} \Ga^{\ex}(\un j_{\de})\chi^{\ex}   \cGa(j_{\de})\e^{-\i tH}  \Rs^{-1} \Psi+O(t^{-1})\non\\
& \qquad=\e^{\i tH}\cGa(1,1)^*\tP \Rs^{\ex} [\Ga^{\ex}(\un j_{\de}), H^{\ex}]\Rs^{\ex} \chi^{\ex} \cGa(j_{\de})
\e^{-\i tH} \Rs^{-1}  \Psi\non\\
& \qquad\quad+\e^{\i tH}\cGa(j_{\de})^* \Rs^{\ex} \chi^{\ex} \cGa(j_{\de})\e^{-\i tH}  \Rs^{-1} \Psi+O(t^{-1})=\Psi+O(t^{-1}). \label{isometry-proof}
\end{align}
Here in the first step we used the identity
$\cGa(q_\delta) = \Ga^\ex(\un j_\delta) \cGa(j_\delta)$
and introduced a function $\chi\in C_0^{\infty}(\real^{\nu+1})_{\real}$, supported in $\mcoz$ and equal to one on $\mco$
so that $\chi(P,H)\Psi=\Psi$.
Making use of Lemma~\ref{Helffer}, we got
\beqa
\cGa(j_{\de})\chi=\chi^{\ex}\cGa(j_{\de})+O(t^{-1}),
\eeqa
where we set $\chi:=\chi(P,H)$ and $\chi^{\ex}:=\chi(P^{\ex},H^{\ex})$.  In the second step we commuted $\Ga^{\ex}(\un j_{\de})$
to the left and used the fact that $\tP \chi(P^{\ex},H^{\ex})=\chi(P^{\ex},H^{\ex})$.
In the third step we exploited  Proposition~\ref{comm-H-Ga}  to show that the resulting commutator is $O(t^{-1})$. In the last step we made use  again of Lemma~\ref{Helffer}. 
\end{proof}

\bel\label{Mourre-lemma}  Let $(\mcA,S)$ be  an analytic shell in $\Sigma_{\iso}$ (cf. Section~\ref{isolated-spectrum}). Let $\De$ be 
a Borel subset of the graph $\mcG_S$ of this shell. 
Then the  operator 
\beqa
B:= (\mathbf{1}_{\De}(P,H)\otimes v_{\xi_0})\cdot (1\otimes\nabla\om-\nabla S(P)\otimes 1)
\eeqa
satisfies
\beqa
\hat P B\hat P\geq \cm \hat P, \label{Mourre-inequality}
\eeqa
where $\hat P:=\mathbf{1}_{\mcoz}(P^{(1)},H^{(1)})(\mathbf{1}_{\De}(P,H) \otimes 1)$. We recall that $\mcoz$ appeared
in Definition~\ref{small-neighbourhood}.  
\eel
\begin{proof} We set $v:=v_{\xi_0}$ and $a:=a_{\xi_0}$ as abbreviated already  below definition (\ref{tilde-a}) above.
Let $\chi\in C_0^{\infty}(\real)_{\real}$ be supported in $\mcJz$. 
We recall from Theorem~\ref{Mourre-estimate-theorem} that for $\xi\in\Nz$ (where $\Nz$ appeared in Theorem~\ref{Mourre-estimate-theorem})
\beqa
\chi(H^{(1)}(\xi))\i [H^{(1)}(\xi),  1\otimes a]\chi(H^{(1)}(\xi))\geq \cm \chi(H^{(1)}(\xi))^2.
\eeqa
Let us set $\ti{\mathbf{1}}_{\De}(\xi):=\int^{\oplus}dk\,\mathbf{1}_{\De}(\xi-k,H(\xi-k))$, note that
\begin{align}
I_{\LLP}^{(1)}\mathbf{1}_{N_0}(P^{(1)})(\mathbf{1}_{\De}(P,H)\otimes 1)   I_{\LLP}^{(1)*}=\int^{\oplus}_{N_0}d\xi\, \ti{\mathbf{1}}_{\De}(\xi) \label{I-llp-bf-one}
\end{align}
and set
$\De_{\xi}:=\{\,k\in\real^{\nu}\,|\, (\xi-k,\la)\in\De, \textrm{  for some } \la\in\real \,    \}$. Then we get
\begin{align}
& \ti{\mathbf{1}}_{\De}(\xi)\Bigl(\int^{\oplus}_{\De_{\xi}}dk\,  \chi(H^{(1)}(\xi;k))\i [S^{(1)}(\xi;k),  1\otimes a]\chi(H^{(1)}(\xi;k))\Bigr) \ti{\mathbf{1}}_{\De}(\xi) \non\\
& \qquad \geq \cm \chi(H^{(1)}(\xi))^2\ti{\mathbf{1}}_{\De}(\xi),
\end{align}
and consequently
\begin{align}
&\ti{\mathbf{1}}_{\De}(\xi)\Bigl(\int^{\oplus}_{\De_{\xi}}dk\,  \chi(H^{(1)}(\xi;k)) v(k)\cdot \nabla_k S^{(1)}(\xi;k) \chi(H^{(1)}(\xi;k))\Bigr) \ti{\mathbf{1}}_{\De}(\xi) \non\\
&\qquad =\chi(H^{(1)}(\xi))^2\ti{\mathbf{1}}_{\De}(\xi)  \int^{\oplus}_{\De_{\xi}}dk\, v(k)\cdot \nabla_k S^{(1)}(\xi;k) 
\geq \cm \chi^2(H^{(1)}(\xi))\ti{\mathbf{1}}_{\De}(\xi),
\end{align}
where we set $S^{(1)}(\xi;k)=S(\xi-k)+\om(k)$. By approximating with functions $\chi$ the characteristic function of $\mcJz$,
taking the direct integral of both sides over $\xi\in \Nz$, and conjugating with $I_{\LLP}^{(1)}$, we obtain
\begin{equation}
(\mathbf{1}_{\De}(P,H)  \otimes v)\cdot (1\otimes\nabla\om-\nabla S(P)\otimes 1) \hat P \geq \cm\hat P. 
\end{equation}
where we made use of (\ref{I-llp-bf-one}). This concludes the proof. \end{proof}
\bet\label{isometric-property-one} The \localized wave operators, defined as in Theorem~\ref{final-wave-operators}, satisfy
\beqa
W^{+*}_{\mco}W^{+}_{\mco}=E^{\ex}(\mco)(E(\Sigma_{\iso})\otimes 1). \label{W-isometry}
\eeqa
\eet
\begin{proof} Let us set $W^+:=W^+_{\mco}$ and let $\Psi\in \Ran\, E^{\ex}(\mco)$. 
By Lemma~\ref{energetic-considerations}, $\Psi$ belongs to the closed span of vectors of  two types. The first type are vectors of the form
\beqa
 \Psi_1\otimes \Om, \label{vectors-of-first-type}
\eeqa
where $\Psi_1\in E(\mco)\hil$. Such vectors are elements of the kernel of $W^+$ due to the fact that
\begin{align}
W^{+}(\Psi_1\otimes \Om)& = \lim_{t\to\infty}\e^{\i tH}\cGa(q_{\de}(\ta/t))^*(\e^{-\i tH}\Psi_1\otimes\Om)\non\\
& = \lim_{t\to\infty}\e^{\i tH}\Ga(q_{0,\de}(\ta/t))\e^{-\i tH}\Psi_1=Q^+(H)E(\mco)\Psi_1=0, 
\end{align}
where we made use of  Proposition~\ref{asymp-obs-del-lim}. This proves relation~(\ref{W-isometry}) on vectors of type~(\ref{vectors-of-first-type}).  Vectors of the second type that span $\Ran\, E^{\ex}(\mco)$, provided by Lemma~\ref{energetic-considerations}, have the form
\beqa
\Psi_2\otimes a^*(h)\Om, \label{vectors-of-second-type}
\eeqa
where $h\in C_0^{\infty}(\real^{\nu})$ and $\Psi_2\in E(\De)\hil$  are s.t. $\De\subset\Sigma_{\iso}$ is a bounded Borel set  and $\De+(k,\om(k))\subset \mco$ for all $k\in \supp\, h$. 
For such vectors  we obtain
\begin{align}
W^{+}(\Psi_2\otimes a^*(h)\Om)& = \lim_{t\to\infty}\e^{\i tH}\cGa(q_{\de}(\ta/t))^*(\e^{-\i tH}\Psi_2\otimes a^*(h_t)\Om)\non\\
& = \lim_{t\to\infty}\e^{\i tH} a^*(q_{\infty,\de}(\ta/t)h_t)\e^{-\i tH}\e^{\i tH}\Ga(q_{0,\de}(\ta/t))\e^{-\i tH}\Psi_2\non\\
& = \lim_{t\to\infty}\e^{\i tH} a^*(q_{\infty,\de}(\ta/t)h_t)\e^{-\i tH}Q^{+}(H)\Psi_2\non\\
& = \lim_{t\to\infty}\e^{\i tH} a^*(q_{\infty,\de}(\ta/t)h_t)\e^{-\i tH}\Psi_2, \label{vectors-of-second-kind-computation}
\end{align}
where $h_t=\e^{-\i\om t}h$ and in the last step we made use of  Proposition~\ref{asymp-obs-del-lim} and the fact that $\Psi_2$ belongs to the range of $E(\Sigma_{\iso})$.
In view of the discussion of the isolated spectrum in Section~\ref{isolated-spectrum}, we can assume that $\Psi_2$ belongs to the
range of $E(\De)$, where $\De$ is a subset of the graph $\mcG_S$  of an analytic shell $(\mcA,S)$. Here we used that level crossings sit above a set of momenta,  a union of spheres, with zero Lebesgue measure.
Let $\haPsi_2\otimes a^*(\hah)\Om$ be  another vector of the 
form (\ref{vectors-of-second-type}), s.t. $\haPsi_2$ belongs to the range of $E(\haDe)$, where $\haDe$ is a subset of the graph of some other shell  $\haS$ which may, but does not have to coincide with $S$.
Now we obtain from (\ref{vectors-of-second-kind-computation})
\begin{align}
& \lan W^{+}(\haPsi_2\otimes a^*(\hah)\Om), W^{+}(\Psi_2\otimes a^*(h)\Om)\ran \non\\
& \qquad =\lim_{t\to\infty}\lan \haPsi_2,  \e^{\i tH}a(q_{\infty,\de}(\ta/t)\hah_t) a^*(q_{\infty,\de}(\ta/t)h_t)\e^{-\i tH}\Psi_2\ran.
\end{align}
By commuting the annihilation operator to the right, we get 
\begin{align}
& \lan\haPsi_2,  \e^{\i tH}a(q_{\infty,\de}(\ta/t)\hah_t) a^*(q_{\infty,\de}(\ta/t)h_t)\e^{-\i tH}\Psi_2\ran \non\\
& \qquad =\lan\haPsi_{2,t}\otimes \hah_t,  q_{\infty,\de}^2\big((1\otimes (v/2))\cdot (1\otimes x-y\otimes 1)/t+\mathrm{h.c.}\big) (\Psi_{2,t}\otimes h_t) \ran\non\\
& \qquad\quad+\lan a(q_{\infty,\de}(\ta/t)h_t)\e^{-\i t \haS(P)}\haPsi_2, a(q_{\infty,\de}(\ta/t)\hah_t)\e^{-\i tS(P)}\Psi_2\ran,
\label{asymptotic-vacuum}
\end{align}
where $\Psi_{2,t}=\e^{-\i t S(P)}\Psi_{2}$, $\haPsi_{2,t}=\e^{-\i t \haS(P)}\haPsi_{2}$. We recall that $v:=v_{\xi_0}$ and $x=\i\nabla_k$ is the position operator of the boson.

Let us first show that  the last term on the r.h.s. of (\ref{asymptotic-vacuum}) tends to zero: Due to the fact that $\haPsi_2\in E(\De)\hil$,
where  $\De$ is bounded, it is easy to see that
\beqa
\|(1+H_{\pho})a(q_{\infty,\de}(\ta/t)\hah_t)\e^{-\i t\haS(P)}\haPsi_2\|\leq c
\eeqa
for some $c$ independent of $t$. (See Lemmas~\ref{pseudo-lemma} and \ref{technical}). On the other hand, as shown 
in Proposition~\ref{asymp-obs-del-lim}, for $\de'$ sufficiently small,
\beqa
\Psi_2=\lim_{t\to\infty}\e^{\i tH}\Ga(q_{0,\de'}(\ta/t))\e^{-\i tH}\Psi_2.
\eeqa
Proceeding similarly as in the proof of \cite[Lemma~14]{FGSch3}, we get
\begin{align}
& (1+H_{\pho})^{-1} a(q_{\infty,\de}(\ta/t)h_t)\e^{-\i tH}\Psi_2\non\\
& \qquad = (1+H_{\pho})^{-1}  a(q_{\infty,\de}(\ta/t)h_t)\Ga(q_{0,\de'}(\ta/t))\e^{-\i tH}\Psi_2+o(1)\non\\
& \qquad =(1+H_{\pho})^{-1}\Ga(q_{0,\de'}(\ta/t))a((q_{0,\de'}q_{\infty,\de})(\ta/t)h_t)\e^{-\i tH}\Psi_2+o(1),
\end{align}
where $o(1)$ denotes a term which tends in norm to zero as $t\to\infty$. Here we made use of the fact that $(1+H_{\pho})^{-1} a(q_{\infty,\de}(\ta/t)h_t)$ is bounded, uniformly in $t$.  Noting that  for $\de'$ sufficiently small $q_{0,\de'}q_{\infty,\de}=0$, we obtain that the r.h.s. above tends to zero
and therefore the last term on the r.h.s. of (\ref{asymptotic-vacuum}) tends to zero.

Let us now consider the limit of the first term on the r.h.s. of (\ref{asymptotic-vacuum}):
\begin{align}
&  \lim_{t\to\infty}\lan\haPsi_{2,t}\otimes \hah_t,  q_{\infty,\de}^2\bigl( (1\otimes (v/2))\cdot (1\otimes x-y\otimes 1)/t+\mathrm{h.c.}  \bigr) (\Psi_{2,t}\otimes h_t) \ran\non\\
& \quad =\lim_{t\to\infty}\lan\haPsi_{2}\otimes \hah, q_{\infty,\de}^2\bigl( (1\otimes (v/2))\cdot\non\\
& \quad \quad\times\{ (1\otimes x-y\otimes 1)/t+1\otimes\nabla\om-\nabla \haS (P)\otimes 1 \}+ \mathrm{h.c.} \bigr)\e^{\i t (\haS(P)-S(P))} (\Psi_{2}\otimes h) \ran\non\\
& \quad =\lim_{t\to\infty}\lan\haPsi_{2}\otimes \hah,  q_{\infty,\de}^2\bigl((1\otimes v)\cdot(1\otimes\nabla\om-\nabla \haS(P)\otimes 1)  \bigr)\e^{\i t (\haS(P)-S(P))} (\Psi_{2}\otimes h) \ran.\non\\ \label{isometry-long-computation}
\end{align}
Here in the second step we made use of the strong resolvent convergence of the  sequence of operators in the argument of $q_{\infty,\de}^2$. 
Since $\hat S$ is only defined on a subset of $\real^{\nu}$, the symbol $\nabla \haS (P)$ is to be understood as $\nabla \haS_{\mathrm{r}} (P) $,
where $\haS_{\mathrm{r}}$ is the restriction of $\haS$
to the spectral support of the vector $\Psi_2$, which is then extended by zero to $\real^{\nu}$. 
Clearly, the last expression on the r.h.s. of (\ref{isometry-long-computation}) is equal to zero if $S\neq \haS$, since the argument of  $q_{\infty,\de}^2$ commutes with the spectral projection
$E(\mcG_S)\otimes 1$ and $E(\mcG_S)E(\mcG_{\haS})=0$ in this case. Thus we have verified (\ref{W-isometry}) for $S\neq \haS$ and we can assume that $S=\haS$. We can also assume that
both $\Psi_2$ and $\haPsi_2$ belong  to the range of $E(\De)$ for some bounded  Borel subset $\De\subset S$. (Again, since level crossings live on
a subset of momentum space with Lebesgue measure zero, we can exclude them from this discussion).
Then the last term on the r.h.s. of (\ref{isometry-long-computation}) equals
\begin{align}
& \lan\haPsi_{2}\otimes \hah,  q_{\infty,\de}^2\bigl((1\otimes v)\cdot(1\otimes\nabla\om-\nabla S(P)\otimes 1)\bigr) (\Psi_{2}\otimes h) \ran\non\\
& \qquad =\lan\haPsi_{2}\otimes \hah,  q_{\infty,\de}^2\bigl((\mathbf{1}_{\De}(P,H)\otimes v)\cdot(1\otimes\nabla\om-\nabla S(P)\otimes 1) \bigr) (\Psi_{2}\otimes h) 
\ran\quad\quad\quad \non\\
& \qquad =\lan\haPsi_{2}\otimes \hah,   q_{\infty,\de}^2(\hat P B \hat P)  (\Psi_{2}\otimes h) \ran, \label{auxiliary-computation-isometry}
\end{align}
where $B$ and $\hat P$ were defined in Lemma~\ref{Mourre-lemma}. 
Next, we note that
\beqa
\slim_{\de\to 0}q_{\infty,\de}^2(\hat P B \hat P)=\mathbf{1}_{ (0,\infty) }(\hat P B \hat P)=1-\mathbf{1}_{\{0\}}(\hat P B \hat P),
\eeqa 
where we exploited the fact that $\hat P B \hat P\geq 0$.  Next, we observe that
\beqa
\lan\haPsi_{2}\otimes \hah, \mathbf{1}_{\{0\}}(\hat P B \hat P)   (\Psi_{2}\otimes h) \ran=0. \label{zero-projection-vanishing}
\eeqa
In fact, if the l.h.s. above was different from zero, then $\Psi_0:=\mathbf{1}_{\{0\}}(\hat P B \hat P)   (\Psi_{2}\otimes h)\neq 0$.
Now inequality~(\ref{Mourre-inequality}) gives that $\hat P \Psi_0=0$. Since $\hat P$ commutes with $B$, and  $\Psi_{2}\otimes h$
belongs to the range of $\hat P$, we obtain that  $\Psi_0=\hat P \Psi_0=0$. This   contradiction  justifies (\ref{zero-projection-vanishing}).
Thus we obtain  that the r.h.s. of (\ref{auxiliary-computation-isometry})  satisfies
\beqa
\lim_{\de\to 0}\lan\haPsi_{2}\otimes \hah,   q_{\infty,\de}^2(\hat P B \hat P)  (\Psi_{2}\otimes h) \ran=\lan\haPsi_{2},  \Psi_{2}\ran \lan \hah,h\ran.
\eeqa
Summing up,  we have shown that
\beqa
\lan W^{+}(\haPsi_2\otimes a^*(\hah)\Om),\  W^{+}(\Psi_2\otimes a^*(h)\Om)\ran=\lan\haPsi_{2},\Psi_2\ran \lan \hah,h\ran, \label{second-norm}
\eeqa
which concludes the proof.
\end{proof}

\section{Wave operators and asymptotic completeness}\label{AC-section}
\setcounter{equation}{0}

Let us now proceed to the construction of  the conventional wave operators and to the proof of their completeness below the two-boson 
threshold. It will be convenient to work with  
 wave operators $\tOmeR^+$, introduced in (\ref{definition-of-wave-op}) below,  which are defined on the entire Hilbert 
space $\hil^{\ex}$. As we show in the proof of Theorem~\ref{main-theorem} given below, their restrictions to $E^{\ex}(\mcR)\hil_+$ 
coincide with   the wave operators $\OmeR^+$ defined in (\ref{intro-wave-operator}).
We construct them first in the
small regions $\mco$ of the energy-momentum spectrum, in which we constructed the \localized wave operators.
\bep\label{local-existence-of-wave-operators} Let $\mco$ be as specified in Definition~\ref{small-neighbourhood}. Then there exists the limit
\beqa
\tOme^+_{\mco}:=\slim_{t\to \infty} \e^{\i tH}\cGa(1,1)^*\e^{-\i tH^{\ex}} (E(\Sigma_{\iso})\otimes 1)E^{\ex}(\mco).
\eeqa
Moreover, $\tOme^+_{\mco}=W^{+}_{\mco}$. 
\eep
\begin{proof}  We set $W^+:=W^+_{\mco}$ and $\Rs^{\ex}:=(\i+H^{\ex})^{-1}$. Theorem~\ref{isometric-property-one} gives us that
$(E(\Sigma_{\iso})\otimes 1)E^{\ex}(\mco)=W^{+*}W^{+}$.
Thus we can write
\begin{align}
& \e^{\i tH}\cGa(1,1)^*\e^{-\i tH^{\ex}} (E(\Sigma_{\iso})\otimes 1)E^{\ex}(\mco)\non\\
& \qquad =\e^{\i tH}\cGa(1,1)^*\tP \Rs^{\ex} \e^{-\i tH^{\ex}}   W^{+*}W^{+} (\Rs^{\ex})^{-1} E^{\ex}(\mco),
\end{align}
where we use the notation $\tP=1\otimes (P_{0}+P_{1})$ from the proof of Theorem~\ref{isometric-property}. By property~(\ref{Gamma-one-one}), $\cGa(1,1)^*\tP(\i+H^{\ex})^{-1}$ is a bounded operator. Thus,
up to an error term which tends strongly to zero, the last expression equals
\begin{align}
& \e^{\i tH}\cGa(1,1)^*\tP  \Rs^{\ex} \cGa(q_{\de})\e^{-\i tH}  W^{+}(\Rs^{\ex})^{-1}E^{\ex}(\mco) \non\\
& \qquad =\e^{\i tH}\cGa(1,1)^*\tP  \Rs^{\ex} \Ga^{\ex}(\un j_{\de})\chi^{\ex}   \cGa(j_{\de})\e^{-\i tH} W^{+} (\Rs^{\ex})^{-1}  E^{\ex}(\mco)   +O(t^{-1})\non\\
& \qquad=\e^{\i tH}\cGa(1,1)^*\tP \Rs^{\ex}  [\Ga^{\ex}(\un j_{\de}), H^{\ex}] \Rs^{\ex}  \chi^{\ex} \cGa(j_{\de})
\e^{-\i tH}  
W^{+} (\Rs^{\ex})^{-1}E^{\ex}(\mco)   \non\\
& \qquad\quad +\e^{\i tH}\cGa(j_{\de})^* \Rs^{\ex} \chi^{\ex} \cGa(j_{\de})\e^{-\i tH}  W^{+} (\Rs^{\ex})^{-1} E^{\ex}(\mco) +O(t^{-1})\non\\
& \qquad=W^{+}+O(t^{-1}). \,\,\, \label{isometry-proof-one}
\end{align}
These steps are justified exactly as in the discussion after formula~(\ref{isometry-proof}) above, in particular $\chi\in C_0^{\infty}(\real^{\nu+1})_{\real}$
is a function supported in $\mcoz$ and equal to one in $\mco$ and we set $\chi^{\ex}:=\chi(P^{\ex},H^{\ex})$.
Relation~(\ref{isometry-proof-one})  proves the existence of $\tOme^+_{\mco}$ and the fact that $\tOme^+_{\mco}=W^+$. 
\end{proof}
In the following two theorems we state and prove our main results. We recall that $\mcR=\{\,(\xi, E)\in \real^{\nu+1}\,|\, E<\Sigma^{(2)}(\xi)\,   \}$.
\bet \label{isometry} There exists the wave operator $\tOmeR^+: \hil^{\ex}\to\hil$ given by
\beqa
\tOmeR^+:=\slim_{t\to \infty} \e^{\i tH}\cGa(1,1)^*\e^{-\i tH^{\ex}} (E(\Sigma_{\pp})\otimes 1)E^{\ex}(\mcR). \label{definition-of-wave-op}
\eeqa
 The wave operator satisfies
\beqa
\tOmeR^{+*}\tOmeR^+= (E(\Sigma_{\pp})\otimes 1)E^{\ex}(\mcR) \label{conv-wave-isometry}.
\eeqa
Moreover, for any $\chi\in C_0^{\infty}(\real^{\nu+1})_{\real}$,
\beqa
\tOmeR^+\chi(P^{\ex}, H^{\ex})=\chi(P, H)\tOmeR^+. \label{finite-intertwining}
\eeqa
\eet
\begin{proof} Let us recall that $\mcR\cap\Sigma=\mcE^{(1)}\cup\Sigma_{\iso}$ and the union is disjoint. Since the lower boundary of the joint spectrum of 
$(P^{(1)},H^{(1)})$ is $\xi\to\Sigma_0^{(1)}(\xi)$, we note that 
\beqa
E^{\ex}(\mcR)=E(\mcR) \oplus  E^{(1)}(\mcE^{(1)}), \label{decomposition-of-E}
\eeqa
where $E^{(1)}(\,\cdot\,)$ is the joint spectral resolution of $(P^{(1)}, H^{(1)})$. As for the first component, we obtain that
for any $\Psi\in\hil$
\beqa
\e^{\i tH}\cGa(1,1)^*\e^{-\i tH^{\ex}} (E(\Sigma_{\pp}\cap\mcR)\Psi \otimes \vac)=E(\Sigma_{\pp}\cap \mcR )\Psi, \label{trivial-action-wave}
\eeqa
thus $\tOmeR^+$ trivially exists and is an isometry on  this subspace. 

Let us now consider the second component of the direct sum (\ref{decomposition-of-E}). Let $K\subset \mcE^{(1)}$ be a compact set.
Let us  show that $\tOmeR^+$ exists on the range of $E^{(1)}(K)$. 
We set  $\mcT:=(\mcT^{(1)}\cup\Exc\cup \Sigma_{\pp})$ and pick a vector $\Psi\in E^{(1)}(K)(\hil\otimes \Fock^{(1)})$. Now we choose open sets $G_n\supset \mcT$ s.t.
\beqa
\lan \Psi, E^{(1)}(G_n\backslash \mcT) \Psi \ran\leq \fr{1}{n}. \label{measure-regularity-bound}
\eeqa
Such sets exist by the regularity of the spectral measure. We note that $K_n:=K\backslash G_n$ are compact sets. 
 By Proposition~\ref{local-existence-of-wave-operators}, for any
$(\xi_0,\la_0)\in K_n$ there exists a neighbourhood $\mco$,  specified in Definition~\ref{small-neighbourhood},  s.t. 
\beqa
\tOme^+_{\mco}=\tOmeR^+E^{(1)}(\mco)
\eeqa
exists. Such sets $\mco$ form a covering   of $K_n$ from which we can choose a finite sub-covering  $\{\mco_j\}_{j=1}^{N_0}$.
By taking intersections of the sets in this sub-covering, we can find a family of disjoint Borel sets $\{\bmco_{i}\}_{i=1}^{N}$, whose
union coincides with $K_n$ and  s.t. each
$\bmco_{i}$ is contained in some set $\mco_{j_i}$, as specified above. 
Thus we can write
\begin{align}
&  \e^{\i tH}\cGa(1,1)^*\e^{-\i tH^{\ex}} (E(\Sigma_{\pp})\otimes 1)E^{(1)}(K)\Psi\non\\
& \qquad=\e^{\i tH}\cGa(1,1)^*\e^{-\i tH^{\ex}} (E(\Sigma_{\pp})\otimes 1)E^{(1)}(K_n)\Psi+O(1/n)\non\\
& \qquad =\sum_{i=1}^{N}\e^{\i tH}\cGa(1,1)^*\e^{-\i tH^{\ex}} (E(\Sigma_{\pp})\otimes 1)E^{(1)}(\bmco_{i})\Psi+O(1/n). 
\label{convergence-of-Omeg}
\end{align}
In the first step above we made use of the relation
 \begin{align}
E^{(1)}(K) &= E^{(1)}(K_n)+E^{(1)}(G_n\cap K)\non\\
&= E^{(1)}(K_n)+E^{(1)}((G_n\backslash \mcT)\cap K)+E^{(1)}( \mcT\cap K). \label{spectral-measure-decomp}
\end{align}
The second term on the r.h.s above, together with the bound  (\ref{measure-regularity-bound}) and property~(\ref{Gamma-one-one}),
gives rise to the term  $O(1/n)$ on the r.h.s. of (\ref{convergence-of-Omeg}), i.e. a term whose norm is  bounded by $c/n$ for some $c$ independent
of $t$. (Here we exploit compactness of $K$). The last term on the r.h.s. of (\ref{spectral-measure-decomp}) is zero due to the relation
\beqa
E^{(1)}(\mcT\cap\mcE^{(1)})=I^{(1)*}_{\LLP}\Bigl(\int^{\oplus}d\xi\, E^{(1)}_{\xi}( \mcT(\xi)\cap \mcE^{(1)}(\xi)  )\Bigr)I^{(1)}_{\LLP},
\eeqa
the fact that the set $\mcT(\xi)\cap \mcE^{(1)}(\xi)$ is countable for any $\xi$ (Theorem~\ref{Mourre-estimate-theorem}) and therefore
 $E^{(1)}_{\xi}( \mcT(\xi)\cap \mcE^{(1)}(\xi)  )=E^{(1)}_{\xi}( \Sigma^{(1)}_{\pp}(\xi)\cap \mcE^{(1)}(\xi)  )$ which is equal to zero
except for $\xi$ from some set of zero Lebesgue measure   (Lemma~\ref{Extended-spectrum-lemma}).  
Now relation~(\ref{convergence-of-Omeg}) and Proposition~\ref{local-existence-of-wave-operators}
give the existence of $\tOme^+_K:=\tOmeR^+E^{(1)}(K)$ by  the Cauchy criterion. 

Let us now show that $\tOme^+_K$ is isometric on the range of $E^{(1)}(K)$. We obtain
from~(\ref{convergence-of-Omeg}) that
\beqa
\tOme^+_K\Psi=\sum_{i=1}^{N} W^+_{\mco_{j_i}} E^{(1)}(\bmco_{i})\Psi+O(1/n), \label{limiting-relation-convergence}
\eeqa
where we made use of Proposition~\ref{local-existence-of-wave-operators} to replace the conventional wave operators with the
\localized wave operators $W^+_{\mco_{j_i}}$. Recalling that the sets $\bmco_{i}$ are disjoint and 
the \localized wave operators intertwine $(P,H)$ with $(P^{\ex}, H^{\ex})$, we can write
\begin{align}\label{checking-isometry}
\non \|\tOme^+_K\Psi\|^2 & =\sum_{i=1}^{N} \|W^+_{\mco_{j_i}} E^{(1)}(\bmco_{i})\Psi\|^2\\
& \quad +2\mathrm{Re}\lan(\tOme^+_K\Psi-O(1/n)), O(1/n)\ran
+\lan  O(1/n), O(1/n) \ran. 
\end{align}
The first term on the r.h.s. above satisfies
\beqa
\sum_{i=1}^{N} \|W^+_{\mco_{j_i}} E^{(1)}(\bmco_{i})\Psi\|^2=\sum_{i=1}^{N}\lan \Psi,   E^{(1)}(\bmco_{i}) \Psi \ran=
\lan\Psi,  E^{(1)}(K_n) \Psi\ran=\|\Psi\|^2+O(1/n),  \label{local-isometry-application}
\eeqa
where in the first step we made use of Theorem~\ref{isometric-property-one}, Lemma~\ref{energetic-considerations} and the fact that 
$\bmco_{i}\subset \mco_{j_i}$. In the second step we used that the union of $B_n$ coincides with $K_n$ and in the last step we exploited
formula~(\ref{spectral-measure-decomp}) and the subsequent discussion.
The last two terms on the r.h.s. of (\ref{checking-isometry}) and the last term on the r.h.s. of (\ref{local-isometry-application}) can be made arbitrary small by taking $n$ sufficiently large. Thus we have shown that
\beqa
\|\tOme^+_K\Psi\|=\|\Psi\|,
\eeqa
i.e. $\tOme^+_K$ is isometric on the range of $E^{(1)}(K)$.

Now let $K^n\subset \mcE^{(1)}$ be an increasing family of compact sets s.t. $\bigcup_{n\geq 0} K^n=\mcE^{(1)}$. Then
\beqa
\mcD:=\bigcup_{n\geq 0} E^{(1)}(K^n)(\hil\otimes\Fock^{(1)})   
\eeqa
is a dense domain in $E(\mcE^{(1)})(\hil\otimes\Fock^{(1)})$. (Here we exploit the inner regularity of the spectral measure). It follows from our above
considerations that 
$\tOmeR^+$ is well defined on $\mcD$ and is an isometry on this domain. Thus $\tOmeR^+$ extends to an isometry  on  
$E(\mcE^{(1)})(\hil\otimes\Fock^{(1)})$. 

We conclude that $\tOmeR^+$, as defined in  (\ref{definition-of-wave-op}), exists. In view of relation~(\ref{trivial-action-wave}),  
to complete the proof of (\ref{conv-wave-isometry}), it suffices to show that for any 
$\Psi_0\in E(\Sigma_{\pp}\cap\mcR)\hil$ and $\Psi\in E^{(1)}(K)(\hil\otimes\Fock^{(1)})$ as specified above,   there holds
\beqa
\lan \Psi_0, \tOmeR^+\Psi\ran=0.
\eeqa
To this end we make use again of relation~(\ref{limiting-relation-convergence}) and of the intertwining property of the \localized
wave operators to write
\beqa
\lan \Psi_0, \tOmeR^+\Psi\ran=\sum_{i=1}^{N} \lan E(\bmco_{i})\Psi_0,  W^+_{\mco_{j_i}} E^{(1)}(\bmco_{i})\Psi\ran+\lan \Psi_0, O(1/n)\ran.
\eeqa
Now we note that $ E(\bmco_{i})\Psi_0=0$, since the sets $\bmco_{i}$ do not intersect with the point spectrum of $(P,H)$. The last term
on the r.h.s. above can be made arbitrarily small by choosing large $n$. This concludes the proof of (\ref{conv-wave-isometry}). 

Finally, let us show the intertwining property~(\ref{finite-intertwining}). In view of the decomposition~(\ref{decomposition-of-E}),
it suffices to check (\ref{finite-intertwining}) first on vectors of the form $\Psi_0\otimes \vac$, $\Psi_0\in E(\Sigma_{\pp}\cap\mcR)\hil$
and then on vectors $\Psi\in E^{(1)}(K)(\hil\otimes\Fock^{(1)})$  as specified above. In the first case we get
\begin{align}
& \e^{\i tH}\cGa(1,1)^*\e^{-\i tH^{\ex}} \chi(P^{\ex},H^{\ex})(\Psi_0\otimes \vac)=\e^{\i tH}\cGa(1,1)^*\e^{-\i tH} (\chi(P,H)\Psi_0\otimes \vac)\non\\
& \qquad =\chi(P,H)\Psi_0
=\chi(P,H)\e^{\i tH}\cGa(1,1)^*\e^{-\i tH^{\ex}}(\Psi_0\otimes \vac).
\end{align}
As for the vectors of the second type, we make use of relation~\eqref{limiting-relation-convergence}:
\begin{align}
\tOmeR^+\chi(P^{\ex},H^{\ex})\Psi & = \sum_{i=1}^{N} W^+_{\mco_{j_i}}\chi(P^{\ex},H^{\ex}) E^{(1)}(\bmco_{i})\Psi+O(1/n)\non\\
&= \chi(P,H)\sum_{i=1}^{N} W^+_{\mco_{j_i}} E^{(1)}(\bmco_{i})\Psi+O(1/n)\non\\
&= \chi(P,H)\tOmeR^+\Psi+O(1/n),
\end{align}
where we made use of the intertwining relation for the \localized wave operators shown in Theorem~\ref{final-wave-operators}.
Since $O(1/n)$ can be made arbitrarily small, this concludes the proof. \end{proof}
\bet\label{asymptotic-completeness-theorem} The wave operator $\tOmeR^+$, defined in (\ref{definition-of-wave-op}), satisfies 
\beqa
\mathrm{Ran} \,\tOmeR^+=E(\mcR)\hil. \label{asymptotic-completeness-relation}
\eeqa
\eet
\begin{proof} We recall that $\mcR\cap \Sigma=\Sigma_{\iso}\cup \mcE^{(1)}$. We note that for any $\Psi\in E(\Sigma_\iso)\hil = \hil_\iso$ there holds
\beqa
\Psi=\tOmeR^{+}(\Psi\otimes \vac), \label{eigenvectors-relation}
\eeqa
so  $\Ran\,E(\Sigma_{\iso})\subset \Ran \,\tOmeR^+$. 

Next, let us choose a compact set $K\subset \mcE^{(1)}$. 
We denote  $\mcT:=(\mcT^{(1)}\cup\Exc\cup \Sigma_{\pp})$ and choose a vector $\Psi\in E(K)\hil$.   We select open 
sets $G_n\supset \mcT$ s.t.
\beqa
\lan \Psi, E(G_n\backslash \mcT) \Psi \ran\leq \fr{1}{n}. \label{measure-regularity-bound-one}
\eeqa
Similarly as in the proof of Theorem~\ref{isometry}, such sets exist by the regularity of the spectral measure. We define sets $K_n:=K\backslash G_n$, which are  compact. 
Now we write
\beqa
E(K)=E(K_n)+E(G_n\cap K)=E(K_n)+E((G_n\backslash \mcT)\cap K)+E(\mcT\cap K). \label{measure-decomposition}
\eeqa
Due to relation (\ref{measure-regularity-bound-one}), the second term on the r.h.s. above satisfies
\beqa
\|E((G_n\backslash \mcT)\cap K)\Psi\|\leq \fr{1}{n}.
\eeqa
As for the last term on the r.h.s. of (\ref{measure-decomposition}), we note that
\beqa
E(\mcT\cap \mcE^{(1)})=I_{\LLP}^*\Bigl(\int^{\oplus} d\xi\, E_{\xi}(\mcT(\xi)\cap \mcE^{(1)}(\xi))\Bigr)I_{\LLP}=E(\Sigma_{\pp}\cap \mcE^{(1)}),
\eeqa
where we made use of the fact that, by Theorem~\ref{Mourre-estimate-theorem}, $\mcT(\xi)\cap \mcE^{(1)}(\xi)$ is countable for 
any $\xi$. Thus $E(\mcT\cap K)\Psi$ belongs to the range of $E(\Sigma_{\pp}\cap\mcE^{(1)})$ and hence belongs to the range
of the wave operator. In fact, for any $\Psi\in E(\Sigma_{\pp}\cap\mcE^{(1)})\hil$  we have
\beqa
\tOmeR^{+}(\Psi\otimes \vac)=\Psi,   \label{second-eigenvectors-relation}
\eeqa
similarly as in  formula~(\ref{eigenvectors-relation}) above. 

Let us now consider the first term on the r.h.s. of (\ref{measure-decomposition}). For any $(\xi_0,\la_0)\in K_n$ there
exists a neighbourhood $\mco$ as specified in Definition~\ref{small-neighbourhood}. Such neighbourhoods form a
covering of $K_n$ from which we can choose a finite sub-covering $\{\mco_j\}_{j=1}^{N_0}$. By taking intersections of these sets, if necessary, we can
find a finite family of disjoint Borel sets $\{B_i\}_{i=1}^N$, such that each $B_i\subset \mco_{j_i}$ for some $i\in\{1,\ldots, N_0\}$ and their union 
coincides with $K_n$. Thus, by Theorem~\ref{isometric-property}, and Proposition~\ref{local-existence-of-wave-operators}, we can write
\beqa
E(K_n)=\sum_{i=1}^N E(B_i)=\sum_{i=1}^N W_{\mco_{j_i}}^+W^{+*}_{\mco_{j_i}} E(B_i)=\sum_{i=1}^N \tOmeR^+W^{+*}_{\mco_{j_i}} E(B_i),
\eeqa
where $W^+_{\mco_{j_i}}$ are the \localized wave operators  and we made use of the fact that
they  intertwine $(P^{\ex},H^{\ex})$ with $(P,H)$. Thus the range of $E(K_n)$ is contained
in the range of $\tOmeR^{+}$.

Summing up, for any $\Psi\in E(K)\hil$ there exists a sequence of vectors $\Psi_n\in \Ran\,\tOmeR^{+}$ s.t.
\beqa
\|\Psi-\Psi_n\|\leq \fr{1}{n}.
\eeqa
Note that by construction $\tOmeR^+$ vanishes on $E^\ex(\mcR)\hil^\ex \ominus (E(\Sigma_\pp)\otimes 1)E^\ex(\mcR)\hil^\ex$. Hence, by relation~(\ref{conv-wave-isometry}),  
$\Ran\,\tOmeR^{+}$ is closed and we obtain that $\Psi\in \mathrm{Ran}\,\tOmeR^{+}$.
This completes the proof of the fact that $\mathrm{Ran} \,\tOmeR^+\supset E(\mcR)$. The opposite inclusion follows from the
intertwining relation~(\ref{finite-intertwining}). 
\end{proof}

We are now in a position to extract our main theorem from 
Subsection~\ref{results-subsection}, as well as its corollary.
In proofs we will make use of the identity
\beqa
E^{\ex}(\mcR)\hil_+= E^{\ex}(\mcR)(E(\Sigma_{\pp})\otimes 1) \hil^{\ex} = (E(\Sigma_{\pp})\otimes 1) E^{\ex}(\mcR) \hil^{\ex}, \label{extended-hilbert-space-decomp}
\eeqa 
which follows from the definition $\hil_+=(E(\Sigma_{\pp})\otimes 1)\hil^{\ex}$
of the outgoing Hilbert space.

\begin{proof}[Proof of Theorem~\ref{main-theorem}] 
By \eqref{extended-hilbert-space-decomp} and Theorem~\ref{isometry} we conclude the existence of $\OmeR^+=(\tOmeR^+)_{|E^{\ex}(\mcR)\hil_+}$ and the
property
\beqa
\OmeR^{+*}\OmeR^+= E^{\ex}(\mcR)_{|\hil_+}. \label{restricted-wave-operator-isometry}
\eeqa
By construction of $\tOmeR^+$, we observe that $\tOmeR^+\Psi = 0$, for $\Psi \in E^\ex(\mcR)\hil^\ex \ominus E^\ex(\mcR)\hil_+$. Hence,
 Theorem~\ref{asymptotic-completeness-theorem} gives  $\Ran\, \OmeR^+=E(\mcR)\hil$.
Together with (\ref{restricted-wave-operator-isometry})  this implies unitarity of $\OmeR^+\colon E^{\ex}(\mcR)\hil_+\to E(\mcR)\hil$. That is, 
$\OmeR^+\OmeR^{+*}=E(\mcR)$, which concludes the proof. 
\end{proof}

\begin{proof}[Proof of Corollary~\ref{HR-corollary}] To prove  part \ref{HR-Cor-a}, we recall that $a_+^*(h)\Psi:=\OmeR^+( \Psi \otimes a^*(h)\vac)$,
where $\Psi\in E(\mcR)\hil_{\bnd}$ and $h$ are $\mcR$-compatible. Now we compute, making use of (\ref{restricted-wave-operator-isometry}) 
\begin{align}
\lan a_+^*(h)\Psi,a_+^*(h')\Psi'\ran=     \lan \OmeR^+(\Psi\otimes a^*(h)\vac), \OmeR^+( \Psi'  \otimes a^*(h') \vac )\ran    
=\lan \Psi,\Psi'\ran \lan h,h'\ran.
\end{align}
Similarly, for  $\Psi''\in  E(\mcR)\hil_{\bnd}$,
\begin{align}
\lan a_+^*(h)\Psi, \Psi''\ran=\lan \OmeR^{+}(\Psi\otimes a^*(h)\vac),   \OmeR^{+}(\Psi''\otimes \vac)\ran=0,
\end{align}
where we made use of the fact that $\OmeR^{+}(\Psi''\otimes \vac)=\Psi''$ 
(cf. relation~(\ref{trivial-action-wave})) and of (\ref{restricted-wave-operator-isometry}). 

To prove part \ref{HR-Cor-b} of the corollary, we recall that $\Ran\, \OmeR^+=E(\mcR)\hil$ and therefore, any vector
$\Psi_1\in E(\mcR)\hil$ can be written as $\Psi_1=\OmeR^+\Psi_2$, where $\Psi_2\in E^{\ex}(\mcR)\hil_+$. By \eqref{PreFormOfAsymptoticStates} in Lemma~\ref{energetic-considerations}, applied with $O=\mcR$,  we find that
$E^{\ex}(\mcR)\hil^\ex = E(\mcR)\hil \oplus E^{(1)}(\mcR)(\hil_\iso\otimes\mfh)$. Since the second summand is already sitting inside $E^\ex(\mcR)\hil_+$, cf.~\eqref {extended-hilbert-space-decomp}, we find that
\begin{equation}
E^{\ex}(\mcR)\hil_+ = E(\mcR)\hil_\bnd \oplus E^{(1)}(\mcR)(\hil_\iso\otimes\mfh).
\end{equation}
The claim \ref{HR-Cor-b} now follows from \eqref{FormOfAsymptoticStates}, applied with $O = \mcR$.
\end{proof}

\appendix

\section{Fock space combinatorics}\label{Fock-combinatorics}

\setcounter{equation}{0}

\subsection{Fock space}\label{Fock-space-general}

Let $\mfh$ be the single-particle space and $\Ga(\mfh)$ be the symmetric Fock space over $\mfh$ 
given by
\beqa
\Ga(\mfh):=\bigoplus_{n\geq 0}\Ga^{(n)}(\mfh),  
\eeqa
where $\Ga^{(n)}(\mfh)=\mfh^{\otimes_s^{n}}$.  $\Ga^{(0)}(\mfh)$ is spanned by the  vacuum vector  denoted by $\vac$. 
(If the single-particle space $\mfh$ is fixed, we  use a shorter notation $\Fock:=\Ga(\mfh)$ and $\Fock^{(n)}:=\Ga^{(n)}(\mfh)$).
For any set $D\subset \mfh$ we set $\Ga^{(n)}(D)=D^{\otimes_s^{n}}$ and define $\Ga_{\fin}(D)$ as  
the space of finite linear combinations of vectors from $\Ga^{(n)}(D)$, $n=0,1,2\dotsc$

Let $D\subset \mfh$ be a dense domain and  $a\colon D\to \mfh$ a linear map.
Then $\dGa(a)$ is defined on $\Ga^{(n)}(D)$, $n\geq 1$, by
\beq
\dGa(a):=\sum_{i=1}^n \underbrace{1 \otimes\cdots\otimes 1}_{i-1}\otimes a \otimes \underbrace{1\otimes \cdots \otimes 1}_{n-i}. \label{dGamma-def}
\eeq
and extended to $\Ga_{\fin}(D)$ by linearity and the relation $\dGa(a)\vac=0$.  In particular,  $N:=\dGa(1)$ is called the number operator. 
We recall that if $a$ is closable then so is $\dGa(a)$.  Moreover, if $a$ is essentially self-adjoint on $D$, then $\dGa(a)$ is essentially self-adjoint
on $\Ga_{\fin}(D)$. Finally, if $b$ is a quadratic form on $D_1\times D_2$, where $D_1,D_2$ are dense domains in $\mfh$, then one can define $\dGa(b)$
as a quadratic form on $\Ga_{\fin}(D_1)\times \Ga_{\fin}(D_2)$.

Let $\mfh_1,\mfh_2$ be two single-particle spaces and let  $D_1\subset \mfh_1$ be a dense domain.  For any linear map $q: D_1\to\mfh_2$
we define a map $\Ga(q)$ on $\Ga^{(n)}(D_1)$, $n\geq 1$, by
\beq
\Ga(q):= \underbrace{q\otimes\cdots \otimes q}_n   
\eeq 
and extend it to $\Ga_{\fin}(D_1)$ by linearity and the relation $\Ga(q)\vac=\vac$. If $q$ is a contraction (i.e. $\|q\|\leq 1$) then
$\Ga(q)$ extends to a contraction $\Ga(\mfh_1)\to\Ga(\mfh_2)$. We  recall that for a contraction $q$ acting on $\K\otimes \mfh$ one
can define $\Ga(q)$ as a contraction on $\K\otimes \Ga(\mfh)$. 
See \cite[Remark~1.1]{MR12}.

Let $q,a_1,\ldots, a_m$ be operators $D_1\to \mfh_2$ defined on some
common dense domain $D_1\subset\mfh_1$. Then we define $\dGa(q,a_1,\ldots, a_m)$ on  $\Ga^{(n)}(D_1)$, $n \geq m$, by
\beqa
\dGa(q,a_1,\ldots, a_m):=\!\!\!\!\sum_{\su{i_1,\ldots,i_m \\ \forall_{k\neq l} i_k\neq i_l}}
(\underbrace{q\otimes\cdots\otimes q}_{i_1-1}\otimes a_1\otimes q\otimes \cdots\cdots\otimes q \otimes a_m
\otimes \underbrace{q\otimes \cdots\otimes  q}_{n-i_m} ), \label{new-object-def}
\eeqa 
and extend it to $\Ga_{\fin}(D_1)$ by linearity and by setting $\dGa(q,a_1,\ldots, a_m)=0$ on $D_1^{\otimes_s^n}$, where $0\leq n < m$.

We note the following simple relation between the objects introduced above: Let $q$, $p$ be  bounded operators on $\mfh$ which commute  and let $a$ be a self-adjoint operator on some domain $D\subset \mfh$. Then
\beqa
[\Ga(q),\dGa(p,a)]=\dGa(qp,[q,a]) \label{Gamma-dGamma-comm}
\eeqa
in  the sense of quadratic forms on $\Ga_{\fin}(D)\times \Ga_{\fin}(D)$.

\subsection{Extended Fock space} \label{Extended-appendix}

Extended Fock space is defined by $\Ga^{\ex}(\mfh):=\Ga(\mfh)\otimes\Ga(\mfh)$ (or, in a shorter notation, 
$\Fock^{\ex}:=\Fock\otimes\Fock$).
Let $U\colon\Ga(\mfh\oplus\mfh)\to\Ga(\mfh)\otimes \Ga(\mfh)$ be the canonical identification, defined
by the relations
\begin{equation}
Ua^*(h)=(a^*(h)\otimes 1+1\otimes a^*(h))U \quad \textup{and} \quad U\vac=\vac\otimes \vac.
\end{equation}
We will use $U$ to transport objects defined in the previous subsection to the extended Fock space:
Let $q_0, q_{\infty}$ and $a_0,a_{\infty}$  be operators  on  $\mfh$ defined on a dense domain $D\subset\mfh$. 
Let $\un q:=\diag(q_0,q_{\infty})$ and  $\un a:=\diag(a_0,a_{\infty})$ be  operators  on $\mfh\oplus\mfh$
defined on the domain $D\oplus D$.
Then we introduce the following operators on the extended Fock space:
\begin{align}
\Ga^{\ex}(\un q)&:= U\Ga(\un q)U^*=\Ga(q_0)\otimes \Ga(q_{\infty}),\\
\dGa^{\ex}(\un a)&:= U\dGa(\un a)U^*=\dGa(a_0)\otimes 1+1\otimes \dGa(a_{\infty}),\\
\dGa^{\ex}(\un q,\un a)&:= U\dGa(\un q,\un a)U^*=\dGa(q_0, a_0)\otimes \Ga(q_{\infty})+\Ga(q_{0})\otimes \dGa(q_{\infty},a_{\infty}),
\label{dGa-q-p-def-double}
\end{align}
which are defined on $\Ga_{\fin}(D)\otimes \Ga_{\fin}(D)$ and in the last equality we used Lemma~\ref{more-fock-combinatorics}, 
stated below. 
We note that if $q_0$ and $q_{\infty}$ are contractions, then $\Ga^{\ex}(\un q)$ is also a contraction.
In this situation we set 
\begin{align}
\Ga^{(n)}(\un q)&:= \Ga^{\ex}(\un q)_{|\Ga(\mfh)\otimes\Ga^{(n)}(\mfh)},\\
\dGa^{(n)}(\un a)&:= \dGa^{\ex}(\un a)_{|\Ga_{\fin}(D)\otimes \Ga^{(n)}(D) },\\ 
\dGa^{(n)}(\un q,\un a)&:= \dGa^{\ex}(\un q,\un a)_{|\Ga_{\fin}(D)\otimes \Ga^{(n)}(D)}.
\end{align}
In the special case where $q_0=q_{\infty}=:q$, $a_0=a_{\infty}=:a$ we will drop the underlining
and write
\begin{align}
\Ga^{\ex}(q)&= \Ga(q)\otimes \Ga(q),\\
\dGa^{\ex}(a)&= \dGa(a)\otimes 1+1\otimes \dGa(a),\\
\dGa^{\ex}(q, a)&= \dGa(q, a)\otimes \Ga(q)+\Ga(q)\otimes \dGa(q,a), \label{dGa-q-p-def}
\end{align}
which is the standard notation. Since extended objects are unitarily equivalent to
operators on $\Ga(\mfh\otimes\mfh)$, the properties of closedness and essential
self-adjointness are naturally inherited from the single-particle level, as discussed
after formula~(\ref{dGamma-def}) above.

Now let $c_0,c_{\infty}$ be bounded operators on $\mfh$. We define
$c: \mfh \to \mfh\oplus\mfh$, which acts on $h\in\mfh$ by  $c h=(c_0h, c_{\infty}h)$, 
and is s.t. $\|c^*c\|=\|c_0^*c_0+c_{\infty}^*c_{\infty}\|\leq 1$. Then 
\beqa
\cGa(c)&:=&U\Ga(c) \label{check-gamma-def}
\eeqa
is a mapping $\Fock\to\Fock^{\ex}$ of norm one \cite{DGe1}. 
We also define $\cGa^{(n)}(c):=P_{n}\cGa(c)$, where  $P_n\colon \Fock^{\ex}\to \Fock\otimes\Fock^{(n)}$ 
is  the natural restriction map. 
Next, given a linear map $a\colon D\to \mfh\oplus \mfh$, where $D\subset \mfh$, we set
\beqa
\cdGa(c,a):=U\dGa(c,a),
\eeqa
which is a mapping  $\Ga_{\fin}(D)\to\Fock^{\ex}$. 
We also define $\cdGa^{(n)}(c,a)=P_{n}\cdGa(c,a)$.

Let us denote by $(1,1)$ the map $\mfh\to\mfh\oplus\mfh$ which acts by $(1,1)h=(h,h)$,
where $h\in\mfh$. We note that $\|(1,1)^*(1,1)\|=\sqrt{2}$ and  define $\cGa(1,1)$ as an
unbounded operator on $\Ga_{\fin}(\mfh)$. As stated in \cite{DGe1}, the following operators 
\beqa
\cGa(1,1)^*\bigl((N+1)^{-\frac{n}2}\otimes \mathbf{1}_{\{n\}}(N)\bigr) \label{Gamma-one-one}
\eeqa
are bounded for any $n\in\nat$.

\subsection{Useful lemmas}

In this subsection we collect some simple relations between operators on Fock space, which are used repetitively
in the paper. Most of these relations are well known (see e.g. \cite[Section~2]{DGe1} for (\ref{dGa-astar}) and (\ref{dGa-a})).
\bel\label{q-p-lemma} Let $q,p$ be bounded operators and $h\in\mfh$. Then the following
equalities hold in the sense of quadratic forms on $\Ga_{\fin}(\mfh)\times \Ga_{\fin}(\mfh)$:
\begin{align}
&  \dGa(q,p)a^*(h)=a^*(ph)\Ga(q)+a^*(qh)\dGa(q,p), \label{dGa-astar}\\ 
&  a(h)\dGa(q,p)=\Ga(q)a(p^*h)+\dGa(q,p)a(q^*h). \label{dGa-a}
\end{align}
\eel
\begin{proof} Note the identity $\Ga(q+sp)a^*(h)=a^*((q+sp)h)\Ga(q+sp)$ valid for any $s\in \real$.
By computing the matrix elements of this expression between vectors from $\Ga_{\fin}(\mfh)$
and differentiating them  w.r.t. $s$ at $s=0$ we conclude the proof. \end{proof}

\bel\label{rest-terms} Let $\om, a_1,\ldots a_n$ be operators defined on a common domain $D\subset\mfh$, whose
adjoints are defined on a common domain $D^*\subset\mfh$. Let $j$ be a bounded operator on $\mfh$. Then, 
in the sense of quadratic forms on $\Ga_{\fin}(D^*)\times \Ga_{\fin}(D)$
\begin{align}
&  [\dGa(\om),\dGa(j,a_1,\ldots, a_n)]\non\\
& \qquad = \dGa(j, [\om,j], a_1,\ldots, a_n)+\sum_{i=1}^n\dGa(j,a_1,\ldots, [\om, a_i],\ldots, a_n ).\label{new-object}
\end{align}
Now suppose that $j\colon\mfh_1\to\mfh_2$ is  s.t. $\|j\|\leq 1$  and $ a_1,\ldots, a_n\colon \mfh_1\to\mfh_2$ are bounded operators. Then
\beqa
\|\dGa(j, a_1,\ldots, a_n)(1+N)^{-n}\|\leq C\|a_1\|\ldots\|a_n\|.\label{norm-bounds-new-objects}
\eeqa
\eel
\begin{proof} Relation (\ref{new-object}) can easily be seen by differentiating the function
\beqa
(s,s_1,\ldots, s_n)\to\lan \Psi_1,[\Ga(1+s\om), \Ga(j+s_1a_1+\cdots+s_na_n)]\Psi_2\ran,
\eeqa
where $\Psi_1\in \Ga_{\fin}(D^*)$ and $\Psi_2\in \Ga_{\fin}(D)$,
w.r.t. each of the parameters separately and then setting $(s,s_1,\ldots s_n)=0$. The bound (\ref{norm-bounds-new-objects})
follows immediately from  definition~(\ref{new-object-def}). \end{proof}

\bel\label{more-fock-combinatorics} 
Let $q_{0}, q_{\infty}$ be bounded operators on $\mfh$, and let  $p_0, p_{\infty}$ be defined  on a domain $D\subset \mfh$. We define the following operators on $\mfh\oplus\mfh$:
\begin{equation}
\un{q}:=\begin{pmatrix} q_0 & 0 \\ 0 & q_{\infty} \end{pmatrix} 
\qquad \textup{and} \qquad   
\un{p}:=\begin{pmatrix} p_0 & 0 \\ 0 & p_{\infty} \end{pmatrix}.
\end{equation}
 There holds the following identity on vectors from $\Ga_{\fin}(D)\otimes \Ga_{\fin}(D)$:
\beqa
U\dGa(\un q,\un p)U^*=\big(\dGa(q_0,p_0)\otimes \Ga(q_{\infty})+\Ga(q_0)\otimes \dGa(q_{\infty},p_{\infty})\big). \label{cGa-equality-one}
\eeqa
\eel
\begin{proof} Note that $U\Ga(\un q+s\un p)=(\Ga(q_0+sp_0)\otimes \Ga(q_{\infty}+sp_{\infty}))U$ for $s\in\real$.
By computing the matrix elements of this expression between vectors from the specified domains 
and differentiating them  w.r.t. $s$ at $s=0$ we conclude the proof. \end{proof}

\bel\label{rest-terms-double} Let $\om,c_{i,0},c_{i,\infty}$, $1\leq i\leq n$, be operators defined on a common domain $D\subset\mfh$,
whose adjoints are defined on a common domain $D^*\subset\mfh$.  We define $\un\om :=\diag(\om,\om)$ as an operator
 on $\mfh\oplus\mfh$ with a domain $D\oplus D$.  
Now let $j_{0}, j_{\infty}$ be bounded operators
on $\mfh$. We define $c_i:=(c_{i,0},c_{i,\infty})$ as  maps $D\to \mfh\oplus\mfh$ and $j:=(j_{0}, j_{\infty})$ as a map
$\mfh\to\mfh\oplus\mfh$.  Then the following relation holds in the sense of quadratic forms on 
$\Ga_{\fin}(D^*\oplus D^*)\times \Ga_{\fin}(D)$: 
\begin{align}
&  \dGa(\un \om)\dGa(j,c_1,\ldots, c_n)-\dGa(j,c_1,\ldots,c_n)\dGa(\om)\non\\
& \qquad=\dGa(j, [\un\om,j], c_1,\ldots, c_n )+\sum_{i=1}^n\dGa(j, c_1,\ldots,[\un\om, c_i],\ldots, c_n), \label{new-object-repeated}
\end{align}
where  $[\un\om,j]:=\un\om j-j\om=([\om,j_0],[\om,j_{\infty}])$.
\eel
\begin{proof}  The relation follows by differentiating the function 
\beqa
(s,s_1,\ldots,s_n )\to\lan \Psi_1,(\Ga(1+ s\un\om)\Ga(j+\sum_{i=1}^n s_ic_i)-\Ga(j+\sum_{i=1}^n s_ic_i)\Ga(1+s\om))\Psi_2\ran,\quad
\eeqa
where $\Psi_1\in \Ga_{\fin}(D^*\oplus D^*)$ and $\Psi_2\in \Ga_{\fin}(D)$, w.r.t. $s$ and $s_i$
separately and setting $(s,s_1,\ldots,s_n)=0$. \end{proof}
\bel \label{high-order-commutators} Let $a,b$  be operators defined on some common domain $D\subset\mfh$,
whose adjoints are defined on some common domain $D^*\subset\mfh$.  We define $\un a:=\diag(a,a)$ and
$\un b:=\diag(b,b)$ as operators on $\mfh\oplus\mfh$ with domains $D\oplus D$.
Let $q,j_0,j_{\infty}\in B(\mfh)$ and suppose that $[q,j_0]=0$ and $[q,j_{\infty}]=0$.
We define $j:=(j_0,j_{\infty})$ to be a map $\mfh\to\mfh\oplus\mfh$ and we specify $\un q:=\diag(q,q)$  
to be an operator on $\mfh\oplus\mfh$. 
Then, in the sense of quadratic forms on $\Ga_{\fin}(D^*\oplus D^*)\times \Ga_{\fin}(D)$ 
\begin{align}
 \Ga(j)\dGa(q,a)-\dGa(\un q,\un a)\Ga(j) & =\dGa(jq,[j,\un a]), \\
 \Ga(j)\dGa(q,a,b)-\dGa(\un q,\un a,\un b)\Ga(j)& =\dGa(jq,[j, \un a], [j, \un b] )+\dGa(jq,[j, \un a], \un b j)\non\\
& \quad  +\dGa(jq, \un a j, [j, \un b] ),
\end{align}
where $[j,\un a]=j a-\un a j$.
\eel
\begin{proof} The relations follow by differentiating the matrix elements of the functions
\begin{align}
s&\to \Ga(j)\Ga(q+sa)-\Ga(\un q+s\un a)\Ga(j),\\
(s,s_1)&\to  \Ga(j)\Ga(q+sa+s_1b)-\Ga(q+s\un a+s_1\un b)\Ga(j),
\end{align}
in each argument separately, and setting $s=0$, respectively $(s,s_1)=0$. \end{proof}

\section{Commutator expansions}
\setcounter{equation}{0}

Commutator expansions for  functions of several  commuting observables, which we need in the present work, 
were established  in \cite{Ra10}. To state this result, we need first several definitions: 
For $\rho\in \real$, we define the class of functions $S^{\rho}(\real^{\nu})\subset C^{\infty}(\real^{\nu})$,
s.t.
\beqa
|\pa^\al f(x)|\leq C_\al\lan x \ran^{\rho-|\al|},
\eeqa
for any multiindex $\al$. 
In the definition below we use the notation
 $\de_j:=(0,\ldots, 1,\ldots, 0)\in\nat^\nu$, with $1$ on the $j$-th entry.

\bed\label{C-1-definition} Let $A = (A_1,\dotsc,A_\nu)$ be a vector of 
commuting self-adjoint operators with domains $D(A_j)\subset \hil$, and $B$ a bounded operator on $\hil$.
We say that $B\in C^1(A)$, or $B$ is of class $C^1(A)$, if 
the commutator forms $[A_j,B]$, a priori defined as quadratic forms on $D(A_j)$, extend by continuity to bounded operators $[A_j,B]^\circ =: \ad_{A_j}(B) = \ad^{\de_j}_A(B)$.  For $n_0>1$ we define
the class $C^{n_0}(A)$ and iterated commutators $\ad^\al_{A} (B)$
recursively: We say that $B\in C^{n_0}(A)$ if $B\in C^{n_0-1}(A)$
and for any $|\al|< n_0$ and $j\in\{1,\dotsc,\nu\}$, the commutator forms $[A_j, \ad_A^{\al}(B)]$ extend by continuity to bounded operators $\ad_A^{\al+\de_j}(B)$.
\eed
\begin{remark} In the case $\nu=1$ the above definition reduces to a more standard one: $B\in B(\hil)$ belongs to $C^{n_0}(A)$
if for any $\Psi\in\hil$  the map 
\beqa
\real\ni s\to \e^{\i sA}B\e^{-\i sA}\Psi
\eeqa
is $n_0$ times continuously differentiable in the norm topology. We also recall that this definition can be naturally extended to (possibly unbounded) 
self-adjoint operators: In this case we say that $B\in  C^{n_0}(A)$ if $(B-z)^{-1}\in B(\hil)$ is in $C^{n_0}(A)$ for some  -- hence all -- $z\in\complex\backslash \sigma(H)$  \cite{MR12}. Finally, we remind the reader that 
if $B\in C^1(A)$ we have $B D(A)\subset D(A)$.
\end{remark}

Now we are ready to state  the main result of \cite{Ra10}. 
\bel\label{Lemma-Morten} Let $B\in B(\hil)$ and $A=(A_1,\ldots, A_{\nu})$ be a family of self-adjoint, pairwise commuting operators. 
Assume that $B\in C^{n_0}(A)$ for some $n_0\geq n+1\geq 1$, $0\leq t_1\leq n+1$  and  $0\leq t_2\leq 1$
and that $f\in S^{\rho}(\real^{\nu})$ for some $\rho\in \real$ s.t. $t_1+t_2+\rho<n+1$. Then
\beqa
[f(A),B]=\sum_{\al: 1\leq |\al|\leq n}(-1)^{|\al|+1}\fr{1}{\al!}\pa^{\al} f(A)\ad_{A}^{\al}(B)+R_{n+1}(f,A,B),
\eeqa 
as an identity on $D(\lan A \ran^\rho)$, where $R_{n+1}(f,A,B)\in B(\hil)$ and there exists a
constant $c_{n}(f)$, independent of $A$, $B$, s.t. 
\beqa
\|\lan A \ran^{t_1}  R_{n+1}(f,A,B) \lan A \ran^{t_2} \|\leq c_{n}(f)\sum_{\al:|\al|=n+1}\|\ad_A^{\al}(B)\|.
\eeqa
\eel

\begin{remark} One can of course read the commutator expansion in Lemma~\ref{Lemma-Morten} as a form identity
on $D(\lan A\ran^\rho)$.
We wish to argue in this remark, that one can also read it is an operator identity
on $D(\lan A\ran^\rho)$. Suppose $B\in C^{n_0}(A)$, with $n_0\in\nat$. 
Assume $B\in C^{n_0}(A)$. We wish to prove by induction after $n_0$ that 
\begin{equation}\label{AbstractDomInv}
\ad^\alpha_A(B) D(\lan A\ran^{n_0})
\subset D(\lan A\ran^{n_0-|\alpha|}),
\end{equation}
 for multiindices $\alpha$ with $|\alpha|\leq n_0$. 
From \eqref{AbstractDomInv}, using $\rho<n_0$,  it follows by interpolation that 
$\ad_A^\alpha(B) D(\lan A\ran^\rho)\subset D(\lan A\ran^{\rho - |\alpha|})$.
Hence the expansion in Lemma~\ref{Lemma-Morten} is meaningful as an operator identity.
 
Let $B\in C^1(A)$. 
Recall that $B D(A_j)\subset D(A_j)$, for all $j$. Hence, $B D(\lan A\ran)\subset D(\lan A\ran)$.
This proves the claim for $n_0=1$.

Assume that \eqref{AbstractDomInv} is true with $n_0$ replaced by an integer $n < n_0$. To show that  
\eqref{AbstractDomInv} holds true also with $n_0$ we proceed as follows.
We use descending induction after $|\alpha|$, with the case $|\alpha| = n_0$ being trivial.
For $\alpha$ with $|\alpha|< n_0$ we know that $\ad_A^\alpha(B)\in C^1(A)$, and hence
$\ad_A^\alpha(B)D(\lan A\ran )\subset D(\lan A\ran)$ (by the $n_0=1$ step)
 and $A_j \ad_A^\alpha(B) = \ad_A^{\alpha+\delta_j}(B) + \ad_A^\alpha(B)A_j$. 
 The claim now follows by the two induction hypotheses. 
\end{remark}

In our investigations we will use two  special cases of Lemma~\ref{Lemma-Morten}, which we now state explicitly:
\bel\label{many-A} Let $B\in B(\hil)$ and $A=(A_1,\ldots, A_{\nu})$ be a family of self-adjoint, pairwise commuting operators.
 Assume that $B\in C^{3}(A)$ 
and that $f\in S^{2}(\real^{\nu})$. Then
\beqa
[f(A),B]=\sum_{\al:  |\al|=1}\pa^{\al} f(A)\ad_{A}^{\al}(B)+R(f,A,B),
\eeqa 
as an identity on $D(\lan A \ran^2)$, where $R(f,A,B)\in B(\hil)$ and there exists a
constant $c(f)$, independent of $A$ and $B$, s.t. 
\beqa
\| R(f,A,B) \|\leq c(f)\sum_{\al:2\leq |\al|\leq 3}\|\ad_A^{\al}(B)\|.
\eeqa
\eel
\bel\label{one-A} Let $B\in B(\hil)$ and $A$ be a self-adjoint operator. 
Assume that $B\in C^{n_0}(A)$ for some $n_0\geq n+1\geq 1$, 
and that $f\in S^{0}(\real^{\nu})$. Then 
\beqa
[f(A), B]=\sum_{j=1}^n(-1)^{j-1}\fr{1}{j!}f^{(j)}(A)\ad_A^j(B)+R_{n+1}(f,A,B), \label{one-A-expansion}
\eeqa
as an identity on $\hil$, where $R_{n+1}(f,A,B)\in B(\hil)$ and there exists a
constant $c_{n}(f)$, independent of $A$, $B$, s.t. 
\beqa
\|  R_{n+1}(f,A,B) \|\leq c_{n}(f)\|\ad_A^{n+1}(B)\|. \label{d-g-rest-bound}
\eeqa
\eel

\section{Commutator bounds in $L^2(\real^{\nu})$ }\label{commutator-bounds-appendix}
\setcounter{equation}{0}

 Let $A$ and $H$ be self-adjoint operators on a Hilbert space $\hil$, defined on domains $D(A)$ and $D(H)$, and
 s.t. $H$ is of class $C^1(A)$ (cf. Definition~\ref{C-1-definition} above). We recall that the natural  domain $D(A)\cap D(H)$ of the commutator form $\i[H,A]$ is dense in $D(H)$ in the topology given by the norm $\|\Psi\|_H:=\|H\Psi\|+\|\Psi\|$.  We  will write $\i [H,A]^\circ$  for the extension by continuity of the commutator form
from $D(A)\cap D(H)$ to $D(H)$ (and also  for the associated operator
from $D(H)$ to $D(H)^*$). 
If, furthermore, $\i[H,A]^\circ$ extends by continuity to an element of  
$B(\hil)$, as is sometimes the case below,
then $[H,A]^\circ$ will denote this extension.

First, we recall the following abstract result from \cite{Mo81}:
\bep\label{Mourre} Let $H$ and $A$ be self-adjoint operators that satisfy 
\begin{enumerate}[label = \textup{(\alph*)}, ref =\textup{(\alph*)},leftmargin=*]
\item\label{MourreAss-a} $D(A)\cap D(H)$ is a core for $H$. 
\item\label{MourreAss-b} $\e^{\i t A}D(H)\subset D(H)$  and for each $\Psi\in D(H)$ we have $\sup_{|t|<1}\|H\e^{\i t A}\Psi\|<\infty$.
\item\label{MourreAss-c} There is a set $S\subset D(A)\cap D(H)$ which is a core for $H$ and is invariant under $\e^{\i tA}$. 
The form $\i [H,A]$ on $S$ is bounded below and closable, and the  associated self-adjoint operator $\i [H,A]_S^\circ$
satisfies $D(\i [H,A]_S^\circ)\supset D(H)$.   
\end{enumerate}
Then, for all $\Phi,\Psi\in D(A)\cap D(H)$ 
\beqa
\lan \Phi, \i [H,A]\Psi \ran=\lan \Phi, \i [H,A]_S^\circ\Psi \ran.
\eeqa
\eep
\noindent Making use of the above proposition, we  prove the following technical lemma:
\bel\label{domain-questions} Let $g\in C^{\infty}(\real^{\nu})_{\real}$ 
and  let  $v\in C_0^{\infty}(\real^{\nu};\real^{\nu})$. Let $a:=\h(v\cdot\i \nabla_k+\i \nabla_k\cdot v)$, which defines a
  self-adjoint operator in $L^2(\real^{\nu})$. Using the same notation for real-valued functions and their associated self-adjoint multiplication operators 
  we have: $g$ is of class  $C^1(a)$ and the operator $\i[a,g]^\circ$ extends by continuity from $D(g)$ to a bounded operator on 
$L^2(\real^\nu)$ given by
\begin{equation}
\i[a,g]^{\circ} = -v\cdot \nabla g.
\end{equation}
Moreover,   $(z-a)^{-1}$ leaves $D(g)$ invariant for any $z\in\complex\backslash\real$. More precisely, for any $u\in D(g)$ there holds
\begin{equation}
\|g (z-a)^{-1}u\|_2 \leq  \fr{c}{|\Im z |}(\|u\|_2+\|gu\|_2), \label{first-resolvent-bound}
\end{equation}
for come $c\geq 0$ independent of $u$ and $z$.
\eel
\begin{proof} We set in Proposition~\ref{Mourre} $A=a$ and $H=g$ and verify the assumptions: As for \ref{MourreAss-a}, we note that $C_0^{\infty}(\real^{\nu})$, which is a core for $g$, is a subset of $  D(a)\cap D(g)$.

To prove \ref{MourreAss-b}, we follow  \cite{Mo81, MR12}: We recall  that $w_t:= \e^{\i t a}$ is closely related to the
flow $\psi_t$ of the equation $\dot{\psi_t}=v(\psi_t)$ with the initial condition $\psi_0(k)=k$. 
Let $J_t$ be the determinant of the Jacobi matrix $D_k\psi_t$. There holds
 \beqa
(w_t u)(k)=\sqrt{J_t(k)}u(\psi_t(k)), \label{evolution-equation}
\eeqa
where $u\in D(g)$ and $J_t$ is uniformly bounded in $k$ as a consequence of  the Liouville formula:
\beqa
J_t=\e^{\int_0^t ds \mathrm{Tr}Dv(\psi_s(k)) }.
\eeqa
Making use of the boundedness of $v$ we obtain the property of finite propagation speed of $\psi_t$
\beqa
\sup_{k\in\real^{\nu}} \|\psi_t(k)-k\| \leq \sup_{k\in\real^{\nu}} \int_0^t ds\|v(\psi_s(k))\|\leq t\|v\|_{\infty}. \label{finite-speed}
\eeqa 
Equation~(\ref{evolution-equation}) gives 
\beqa
( |g(\psi_t(k))|+1)(|g|w_t u)(k)=|g(k)| \sqrt{J_t(k)} \big(u(\psi_t(k))+(|g|u)(\psi_t(k))\big),
\eeqa
and consequently,
\beqa
(|g|w_t u)(k)=\fr{|g(k)|}{( |g(\psi_t(k))|+1)} \big( (w_t u)(k)+(w_t|g|u)(k)\big). \label{comparing-evolutions}
\eeqa
We note that the
factor $|g(k)|(|g(\psi_t(k))|+1)^{-1}$ is bounded. This follows from the relation
\beqa
g(\psi_t(k))=g(k)+\int_0^{t}ds\, v(\psi_s(k))\cdot \nabla g(\psi_s(k))
\eeqa
and the fact that $v$ is compactly supported. Thus we obtain from formula~(\ref{comparing-evolutions}) that
$w_t u$ is in the domain of $g$ and
\beqa
\|(gw_t u)\|_2\leq c(1+|t|)(\|u\|_2+\|g u\|_2) \label{t-bound-on-the-norm}
\eeqa
for some $c\geq 0$ independent of $u$ and $t$. This concludes the proof of property \ref{MourreAss-b}. 

As for \ref{MourreAss-c}, we set $S=C_0^{\infty}(\real^{\nu})$ and conclude from the finite propagation speed property~(\ref{finite-speed}) 
that $w_t$ leaves $S$ invariant. On $S$ we easily obtain that
\beqa
\i[a,g]=-v\cdot \nabla g \label{form-equality}
\eeqa
and the r.h.s. is bounded due to the compact support of $v$. Thus  $\i [a,g]_S^\circ=-v\cdot \nabla g$ is defined
on the entire Hilbert space, which concludes our verification of \ref{MourreAss-c}.

Now we obtain from Proposition~\ref{Mourre} that  equality (\ref{form-equality}) holds in the sense of 
forms on $D(a)\cap D(g)$, and that $\i[a,g]$ can be extended to a bounded, self-adjoint operator $\i[a,g]^\circ$ 
which coincides with $-v\cdot \nabla g$.

Now let us show that $g\in C^1(a)$. We set $g_0(k)=(z-g(k))^{-1}$ and note that (\ref{form-equality})
applies to the real and imaginary parts of this function. Thus, by
 (\ref{form-equality}) $g_0$ leaves $D(a)$ invariant and  $\i[a,g_0]=g_0\i[a,g]g_0$, defined first as an  operator on  
$D(a)$, extends to a bounded operator on $L^2(\real^{\nu})$. Thus  Lemma~2.2 of \cite{MR12} gives that $g\in C^1(a)$.

Next, we show estimate~(\ref{first-resolvent-bound}). Let us assume that $\mathrm{Im} z>0$. Then
\beqa
(z-a)^{-1}=-\i\int_0^{\infty}dt \,\e^{\i a t}\e^{\i z t}  
\eeqa
and  property~(\ref{first-resolvent-bound}) follows from (\ref{t-bound-on-the-norm}). 
For $\mathrm{Im} z<0$ the argument is analogous. \end{proof}

\bel\label{iterated-commutators}  Let $g_1, \ldots, g_n\in C^{\infty}(\real^\nu)_{\real}$, 
and let $f \in  S^{0}(\real)_{\real}$. Then $f(a/t)\in C^n(g)$, where $g=(g_1,\ldots, g_n)$ is a family
of commuting self-adjoint operators (functions of $k$). More precisely:
 Let $\ti g_i(k):=\chi(k) g(k)$, where $\chi\in C_0^{\infty}(\real)_{\real}$ is equal to one on the
support of $v$ and vanishes outside of a slightly larger set. We define the following bounded operators for $n\in\nat$
\begin{align}
\ti I_0 & :=f(a/t),\\
\ti I_n & :=\i^n[\ti g_n,[\ldots,[\ti g_1, f(a/t)]\ldots]],\quad n\geq 1.
\end{align}
Then, in the sense of quadratic forms on $D(g_n)$,
\beqa
\i[g_n, \ti I_{n-1} ]=\ti I_{n}.
\eeqa
 Consequently,  $\ti I_{n-1}$ leaves $D(g_n)$ invariant and  $\ti I_n$ is the unique bounded operator
which coincides with $\i[g_n, \ti I_{n-1} ]$ on $D(g_n)$ (i.e. $\i[g_n, \ti I_{n-1} ]^{\circ}=\ti I_n$).
\eel
\begin{proof} Proceeding similarly as in \cite{FMS11}, we  write $h(x)=f(x)(x+i)^{-1}$ and choose an almost-analytic extension $\ti h\in C^{\infty}(\complex)$ of $h$, which satisfies
\begin{equation}
 |\pa_{\bar z}\ti h(z)|\leq C_N\lan z\ran ^{-2-N} |y|^N, \label{almost-analytic-bound-one}
\end{equation}
where $z=x+\i y$.  We set $\ha:=a/t$ and write 
\beqa
f(\ha)=\fr{\i}{2\pi}\int_{\complex}\pa_{\bar z}\ti h(z) (\i+\ha)(z-\ha)^{-1}dz \wedge d\bar z \label{tricky-extension}
\eeqa
as a strong integral on $D(a)$. Let us show that $f(\ha)\in C^1(g_1)$: Making use of Lemma~\ref{domain-questions}, and 
formula~(\ref{tricky-extension}), we can write  for
$u_1,u_2\in D(g_1)\cap D(a)$
\beqa
\lan u_1,[g_1, f(\ha)] u_2\ran=-\fr{\i}{t}\fr{\i}{2\pi}\int_{\complex}\pa_{\bar z}\ti h(z) 
\lan u_1, v\cdot\nabla g_1 \,  (z-\ha)^{-1}u_2\ran dz \wedge d\bar z\non\\
-\fr{\i}{t}\fr{\i}{2\pi}\int_{\complex}\pa_{\bar z}\ti h(z) \lan u_1, (\i+\ha)    (z-\ha)^{-1} v\cdot\nabla g_1 (z-\ha)^{-1} u_2\ran   dz \wedge d\bar z.
\eeqa
Due to  property~(\ref{almost-analytic-bound-one})
and the relations $\|(z-\ha)^{-1}\|=|\Im z|^{-1}$, $\|\ha (z-\ha)^{-1}\|\leq 1+|z|/|\Im z|$ we conclude that the integrals
are convergent. Moreover, we note that we can replace $g_1$ in this
formula by the compactly supported function $\ti g_1 =g_1 \chi$. Thus we obtain 
\begin{equation}
\lan u_1,[g_1, f(\ha)] u_2\ran=\lan u_1,[\ti g_1, f(\ha)] u_2\ran. \label{equality-of-gf-commutators} 
\end{equation}
Since $g_1\in C^1(a)$ by Lemma~\ref{domain-questions}, $D(g_1)\cap D(a)$ is dense in $D(g_1)$.
Hence the form  $\i [g_1, f(\ha)]$ is bounded on $D(g_1)$ and therefore $f(\ha)\in C^1(g_1)$. (Cf. Lemma~2.2 of \cite{MR12}).
As a consequence, $f(\ha)$ preserves $D(g_1)$ and we can write
\begin{equation}
\i [g_1,f(\ha)]^{\circ}=\i [\ti g_1, f(\ha)]. \label{new-equality-of-fg-commutators}
\end{equation}
Thus we have proved the lemma for $n=1$.

Let us now consider the case of  $n>1$.
In the sense of quadratic forms on $D(g_n)$ we can write
\begin{align}
\i[g_n, \ti I_{n-1} ]&=\i^n[\ti g_{n-1},[\ldots,[\ti g_1,\i[g_n, f(a/t)]]\ldots]]\non\\
&=\i^n[\ti g_{n-1},[\ldots,[\ti g_1,\i[\ti g_n, f(a/t)]]\ldots]]=\ti I_{n-1},
\end{align}
where in the second step we made use of (\ref{equality-of-gf-commutators}), which holds on $D(g_n)$ as we justified above. 
 Now the proof can be completed as in the case $n=1$. \end{proof}

Let us now proceed to the decay properties  of commutators constructed in the above lemma:
\bel\label{pseudo-lemma}  Let $g_1, \ldots, g_n\in C^{\infty}(\real^\nu)_{\real}$, 
 $f, j_1,\ldots, j_m\in  S^{0}(\real)$ and let us set $j_i^t:=j_i(a/t)$.  
Then $f(a/t)\in C^n(g)$, where $g=(g_1,\ldots, g_n)$ is a family of self-adjoint commuting operators (functions of $k$)
and   the following relations hold:
\begin{align}
& \i[g_1, f(a/t)]^\circ= \fr{1}{t}v\cdot \nabla g_1 f'(a/t)+O(t^{-2}), \label{low-order-commutator}\\
& \i^m[j_1^t,[\ldots,[j_m^t, g_1]^\circ\ldots]]=O(t^{-m}), \label{multiple-j-commutators}\\
&  \i^{n+m} [j_1^t,[\ldots,[j_m^t, [g_1,[\ldots,[g_n, f(a/t)]^\circ\ldots]^\circ]^\circ]\ldots]]=O(t^{-n-m}),\label{multiple-commutator}
\end{align}
Moreover, if $h\in C_0^{\infty}(\real)$ is s.t. $\supp\, f\cap\supp\, h=\emptyset$, then, for 
any $\ti n\in\nat$ (independent of $n$)
\beqa
\i^n[g_1,[\ldots,[g_n, f(a/t)]^\circ\ldots]^\circ]^\circ h(a/t)=O(t^{-\ti n}). \label{regularity-bound}
\eeqa
\eel
\begin{proof} We recall from Lemma~\ref{domain-questions} that the form $\i[a,g_1]$, defined first on $D(a)\cap D(g_1)$,  has a unique extension 
to a bounded operator $\i[a,g_1]^\circ$ which satisfies
\begin{equation}
\i[a,g_1]^{\circ} = -v\cdot \nabla g_1. \label{single-a-g-commutator}
\end{equation}
Thus we can define $\i^n\ad_{a}^n(g_1)$  iteratively: Suppose that $\i^{n-1}\ad_{a}^{n-1}(g_1)$ is a bounded operator  
which coincides with $(-v\cdot \nabla)^{n-1} g_1$. Then we define $\i[a, \i^{n-1}\ad_{a}^{n-1}(g_1)]$
as a quadratic form on $D(a)$ and set     $\i^n\ad_{a}^n(g_1):=\i[a, \i^{n-1}\ad_{a}^{n-1}(g_1)]^\circ$. It is clear from 
relation~(\ref{single-a-g-commutator}) that 
\beqa
\i^n\ad_{a}^n(g_1)=(-v\cdot \nabla)^{n}g_1. \label{iterted-a-g-commutator}
\eeqa

Now we recall from Lemma~\ref{iterated-commutators} that $\i[j_1^t, g_1]^{\circ}=\i[j_1^t, \ti g_1]$, where $\ti g_1$ is a compactly supported
function of the momentum operator. Since $\ti g_1$  belongs to $C^n(a)$ for any $n\in\nat$, by relation (\ref{iterted-a-g-commutator}),
we conclude from~(\ref{one-A-expansion}) that
\begin{align}
[j_1^t, \ti g_1]&=\sum_{\ell=1}^{n'-1}(-1)^{\ell-1}\fr{1}{\ell!}(j_1^{(\ell)})^t\ad_{a/t}^\ell( \ti g_1)+R_{n'}(j_1,a/t, g_1)\non\\
&=\sum_{\ell=1}^{n'-1}(-1)^{\ell-1}\fr{1}{\ell!}(j_1^{(\ell)})^t \fr{1}{t^l}(\i v\cdot \nabla)^\ell \ti g_1    +O(t^{-n'}), \label{l-expansion}
\end{align}
where we used  (\ref{iterted-a-g-commutator})  and  (\ref{d-g-rest-bound}). This proves  (\ref{low-order-commutator}) and (\ref{multiple-j-commutators}) for
$m=1$. To prove (\ref{multiple-j-commutators}) for arbitrary $m$, we proceed by induction. Suppose that (\ref{multiple-j-commutators}) holds for $m<n'$. Then  formula (\ref{l-expansion}) gives
\begin{align}
& [j_1^t,[\ldots,[j_{n'}^t, \ti g_1]\ldots]\non\\
& \qquad =\sum_{\ell=1}^{n'-1}(-1)^{\ell-1}\fr{1}{\ell!}(j_1^{(\ell)})^t \fr{1}{t^\ell} [j_2^t,[\ldots,[j_{n'}^t,  (\i v\cdot \nabla)^\ell \ti g_1     ]\ldots]  
 +O(t^{-n'}),
\end{align}
which is $O(t^{-n'})$ by the induction hypothesis.

Now we proceed to the proof of (\ref{multiple-commutator}). Similarly as in the proof of  Lemma~\ref{iterated-commutators}, 
 we set $\ha=a/t$ and write
\beqa
f(\ha)=\fr{\i}{2\pi}\int_{\complex}\pa_{\bar z}\ti h(z) (\i+\ha)(z-\ha)^{-1}dz \wedge d\bar z \label{tricky-extension-one},
\eeqa
as a strong integral on $D(a)$, where $\pa_{\bar z}\ti h$ satisfies (\ref{almost-analytic-bound-one}). We recall from 
Lemma~\ref{iterated-commutators} that
\beqa
\i^n[g_1,[\ldots,[g_n, f(a/t)]^\circ=\i^n[\ti g_1,[\ldots,[\ti g_n, f(a/t)], \label{multicommutators-equality}
\eeqa
where  $\ti g_i$ are compactly supported functions of $k$. With the help of (\ref{tricky-extension-one}) we can compute
the  commutator on the r.h.s. of (\ref{multicommutators-equality}) as a quadratic form on $D(a)$ (here we make use of the
fact that $\ti g_i\in C^1(a)$ and thus they preserve $D(a)$). The result is a finite linear combination of terms of two types
\begin{align}
\ti I_{0,n}&:=\fr{1}{t^n}\fr{\i}{2\pi}\int_{\complex}\pa_{\bar z}\ti h(z) (\i+\ha) (z-\ha)^{-1}\prod_{i=1}^n\big\{ v\cdot \nabla\ti g_{\si(i)}(z-\ha)^{-1}\big\}dz \wedge d\bar z,\,\,  \label{Helffer-pseudo-one}\\
\ti I_{j,n}&:=\fr{1}{t^n}\fr{\i}{2\pi}\int_{\complex}\pa_{\bar z}\ti h(z) v\cdot\nabla \ti g_j (z-\ha)^{-1}\prod_{i=1}^{n-1}\big\{ v\cdot \nabla\ti g_{\de(i)}
(z-\ha)^{-1}\big\}dz \wedge d\bar z,\,\, \label{Helffer-pseudo-two}
\end{align}
where $\si$ is some permutation of $(1,\ldots, n)$ and $\de$ is some permutation of $(1,\ldots,\check j,\ldots, n)$. 
Making use of properties~(\ref{almost-analytic-bound-one}), and of 
 the relations $\|(z-\ha)^{-1}\|=|\Im z|^{-1}$, $\|\ha (z-\ha)^{-1}\|\leq 1+|z|/|\Im z|$ we conclude that
\beqa
|\lan u_1, \ti I_{i,n} u_2\ran|\leq \fr{c}{t^n}\|u_1\|\|u_2\|
\eeqa
for $u_1,u_2\in D(a)$ and $i\in \{0,\ldots, n\}$. This gives (\ref{multiple-commutator}) for $m=0$ and also
verifies that the r.h.s. of (\ref{multicommutators-equality}) coincides with a liner combination of bounded operators $\ti I_{i,n}$
on the entire Hilbert space.
Let us now proceed to the case $m>0$.
We note that 
\beqa
[j_1^t,[\ldots,[j_m^t, \ti I_{i,n}]\ldots]]
\eeqa
is again a linear combinations of terms of the form (\ref{Helffer-pseudo-one}) and (\ref{Helffer-pseudo-two}), except that
some of the insertions $v\cdot \nabla  \ti g_{j}$ are replaced with
\beqa
 [j_{i_1}^t,[\ldots,[j_{i_{m'}}^t,  v\cdot \nabla \ti g_{j}]\ldots]] \label{auxiliary-j-g-commutator}
\eeqa
for some $i_1,\ldots, i_{m'}\in \{1,\ldots,m\}$. Since (\ref{auxiliary-j-g-commutator}) is of order $O(t^{-m'})$ by
(\ref{multiple-j-commutators}), the proof of (\ref{multiple-commutator}) can now be completed as in the case  $m=0$.

To prove (\ref{regularity-bound}), 
we proceed by induction: For $n=1$ it follows from (the adjoint of) formula~(\ref{l-expansion}). Now we 
define a sequence $\un{\ti g}:=(\ti g_1, \ti g_2, \ldots )$ and write for any $n\in\nat$
\beqa
[\ti g_1,[\ldots,[\ti g_n, f(a/t)]\ldots]]=\ad_{\un{\ti g}}^{\al_n}(f(a/t)),
\eeqa
where $\al_n$ is a multiindex s.t. $\al_n(j)=1$ for $1\leq j\leq n$
and $\al_n(j)=0$ for $j>n$. Now suppose that (\ref{regularity-bound}) holds for 
for $n<n'$. We obtain 
\begin{align}
 \ad_{\un{\ti g}}^{\al_{n'}}(f(a/t)) h(a/t)&= [\ti g_{n'}, \ad_{\un{\ti g}}^{\al_{n'-1}}(f(a/t))]h(a/t)\non\\
&= \ad_{\un{\ti g}}^{\al_{n'-1}}(f(a/t))[h(a/t), \ti g_{n'}]+O(t^{-\ti n}),
\end{align}
where we made use of the induction hypothesis. The first term on the r.h.s. above is $O(t^{-\ti n})$ by the induction hypothesis
and formula~(\ref{l-expansion}). \end{proof}

\section{Admissible and regular propagation observables}
\setcounter{equation}{0}
\bed\label{admissible} Let $\real\ni t\to b(t)\in B(\mfh)$ be a propagation observable, which is bounded, uniformly in $t$.
Let $j_l\in S^{0}(\real)$  and  $g_i\in C^{\infty}(\real^{\nu})_{\real}$, $i,l\in \nat$ and let us set $j_l^t:=j_l(a/t)$. Suppose that $b(t), b(t)^*\in C^n(g)$ for any $n\in \nat$ and
$t\in \real$, where $g=(g_1,\ldots, g_n)$ is a family of commuting self-adjoint operators understood as functions of $k$. 
(Cf. Definition~\ref{C-1-definition}).

\begin{enumerate}[label = \textup{(\alph*)}, ref =\textup{(\alph*)},leftmargin=*]
\item We say that $b$ is admissible, if  for any $m, n$ 
\begin{equation}\label{admissibility-property}
[j_1^t,[\ldots,[ j_m^t,[g_1,[\ldots,[g_n, b(t)]^\circ\ldots]^\circ]^\circ\ldots]\ldots]]=O(t^{-n-m}). 
\end{equation}
\item 
We say that b is regular, if there exists some neighbourhood of zero $\De$, s.t. for any $h_{\De}\in C_0^{\infty}(\real)$,
supported in $\De$, and any $n,\ti n \in\nat$
\begin{equation}\label{regularity-property}
[g_1,[\ldots,[g_n, b(t)]^\circ \ldots]^\circ]^\circ h_{\De}(a/t)=O(t^{-\ti n}),
\end{equation}
and the same relation holds for $b$ replaced with $b^*$. We will call $\De$ the regularity region of $b$.
\end{enumerate}
\eed
\bel\label{G-properties} Let $G$ be the function appearing in the interaction term of the Hamiltonian~(\ref{fiber-Hamiltonian})
and let $b$ be  a regular propagation observable. Then, for any $n\in \nat$,
with $n\leq 6$, there exists a $C_n$ s.t.
\beqa
 \|b(t)\om^n G\|_2\leq C_n/t^{2}. \label{interaction-decay-zero}
\eeqa
\eel
\begin{proof} By \ref{MC-G-decay} and~\ref{MC-three}, we have $\|\om^n G\|_2<\infty$ for $n\leq 6$.
We recall that $a=\h\{v\cdot \i\nabla_k+\i\nabla_k\cdot v\}$, $v$ is compactly supported and vanishes in a neighbourhood of zero, and by \ref{ST-DistDer}; $\pa^{\al}_kG$ is locally
square-integrable away from zero for $|\al|\leq 2$.
Hence, $\|a^2 \om^n G\|_2<\infty$. Now, exploiting regularity of $b$, we write 
\begin{equation}
 \|b(t) \om^n G\|_2  \leq \|b(t) h_{\De}(a/t)\|  \, \| \om^n G\|_2
 +\|h_{\real\backslash \De  }(a/t)(a/t)^{-2}\|   
\fr{1}{t^2}\|a^2\om^n G\|_2 \leq C/t^2,
\end{equation}
where $\De$ appeared in Definition~\ref{admissible} and $h_{\De}$, $h_{\real\backslash \De  }$ form a smooth partition 
of unity s.t. $h_{\De}$ is supported in $\De$ and equal to one on a neighbourhood of zero. \end{proof}   
\bel\label{concrete-admissible-obs} Let $q\in C^{\infty}(\real)_{\real}$ be  s.t.  $q=0$ on some neighbourhood $\De$ of zero   and  
$q'\in C_0^{\infty}(\real)$.
Then the propagation observables
\beqa
\real\ni t \to q^t, \quad \real\ni t \to t\pa_t q^t  \label{propagation-obs-one} 
\eeqa
are admissible and regular with the regularity region $\De$.  (Here $q^t=q(a/t)$).
\eel
\begin{proof} Follows immediately from Lemma~\ref{pseudo-lemma} and the assumed support properties of $q$. \end{proof}

\section{Auxiliary Hamiltonian and energy bounds}
\setcounter{equation}{0}

\setcounter{equation}{0}

In this section we prove higher-order bounds for $H(\xi)$ w.r.t. the free Hamiltonian $H_{0}(\xi)$
(and their counterparts for the auxiliary Hamiltonians introduced in Definition~\ref{auxiliary-Hamiltonian} below).  
We cannot rely on standard higher-order bounds for 
$H(\xi)$ w.r.t. the free boson Hamiltonian $H_{\pho}$ (see e.g. \cite[Lemma~31]{FGSch2}
and \cite[Lemma~8]{FGSch3}), since they do not suffice in the case of the polaron model. 
\bel\label{interchanging-lemma} Let $F\in C^{\infty}(\real^{\mu})_\real$, $f\in C^{\infty}(\real^{\nu}; \real^{\mu})$, $\mu\in\nat$, and let  
$G\in L^2(\real^{\nu})$. 
Then, in the sense of operators on  $\mcC=\Gamma_\fin(C_0^\infty(\real^\nu))$, 
\begin{align}
\label{creation-operator-identity}
a^*(G) F\bigl(\dGa(f(k))\bigr)&= \int dp \, G(p)  F\bigl( -f(p)+\dGa( f(k)) \bigr)a^*(p), \\
\label{annihilation-operator-identity}
a(G)F\bigl(\dGa(f(k))\bigr)  &= \int dp \, \ov{ G}(p) F\bigl(  f(p)+\dGa( f(k)) \bigr) a(p).
\end{align}
\eel
\begin{proof}  Let  $\Psi=\{\Psi^{(n)}\}_{n\in\nat}\in\mcC$ , and observe that only finitely many $\Psi^{(n)}$'s  are nonzero. 
The expression $F(\dGa(f(k)))$ is well-defined as a symmetric operator on $\mcC$,
where it is also essentially self-adjoint.
Then
\begin{align}
& \bigl\{a^*( G)  F\bigl(\dGa( f(k))\bigr)\Psi\bigr\}^{(n)}(k_1,\ldots, k_n)\non\\
& \quad = \fr{1}{\sqrt{n}}\sum_{i=1}^n G(k_i)F\bigl(f(k_1)+\cdots+\check{f(k_i)}+\cdots+f(k_n) \bigr)\Psi^{(n-1)}(k_1,\ldots,\check k_i,\ldots, k_n)\non\\
& \quad =\fr{1}{\sqrt{n}}\sum_{i=1}^n\int dp\, \de(p-k_i)  G(p)F\bigl(f(k_1)+\cdots+f(k_i)+\cdots+ f(k_n) - f(p)\bigr)\non\\
& \qquad \qquad \qquad \times\Psi^{(n-1)}(k_1,\ldots,\check k_i,\ldots, k_n)\non\\
& \quad =\int dp\,\bigl\{G(p) F\bigl(\dGa( f(k))- f(p)\bigr)a^*(p)\Psi\bigr\}^{(n)}(k_1,\ldots, k_n).
\end{align}
This proves (\ref{creation-operator-identity}). As for (\ref{annihilation-operator-identity}), we compute 
\begin{align}
&  \bigl\{a(G)F\bigl(\dGa( f(k))\bigr) \Psi\bigr\}^{(n)}(k_1,\ldots, k_n)\non\\
& \quad = \sqrt{n+1}\int dp\, \ov{G}(p)F\bigl(f(p)+f(k_1)+\cdots+ f(k_n)\bigr)\Psi^{(n+1)}(p,k_1,\ldots, k_n)\non\\
& \quad=\int dp\, \bigl\{\ov{G}(p)F\bigl(f(p)+\dGa( f(k)) \bigr)a(p)\Psi\bigr\}^{(n)}(k_1,\ldots, k_n).
\end{align}
This concludes the proof. \end{proof}

We have the following higher-order lemma

\bel\label{technical}   Let $n_0\in\nat$ and suppose $\lan k\ran^{(n_0-1)\max\{1, s_\Omega\}}G\in L^2(\real^\nu)$. Then $D(|H(\xi)|^n) = D(H_0(\xi)^n)$ for all $n\leq n_0$. Furthermore,
\begin{equation}
\sup_{\lambda\geq 0,\xi\in\real^\nu}\bigl(
\|(H_0(\xi)+\lambda)^n(H(\xi)-\i +\lambda)^{-n}\| + \|(H(\xi)+\lambda)^n(H_0(\xi)-\i +\lambda)^{-n}\| 
\bigr) <\infty.
\end{equation}
\eel

\begin{proof} It is an easy consequence of the spectral theorem, that it suffices to prove the claimed uniform bound for $\lambda = 0$ and uniformly in $\xi$. 
(This follows by an application of the binomial formula to $(H_0(\xi)+\lambda)^n$ and $(H(\xi)+\lambda)^n$).
In fact, in order to take fractional roots we observe that
we can replace $\i$ by a point below the bottom of the spectrum of the relevant operator.
 For this purpose we recall the notation $\Sigma_0 = \inf \sigma(H)$.
We begin by arguing that for $n\leq n_0$
we have 
\begin{equation}\label{DomIncl}
(H(\xi)-\Sigma_0+1)^{-n}\Fock\subset D(H_0(\xi)^{n})
\end{equation}
 and
\begin{equation}\label{UnifDomIncl}
\sup_{ \xi\in\real^\nu} \|H_0(\xi)^{n}(H(\xi)-\Sigma_0+1)^{-n}\|<\infty,
\end{equation}
for all $n\leq n_0$.
The proof will go by induction in half integer powers $n$. Clearly, since $D(H(\xi)) = D(H_0(\xi))$, 
\eqref{DomIncl} holds true for $n=1$ (and hence by interpolation for $n=1/2$). 
Furthermore, the computation
\begin{equation}
 H_0(\xi)(H(\xi)-\Sigma_0+1)^{-1} = I -(1-\Sigma_0+\phi(G))(H(\xi)-\Sigma_0+1)^{-1},
\end{equation}
together with $N^{1/2}$-boundedness of $\phi(G)$ and the estimate
\begin{align}
\non\| N^\frac12(H(\xi)-\Sigma_0+1)^{-1}\|^2 &\leq \frac1{m} \| (H(\xi)-\Sigma_0+1)^{-1}(H_0(\xi)+1)(H(\xi)-\Sigma_0+1)^{-1}\|\\
& \leq C_1 +
C_2 \|N^\frac12(H(\xi)-\Sigma_0+1)^{-1}\|
\end{align}
implies  \eqref{UnifDomIncl} for $n=1$. Hence by interpolation also for $n=1/2$.

We now assume $n\geq 3/2$ (and $n\leq n_0$) is an element of $\nat/2$, and by induction 
we can assume that \eqref{DomIncl} and \eqref{UnifDomIncl} hold with $n$ replaced with $n-1/2$. To perform the induction step it suffices to show that
\begin{equation}\label{DomIncl-n}
\phi(G)(H(\xi)-\Sigma_0+1)^{-n}\Psi \in D(H_0(\xi)^{n-1}),
\end{equation}
for  $\Psi\in\hil$, and
\begin{equation}\label{UnifDomIncl-n}
\sup_{\xi\in\real^\nu}
\|H_0(\xi)^{n-1} \phi(G) (H(\xi)-\Sigma_0+1)^{-n}\| < \infty.
\end{equation}
The statement \eqref{DomIncl-n} is implied by showing that we have
\begin{equation}\label{enough-tech-1}
\phi(G)(H(\xi)-\Sigma_0+1)^{-n}\Psi \in D(\dGa(\omega)^{n-1})\cap D(\Omega(\xi-\dGa(k))^{n-1})
\end{equation}
and statement \eqref{UnifDomIncl-n} follows from 
\begin{equation}\label{enough-tech-2}
\begin{aligned}
& \sup_{\xi\in\real^\nu}
\|\dGa(\omega)^{n-1} \phi(G) (H(\xi)-\Sigma_0+1)^{-n}\|<\infty,\\
&  \sup_{ \xi\in\real^\nu}
\|\Omega(\xi-\dGa(k))^{n-1} \phi(G) (H(\xi)-\Sigma_0+1)^{-n}\| <\infty,
\end{aligned}
\end{equation}
which by induction is known to hold for $n$ replaced with a half-integer  $n' \leq n-1/2$, cf. what was done for $n=1$. 

Let us now prove (\ref{enough-tech-2}).
Let $F_1(r) = r^{n-1}$, $f_1(k) = \omega(k)$, $F_2(k) = \Omega(\xi-k)^{n-1}$, and $f_2(k) = k$, where $r\in\real$, $k\in\real^{\nu}$.
Write $\phi(G) = a^*(G) + a(G)$. Below we only deal with the $a(G)$ contribution, which is the most complicated. The contribution from $a^*(G)$ is similar but simpler.
Compute for $\Phi,\Psi_1\in\mcC$ 
\begin{equation}\label{EBoundStep1}
\langle F_j(\dGa(f_j(k)))\Phi,a(G)\Psi_1\rangle = \langle a^*(G)F_j(\dGa(f_j(k)))(N+1)^{-\frac12}\Phi,N^{\frac12}\Psi_1\rangle.
\end{equation}
Anticipating the use of \eqref{creation-operator-identity} we introduce
\begin{equation}
\Phi_1 = \int dp\, G(p)\frac{F_j(\dGa(f_j(k)) - f_j(p) )}{1 + F_j(\dGa(f_j(k)))} a^*(p) (N+1)^{-\frac12} \Phi.
\end{equation}
The norm of the $n$-particle ($n\geq 1$) contribution is
\begin{align}
\non\|\Phi_1^{(n)}\|^2 & = \int dk_1\cdots dk_{n} \Bigl| \frac1{\sqrt{n}}\sum_{i=1}^n G(k_i) \frac{F_j(f_j(k_1)+\cdots+f_j(k_n) - f_j(k_i))}{1 + F_j(f_j(k_1)+\cdots+f_j(k_n))}\\
& \qquad  \times n^{-\frac12}\Phi^{(n-1)}(k_1,\dots,\check{k}_i,\dots,k_n) \Bigr|^2.
\end{align}
Observe the bound $F_j(x- f_j(k_i))\leq C(1+ F_j(x))\lan k_i\ran^{(n-1)s}$, valid for $j=1,2$ uniformly in $\xi$,
where $s = \max\{1, s_\Omega\}$. Here we used \ref{MC-three}--\ref{MC-Subadd}. This implies
\begin{align}
\non\|\Phi_1^{(n)}\|^2 & \leq \frac{C^2}{n^2}  \int dk_1\cdots dk_{n} \Bigl(\sum_{i=1}^n 
\lan k_i\ran^{(n-1)s}|G(k_i)
\Phi^{(n-1)}(k_1,\dots,\check{k}_i,\dots,k_n) |\Bigr)^2\\
& \leq C' \|\lan k\ran^{(n-1)s}G\|^2 \|\Phi^{(n-1)}\|^2.
\end{align}
Hence, for some $\xi$-independent constant $C''$ we have 
\begin{equation}
\|\Phi_1\| \leq  C''\|\lan k\ran^{(n-1)s}G\|\|\Phi\|.
\end{equation}

Combining \eqref{creation-operator-identity}, \eqref{EBoundStep1}, and the above estimate we get
\begin{align*}
|\langle F_j(\dGa(f_j(k)))\Phi, a(G)\Psi_1\rangle| & \leq C'' \|\lan k\ran^{(n-1)s}G\|\|\Phi\| \|(1+F_j(\dGa(f_j(k))))N^\frac12 \Psi_1\|\\
& \leq m^{-\frac12} C'' \|\lan k\ran^{(n-1)s}G\|\|\Phi\| \|(1+H_0(\xi)^{n-\frac12})\Psi_1\|.
\end{align*}
From the assumption on $G$ we conclude 
$\|\lan k\ran^{(n-1)s}G\|<\infty$.
Since $\mcC$ is a core for any power of $H_0(\xi)$, the estimate extends from $\mcC$ to $\Psi_1\in D(H_0(\xi)^{n-1/2})$,
and hence in particular to $\Psi_1=(H(\xi)-\Sigma_0+1)^{-n}\Psi$. This shows that
$a(G)(H(\xi)+\Sigma_0+\lambda)^{-n}\Psi\in D(F_j(\dGa(f_j(k))))$, for $j=1,2$, and hence
completes the proof of \eqref{enough-tech-1}. The uniform estimates
\eqref{enough-tech-2} follow by the same estimate just derived, since \eqref{enough-tech-2} holds true
with $n$ replaced by $n-1/2$.

In order to establish the reverse claim, that $(H_0(\xi)+1)^{-n_0}\Fock\subset D(|H(\xi)|^{n_0})$
and that $\sup_{\xi\in\real^\nu} \|(H(\xi)-\Sigma_0)^{n_0}(H_0(\xi)+1)^{-{n_0}}\| <\infty$, we 
argue again 
by induction after half-powers. The computations
\begin{align*}
(H(\xi)-\Sigma_0)^n(H_0(\xi)+1)^{-n} & = \bigl\{(H(\xi)-\Sigma_0)^{n-1} (H_0(\xi)+1)^{-(n-1)}\bigr\}\\
& \qquad \times  \bigl\{  (H_0(\xi)+1)^{n-1} (H(\xi)-\Sigma_0) (H_0(\xi)+1)^{-n}\bigr\}
\end{align*}
and
\[
(H_0(\xi)+1)^{n-1} H(\xi) (H_0(\xi)+1)^{-n} = H_0(\xi)(H_0(\xi)+1)^{-1} + (H_0(\xi)+1)^{n-1} \phi(G) (H_0(\xi)+1)^{-n},
\]
together with the observation that we never used above that the $n$ resolvents were
interacting, conclude the proof of the lemma.
\end{proof}

\bed\label{auxiliary-Hamiltonian} We  define  the free and interacting auxiliary  Hamiltonians on $\hil_1=\Ga(\mfh\oplus\mfh)$ by
\begin{equation}
 H_{1,0}(\xi):=\Omeg(\xi-\dGa(\un k))+\dGa(\un \om) \quad \textup{and}\quad
H_1(\xi):=H_{1,0}(\xi)+\phi(G,0)\label{AuxHam}
\end{equation}
on their domain of essential self-adjointness $\mcC_1:=\Ga_{\fin}(C_0^{\infty}(\real^\nu)\oplus C_0^{\infty}(\real^{\nu}))$.
The operators   $\un k_i=\diag(k_i,k_i)$,$i\in\{1,\ldots,\nu\}$ and $\un \om=\diag(\om,\om)$ are essentially self-adjoint  on 
$C_0^{\infty}(\real^{\nu})\oplus C_0^{\infty}(\real^{\nu})\subset \mfh\oplus\mfh$ and $(G,0)\in\mfh\oplus\mfh$. We note 
that $H_{1,0}(\xi)=UH_{0}^{\ex}(\xi)U^*$, $H_{1}(\xi)=UH^{\ex}(\xi)U^*$ and $\mcC_1=U^*\mcC^{\ex}$.
\eed

\begin{corollary}\label{Cor-technical2}  Let $n_0\in\nat$ and suppose $\lan k\ran^{(n_0-1)\max\{1, s_\Omega\}}G\in L^2(\real^\nu)$. Then for any $n\leq n_0$, $\ell\in\nat$ and $\xi\in\real^\nu$ we have
$D(|H^{\ex}(\xi)|^n) = D(H_0^{\ex}(\xi)^n)$, 
$D(|H^{(\ell)}(\xi)|^n) = D(H_0^{(\ell)}(\xi)^n)$ and $D(|H_1(\xi)|^n) = D(H_{1,0}(\xi)^n)$.
\end{corollary}

\begin{proof}
Using the direct integral decomposition
\[
H_0^{(\ell)}(\xi)^n(H^{(\ell)}(\xi)+\Sigma_0+1)^{-n}
= \int^\oplus dk_1\cdots dk_\ell H_0^{(\ell)}(\xi;\uk)^n(H^{(\ell)}(\xi;\uk)+\Sigma_0+1)^{-n},
\]
and similarly with $H_0$ and $H$ interchanged,
we conclude the corollary from Lemma~\ref{technical}.
\end{proof}

\begin{remark}\label{Remark-AuxHamSuffices}
 The auxiliary Hamiltonian $H_1(\xi)$ is useful when computations and estimates involve 
 multiple Fock space operations, since only one Fock space is involved. However, when one makes manifest use of conservation of asymptotic particle number, the $H^{\ex}(\xi)$ representation is advantageous.

 Having established estimates and identities for $H_1(\xi)$
 we shall by conjugating with the unitary $U$  obtain analogous results for the extended Hamiltonian $H^{\ex}(\xi)$. Then,
by applying the projection $P_\ell$ on the subspace $\Fock\otimes \Fock^{(\ell)}\subset \Fock^{\ex}$ we obtain analogous results for the
Hamiltonians $H^{(\ell)}(\xi)$.
\end{remark} 

\begin{corollary}\label{Cor-technical} Let $n_0\in\nat$ and suppose $\lan k\ran^{(n_0-1)\max\{1, s_\Omega\}}G\in L^2(\real^\nu)$. Then for any $n\leq n_0$, $\ell\in\nat$ 
there exists $C>0$ such that for all $z\in\complex$, with $\Imm z\neq 0$, we have
\begin{equation}
\begin{aligned}
&\bigl\|H_0(\xi)^n(H(\xi)-z)^{-1}(H_0(\xi)+1)^{-(n-1)}\bigr\| \leq C |\Imm z|^{-1},\\
&\bigl\|H_0^{(\ell)}(\xi)^n(H^{(\ell)}(\xi)-z)^{-1}(H_0^{(\ell)}(\xi)+1)^{-(n-1)}\bigr\| \leq C |\Imm z|^{-1},\\
&\bigl\|H_0^\ex(\xi)^n(H^\ex(\xi)-z)^{-1}(H^\ex_0(\xi)+1)^{-(n-1)}\bigr\| \leq C |\Imm z|^{-1},\\
&\bigl\|H_{1,0}(\xi)^n(H_1(\xi)-z)^{-1}(H_{1,0}(\xi)+1)^{-(n-1)}\bigr\| \leq C |\Imm z|^{-1}.
\end{aligned}
\end{equation}
\end{corollary}

\begin{proof} The corollary follows from Corollary~\ref{Cor-technical2}, Remark~\ref{Remark-AuxHamSuffices} and the computation:
\begin{align}
& H_{1,0}(\xi)^n(H_{1}(\xi)-z)^{-1}(H_{1,0}(\xi)+1)^{-(n-1)}   = \bigl\{ H_{1,0}(\xi)^n(H_1(\xi)-\Sigma_0+1)^{-n}\bigr\}\\
 &\non \qquad \times 
\bigl\{(H_1(\xi)-\Sigma_0+1)(H_1(\xi)-z)^{-1}\bigr\}\bigl\{(H_1(\xi)-\Sigma_0+1)^{n-1}(H_{1,0}(\xi)+1)^{-(n-1)}\bigr\}.
\end{align}
Alternatively one can repeat the computation above for each pair of free and interacting Hamiltonian, invoking Corollary~\ref{Cor-technical2} for each of them separately.
\end{proof}

\section{Commutators with the Hamiltonian}
\setcounter{equation}{0}

\subsection{Commutators involving $\dGa(\,\cdot\,,\,\cdot\,)$}

In this subsection we will make use of the auxiliary Hamiltonian $H_1(\xi)$
introduced in Definition~\ref{auxiliary-Hamiltonian}).

\bel\label{comm-H-dGa-lemma} Let $q_0,q_{\infty}\in C^{\infty}(\real)$ be s.t. $q_0',q_{\infty}'\in C_0^{\infty}(\real)$ and $0\leq q_0,q_{\infty}\leq 1$.
 Let $\un q^t=\diag(q_0^t,q_{\infty}^t)$
be the corresponding propagation observable on $\mfh\oplus\mfh$. Let $\real\ni t\to b_j(t)\in B(\mfh)$, $j\in \{0,\infty\}$,  
be two  families of admissible  operators and let $\un b=\diag(b_0,b_\infty)$ be the corresponding 
propagation observable on $\mfh\oplus\mfh$. Let $f\in S^{s_\Omega}(\real)$.
Then, setting $R_{1,0}:=(1+H_{1,0}(\xi))^{-1}$,
we obtain
\begin{align}
 [f(\dGa(\un k)),\dGa(\un q,\un b)]R_{1,0}^4
=\nabla f(\dGa(\un k) )\cdot \bigl(\dGa(\un q,[\un k,\un b]^\circ) 
+\dGa(\un q,[\un k, \un q]^\circ, \un b )\bigr) 
R_{1,0}^4+O(t^{-2})
\label{dGa-comm-double-lemma}
\end{align}
and each term on the r.h.s. above is bounded and $O(t^{-1})$.
\eel
\begin{proof} Observe first that by Lemma~\ref{iterated-commutators}
and Definition~\ref{admissible},
  $\un b, \un q \in C^1(\un k_j)$, for each $j$.
Hence $[\un k_j,\un b]$ and $[\un k_j, q]$ extend from $D(\un k)$ (dense in each $D(\un k_j)$) to bounded operators  $[\un k_j,\un b]^\circ$ and $[\un k_j, q]^\circ$. We write $[\un k, \,\cdot\, ]^\circ$ for the vector $([\un k_1,\,\cdot\,]^\circ,\dotsc,[\un k_\nu,\,\cdot\,]^\circ)$. By Lemma~\ref{pseudo-lemma}
and Definition~\ref{admissible} all these bounded operators are $O(t^{-1})$.

  Making use of Lemma~\ref{many-A}, with $B=\dGa(\un q,\un b)(1+N_1)^{-4}$, we get
\begin{align}
& [f(\dGa(\un k)), \dGa(\un q,\un b)](1+N_1)^{-4}\non\\
& \qquad =\nabla f(\dGa(\un k))\cdot 
\big(\dGa(\un q,[\un k, \un b]^\circ)+\dGa(\un q,[\un k, \un q]^\circ, \un b )    \big)(1+N_1)^{-4}\non\\
& \qquad\quad +R(f, \dGa(\un k),\dGa(\un q,\un b)(1+N_1)^{-4}), \label{Omega-formula-lemma}
\end{align}
as a form equality on $D(\dGa(\un k)^2)$, where $N_1$ is the number operator on $\Ga(\mfh\oplus\mfh)$.
Here we exploited the fact that $f \in S^2(\real^{\nu})$  
and that $\dGa(\un q,\un b)(1+N_1)^{-4}$ is bounded and belongs to $C^3(\dGa(\un k))$ by the assumed properties
of $\un q, \un b$ and by Lemma~\ref{rest-terms}.  Moreover, we obtain from Lemma~\ref{many-A} and Lemma~\ref{rest-terms} that
\begin{align}
& \|R(f, \dGa(\un k),\dGa(\un q, \un b)(1+N_1)^{-4})\|\non\\
& \qquad \leq \sum_{\al: 2\leq |\al| \leq 3 }\|\ad_{\dGa(\un k)}^{\al}(\dGa(\un q, \un b))(1+N_1)^{-4})\|=O(t^{-2}).
\end{align}
Since $(1+N_1)^\ell(1+H_{1,0}(\xi))^{- \ell}$ is bounded for any $\ell\in\nat$, we have shown that 
\begin{align}
& [f(\dGa(\un k)), \dGa(\un q,\un b)]R_{1,0}^4\non\\
& \qquad =\nabla f(\dGa(\un k))\cdot\bigl(\dGa(\un q, [\un k, \un b]^\circ)+\dGa(\un q, [\un k,\un q]^\circ,  \un b) \bigr) R_{1,0}^4+O(t^{-2}).
\end{align}
Let us show that the term involving $\dGa(\un q, [\un k,\un q]^\circ,  \un b)$ above  is $O(t^{-1})$. 
If $s_\Omega\leq 1$, then it follows from  Lemma~\ref{rest-terms} directly, since $\nabla f$ in this case is bounded. If $s_\Omega\geq 1$, we can insert $I = (\dGa(\un k)\cdot\dGa(\un k)+1)^{-1}(\dGa(\un k)\cdot\dGa(\un k)+1)$, noting that $\dGa(\un k)\cdot\dGa(\un k) = \sum_{j=1}^\nu \dGa(\diag(k_j,k_j))^2$,
and compute
\begin{align}
&\nabla f (\dGa(\un k))\cdot \dGa(\un q, [\un k,\un q]^\circ,  \un b)R_{1,0}^4\non\\ 
& \quad =\nabla f(\dGa(\un k))\cdot   (\dGa(\un k)\cdot\dGa(\un k)+1)^{-1} [\dGa(\un q, [\un k,\un q]^\circ,  \un b),\dGa(\un k)\cdot \dGa(\un k)]   R_{1,0}^4\\
\non & \qquad +\nabla f(\dGa(\un k))\cdot (\dGa(\un k)\cdot\dGa(\un k)+1)^{-1}\dGa(\un q, [\un k,\un q]^\circ,  \un b) (\dGa(\un k)\cdot\dGa(\un k)+1) R_{1,0}^4.
\end{align}
This expression is $O(t^{-1})$ by Lemma~\ref{rest-terms} and the bound $|\pa_i f(\eta)|\leq c\lan \eta \ran$. 
Note that to deal with the first term on the right-hand side one should first 
expand the commutator and write
\begin{align*}
[\dGa(\un q, [\un k,\un q]^\circ,  \un b),\dGa(\un k)\cdot \dGa(\un k)] 
& = \sum_{j=1}^\nu \dGa(\diag(k_j,k_j))[\dGa(\un q, [\un k,\un q]^\circ,  \un b), \dGa(\diag(k_j,k_j))]\\
& \quad +   \sum_{j=1}^\nu [\dGa(\un q, [\un k,\un q]^\circ,  \un b), \dGa(\diag(k_j,k_j))]\dGa(\diag(k_j,k_j)).
\end{align*}
 An analogous argument applies to
the term involving $\dGa(\un q, [\un k, \un b]^\circ)$ on the r.h.s of (\ref{Omega-formula-lemma}). \end{proof}

\bep\label{comm-H-dGa} Let $q_0,q_{\infty}$ be as specified in Definition~\ref{j-definition-a} and let $\un q^t=\diag(q_0^t,q_{\infty}^t)$
be the corresponding propagation observable on $\mfh\oplus\mfh$.
Let $\real\ni t\to b_j(t)\in B(\mfh)$, $j\in \{0,\infty\}$,  be two 
 families of admissible  operators s.t. $b_0$ is regular. 
Let $\un b=\diag(b_0,b_\infty)$ be the corresponding propagation observable on $\mfh\oplus\mfh$. Then, setting $\un q = \un q^t$, 
$R_{1,0}:=(1+H_{1,0}(\xi))^{-1}$ and $R_{0}:=(1+H_{0}(\xi))^{-1}$, we obtain
\begin{align}
[H_1(\xi),\dGa(\un q,\un b)]R_{1,0}^4&=-\nabla\Omeg(\xi-\dGa(\un k))\cdot \bigl(\dGa(\un q,[\un k,\un b]^\circ)+\dGa(\un q, [\un k, \un q]^\circ,  \un b)\bigr) R_{1,0}^4\non\\
& \quad +\bigl(\dGa(\un q,  [\un \om,\un b]^\circ)+\dGa(\un q, [\un \om, \un q]^\circ,  \un b) \bigr)R_{1,0}^4+O(t^{-2})
\label{dGa-comm-double}
\end{align}
and consequently
\begin{align}
[H(\xi),\dGa(q_0, b_0)]R_{0}^4&=-\nabla\Omeg(\xi-\dGa(k))\cdot\bigl(\dGa(q_0,[k, b_0]^\circ)+\dGa(q_0, [k, q_0]^\circ,  b) \big)R_{0}^4\non\\
& \quad +\big(\dGa(q_0,  [\om, b_0]^\circ)+\dGa(q_0, [\om, q_0]^\circ,   b_0) \big)R_{0}^4+O(t^{-2}).
\label{dGa-comm}
\end{align}
Each term on the r.h.s. of relations~(\ref{dGa-comm-double}) and (\ref{dGa-comm}) is bounded and $O(t^{-1})$. 
\eep
\begin{proof} Observe first that by Lemma~\ref{iterated-commutators}
and Definition~\ref{admissible},
  $\un q \in C^1(\un k_j)\cap C^1(\un \om)$, for each $j$, and $\un b\in C^1(\un\om)$. See also the first paragraph in the proof of 
  Lemma~\ref{comm-H-dGa-lemma} for notation and the observation that the bounded operators $[\un k,\un q]^\circ,[\un \omega,\un q]^\circ$ and $[\un \om,\un b]^\circ$
  are all $O(t^{-1})$.

 Let us first prove~(\ref{dGa-comm-double}).   
Making use of Lemma~\ref{comm-H-dGa-lemma}, and of the fact that $\Omeg\in S^{s_\Omega}(\real)$, we obtain the identity
\begin{align}
& [\Omeg(\xi-\dGa(\un k)), \dGa(\un q,\un b)]R_{1,0}^4\non\\
& \qquad=-\nabla\Omeg(\xi-\dGa(\un k))\cdot\bigl(\dGa(\un q, [\un k, \un b]^\circ)+\dGa(\un q, [\un k,q]^\circ,  \un b) \bigr) R_{1,0}^4+O(t^{-2})
\end{align}
in the sense of forms on $D(H_{1,0}(\xi))$.
All terms on the r.h.s. above are $O(t^{-1})$.

As for the second term from the free auxiliary fiber Hamiltonian~(\ref{AuxHam}), it suffices to note that Lemma~\ref{rest-terms} gives
\begin{align}
[\dGa(\un \om),\dGa(\un q,\un b)] R_{1,0}^4 
=\big(\dGa(\un q,[\un \om,\un q]^\circ,\un b)+\dGa(\un q, [\un \om, b]^\circ )\big)R_{1,0}^4 =O(t^{-1})
\end{align}
in the sense of forms on $D(H_{1,0}(\xi))$.
The interaction term in the Hamiltonian gives
\begin{align}
&[\phi(G,0),\dGa(\un q,\un b)]R_{1,0}^{4}
=\big(a^*((1-q_0)G,0)\dGa(\un q,\un b)-a^*(b_0G,0)\Ga(\un q)\non\\
 &\qquad+\Ga(\un q)a(b_0^*G,0)+\dGa(\un q,\un b)a((q_0-1)G,0)\big) R_{1,0}^{4}=O(t^{-2}),
\end{align}
where we made use of Lemma~\ref{q-p-lemma}, and exploited regularity of $b_0$ and $1-q_0$.  
This concludes the proof of (\ref{dGa-comm-double}). 

Now let us prove relation~(\ref{dGa-comm}). By conjugating (\ref{dGa-comm-double}) with the unitary $U$,
we get
\begin{align}
 &[H^{\ex}(\xi),\dGa^{\ex}(\un q,\un b)](R_{0}^{\ex})^4\non\\
&\quad=-\nabla\Omeg(\xi-\dGa^{\ex}(\un k))\cdot \bigl(\dGa^{\ex}(\un q,[\un k,\un b]^\circ)
+\dGa^{\ex}(\un q, [\un k, \un q]^\circ,  \un b)\bigr) (R_{0}^{\ex})^4\non\\
 &\qquad +\bigl(\dGa^{\ex}(\un q,  [\un \om,\un b]^\circ)+\dGa^{\ex}(\un q, [\un \om, \un q]^\circ,  \un b) \bigr)(R^{\ex}_{0})^4+O(t^{-2}),
\end{align}
where $R_{0}^{\ex}=(1+H_{0}^{\ex}(\xi))^{-1}$. By applying the projection $P_{0}$ on the subspace $\Fock\otimes\Fock^{(0)}\subset \Fock^{\ex}$ 
to both sides of this equality,  we obtain (\ref{dGa-comm}). \end{proof}

\bel\label{Helffer-dGa} Let $\un q^t$ and $\un b$ be as specified in Proposition~\ref{comm-H-dGa}. Let $\chi\in C_0^{\infty}(\real)_{\real}$. Then
\beq
[\chi(H_1(\xi)),\dGa(\un q^t,\un b)]=O(t^{-1}).\label{Helffer-dGa-formula}
\eeq 
\eel
\begin{proof} As above we abbreviate $R_{1,0}=(1+H_{1,0}(\xi))^{-1}$. Proposition~\ref{comm-H-dGa} gives
\beqa
[H_1(\xi),\dGa(\un q,\un b)]R_{1,0}^4=O(t^{-1}). \label{new-cGa-comm-two}
\eeqa
Now  we will use  the Helffer-Sj\"ostrand calculus. 
(See e.g.~\cite[Proposition~C.2.1]{DGBook}.)
We choose an almost analytic extension $\ti\chi\in C_0^{\infty}(\complex)$ of $\chi$ 
s.t.
\beqa
|\pa_{\bar z}\ti\chi(z)|\leq C_n|\Imm z|^n, \label{almost-analytic-property-one-two-three}
\eeqa
for $n\in \nat$ and write
\begin{align}
&  [\chi(H_1(\xi)),\dGa(\un q,\un b)]R_{1,0}^3  \non\\
&\qquad =\fr{\i}{2\pi}\int \,dz\wedge d\bar z\,\pa_{\bar z}\ti\chi(z) (z-H_{1}(\xi))^{-1} [H_1(\xi),\dGa(\un q,\un b)](z-H_1(\xi))^{-1}R_{1,0}^3.
\end{align}
Making use of relations (\ref{almost-analytic-property-one-two-three}), (\ref{new-cGa-comm-two}), and of the fact that   
$\|(1+H_{1,0}(\xi))^4(z-H_1(\xi))^{-1}(1+H_{1,0}(\xi))^{-3}\|\leq c/|\Imm z|$, proven in 
Corollary~\ref{Cor-technical}, we 
show that
\beqa
[\chi(H_1(\xi)),\dGa(\un q,\un b)]R_{1,0}^3=O(t^{-1}).  \label{Helffer-intermediate-relation}
\eeqa
Next, we choose a function $\ti\chi\in C_0^{\infty}(\real)_{\real}$ s.t. $\chi\ti\chi=\chi$ and write
\begin{align}
[\chi(H_1(\xi)),\dGa(\un q,\un b)]&= [\chi(H_1(\xi)),\dGa(\un q,\un b)]\ti\chi(H_1(\xi))\non\\
&\quad +\chi(H_1(\xi))[\ti\chi(H_1(\xi)),\dGa(\un q,\un b)].
\end{align}
Making use of (\ref{Helffer-intermediate-relation}), we conclude the proof. \end{proof}


\subsection{Commutators involving $\Ga(\,\cdot\,)$}

\bel\label{f-commutator} Let $j_0, j_{\infty}\in C^{\infty}(\real)$, $j_0',j_{\infty}'\in C_0^{\infty}(\real)$ and $0\leq j_0, j_{\infty}\leq 1$. 
Let $\un j^t:=\diag(j_0^t,j_{\infty}^t)$ be defined as a propagation observable on $\mfh\oplus\mfh$. 
Let $f\in S^{s_\Omega}(\real^{\nu})$. Then, setting $R_{1,0}(\xi)=(1+H_{1,0}(\xi))^{-1}$, we obtain
\begin{align}
 [f(\dGa(\un k)-\xi), \Ga(\un j^t)]R_{1,0}(\xi)^{3} 
  =\nabla f(\dGa(\un k)-\xi)\cdot \dGa(\un j^t,[\un k, \un j^t]^\circ)R_{1,0}(\xi)^{3}+O(t^{-2}), \label{f-commutator-double}
\end{align}
and the first term on the r.h.s. above is bounded and $O(t^{-1})$, with both $O$-symbols being uniform in $\xi\in\real^\nu$.
\eel
\begin{proof} Observe first that by Lemma~\ref{iterated-commutators}, we have $\un j^t\in C^1(\un k)$. The operator $[\un k,\un j^t]^\circ$ is bounded and $O(t^{-1})$ by Lemma~\ref{pseudo-lemma}.

We set $\un j:=\un j^t$ and write, making use of Lemma~\ref{many-A}
\begin{align}
[f(\dGa(\un k)-\xi), \Ga(\un j)](1+N_1)^{-3}&=\nabla f(\dGa(\un k)-\xi) \cdot \dGa(\un j,[\un k, \un j]^\circ)(1+N_1)^{-3}\non\\
& \quad +R(f, \dGa(\un k)-\xi,\Ga(\un j)(1+N_1)^{-3}),
\end{align}
as a form equality on $D(H_{1,0}(\xi))$. Here $N_1$ is the number operator on $\Ga(\mfh\oplus\mfh)$ and
we used the fact that  $f \in S^2(\real^{\nu})$ and that $\Ga(\un j)(1+N_1)^{-3}$  belongs to $C^3(\dGa(\un k)-\xi)$.
To justify this latter property, we note that for $|\al|\leq 3$
\beqa
\ad_{\dGa(\un k)-\xi}^{\al}(\Ga(\un j) )(1+N_1)^{-3}=O(t^{-|\al|}),
\eeqa
where we made use of Lemma~\ref{rest-terms}.   We obtain from Lemma~\ref{many-A} that
\begin{align}
& \|R(f, \dGa(\un k)-\xi,\Ga(\un j)(1+N_1)^{-3})\|\non \\
& \qquad  \leq \sum_{\al: 2\leq |\al| \leq 3 }\|\ad_{\dGa(\un k)-\xi}^{\al}(\Ga(\un j))(1+N_1)^{-3}\|=O(t^{-2}),
\end{align}
uniformly in $\xi$. In fact, the right-hand side does not actually depend on $\xi$.
Since $(1+N_1)^\ell(1+H_{1,0}(\xi))^{-\ell}$ is bounded uniformly in $\xi$ for any $\ell\in\nat$, we have shown that 
\beqa
[f(\dGa(\un k)-\xi), \Ga(\un j)]R_{1,0}(\xi)^{3} =\nabla f(\dGa(\un k)) \cdot \dGa(\un j,[\un k, \un j]^\circ)R_{1,0}^{3}(\xi)+O(t^{-2}),
\eeqa
uniformly in $\xi$ and in the sense of forms on $D(H_{1,0}(\xi))$.
To check that the first term on the r.h.s. above is bounded and of order  $O(t^{-1})$ uniformly in $\xi$,
we can as in the proof of Lemma~\ref{comm-H-dGa-lemma} assume that $s_\Omega\geq 1$ and argue in the exact same fashion, replacing however $\dGa(k)\cdot\dGa(k)+1$ by $(\dGa(k)-\xi)\cdot (\dGa(k)-\xi)+1$. We skip the details, which are slightly simpler here since there is only one term.
\end{proof}

\bep\label{comm-H-Ga} 
Let $j_0, j_{\infty}$ be as specified in Definition~\ref{j-definition-a}.
Let $\un j^t:=\diag(j_0^t,j_{\infty}^t)$ be defined as an operator on $\mfh\oplus\mfh$. 
Then, setting $R_{1,0}(\xi):=(1+H_{1,0}(\xi))^{-1}$, $R_{0}(\xi):=(1+H_{0}(\xi))^{-1}$, we get
\begin{align}
 &[H_1(\xi),\Ga(\un \jj^t)]R_{1,0}^{3}(\xi)\non\\
& \quad =\bigl(-\nabla\Omeg(\xi-\dGa(\un k)) \cdot \dGa(\un j^t,[\un k,\un j^t]^\circ)+\dGa(\un j^t,[\un \om,\un j^t]^\circ)\bigr)R_{1,0}(\xi)^{3} +O(t^{-2}),
\label{comm_H-Ga-rel-one}
\end{align}
uniformly in $\xi\in\real^\nu$. Consequently
\begin{align}
& [H(\xi),\Ga(\jj_0^t)]R_{0}^3(\xi)\non\\
&\quad =\bigl(-\nabla\Omeg(\xi-\dGa(k))\cdot \dGa(j_0^t,[k,j_0^t]^\circ)+\dGa(j_0^t,[\om,j_0^t]^\circ)\bigr)R_{0}^3(\xi)+O(t^{-2}),
\label{comm_H-Ga-rel-two}
\end{align}
uniformly in $\xi\in\real^\nu$
The explicit terms on the r.h.s. of relations~(\ref{comm_H-Ga-rel-one}), and (\ref{comm_H-Ga-rel-two}) are bounded and $O(t^{-1})$, uniformly in $\xi$.
\eep
\begin{proof}  Observe first that by Lemma~\ref{iterated-commutators}, we have $\un j^t\in C^1(\un k)\cap C^1(\un \om)$. The operators $[\un k,\un j^t]^\circ$
and $[\un \om, \un j^t]^\circ$ are bounded and $O(t^{-1})$ by Lemma~\ref{pseudo-lemma}.

 We set $\un j:=\un j^t$. Lemma~\ref{f-commutator} gives that
\begin{equation}
 [\Omeg(\xi-\dGa(\un k)), \Ga(\un j)]R_{1,0}(\xi)^{3}
=-\nabla\Omeg(\xi-\dGa(\un k))\cdot \dGa(\un j,[\un k, \un j]^\circ)R_{1,0}(\xi)^{3}+O(t^{-2}),
\end{equation}
and all terms on the r.h.s. above are $O(t^{-1})$ uniformly in $\xi$.

As for the second term from the free auxiliary Hamiltonian~\eqref{AuxHam}, it suffices to note that
 uniformly in $\xi$ we have 
\beqa
[\dGa(\un\om),\Ga(\un j)]R_{1,0}(\xi)^{3} =\dGa(\un j,[\un\om,\un j]^\circ)R_{1,0}(\xi)^{3}=O(t^{-1}).
\eeqa
The interaction term from the interacting auxiliary Hamiltonian~\eqref{AuxHam}  contributes to $O(t^{-2})$. To show this,
we recall the relations
\begin{equation}
\Ga(\un j)a^*(G,0) =a^*(\un j (G,0))\Ga(\un j)\quad \textup{and}\quad
a(G,0)\Ga(\un j) =\Ga(\un j)a(\un j (G,0)), \label{aGa-com}
\end{equation}
which hold on $\Ga_{\fin}(\mfh\oplus\mfh)$ and imply that
\begin{align}
&[\phi(G,0), \Ga(\un j)] R_{1,0}(\xi)^{3}\non\\
&\qquad =\bigl(\Ga(\un j)a((j_0-1)G,0)-a^*((j_0-1)G,0)\Ga(\un j)\bigr)R_{1,0}(\xi)^{3} =O(t^{-2}),
\end{align}
uniformly in $\xi$.
In the last step we made use of the fact that $j_0-1$ is regular and of Lemma~\ref{G-properties}. 
This concludes the proof of (\ref{comm_H-Ga-rel-one}). 

Now let us prove relation~(\ref{comm_H-Ga-rel-two}). By conjugating (\ref{comm_H-Ga-rel-one}) with the unitary $U$,
we get
\begin{align}
& [H^{\ex}(\xi),\Ga^{\ex}(\un \jj)]R_{0}^{\ex}(\xi)^{3}\non\\
&\quad =\bigl(-\nabla\Omeg(\xi-\dGa^{\ex}(\un k))\cdot\dGa^{\ex}(\un j,[\un k,\un j]^\circ)+\dGa^{\ex}(\un j,[\un \om,\un j]^\circ)\bigr)
R_{0}^{\ex}(\xi)^{3} +O(t^{-2}),
\end{align}
where $R_{0}^{\ex}(\xi)=(1+H_{0}^{\ex}(\xi))^{-1}$. By applying the projection $P_{0}$ on the subspace $\Fock\otimes\Fock^{(0)}\subset \Fock^{\ex}$ 
to both sides of this equality,  we obtain (\ref{comm_H-Ga-rel-two}). \end{proof}

\bel\label{j_0Lemma} Let $\chi\in C_0^{\infty}(\real)_{\real}$ and $j_0, j_{\infty}$ be as specified in Definition~\ref{j-definition-a}.
Let $\un j^t:=\diag(j_0^t,j_{\infty}^t)$ be defined as an operator on $\mfh\oplus\mfh$. 
Then there holds the estimate
\beqa
[\chi(H_1(\xi)+\lambda),\Ga(\un \jj^t)]=O(t^{-1}),
\eeqa
uniformly in $\xi\in\real^\nu$ and $\lambda\geq 0$.
\eel
\begin{proof} This lemma follows from Proposition~\ref{comm-H-Ga} by the method of almost analytic extensions. (Cf. the proof of Lemma~\ref{Helffer-dGa}). \end{proof}


\subsection{Commutators involving $\cGa(\,\cdot\,)$}

\bel\label{f-commutator-ex} Let $j_0, j_{\infty}\in C^{\infty}(\real)$, $j_0',j_{\infty}'\in C_0^{\infty}(\real)$, $0\leq j_0, j_{\infty}\leq 1$,
and $j_0^2+j_{\infty}^2\leq 1$.
Let $ j^t:=(j_0^t,j_{\infty}^t)$ be defined as an operator from
$\mfh$ to  $\mfh\oplus\mfh$. 
Let $f\in S^{s_\Omega}(\real^{\nu})$. Then, setting $R_0(\xi) = (1+H_0(\xi))^{-1}$ and $R_{0}^{\ex}(\xi)=(1+H_{0}^{\ex}(\xi) )^{-1}$, we obtain
\begin{align}
& \bigl(f(\dGa^{\ex}(k)-\xi)\cGa(j^t)-\cGa(j^t)f(\dGa(k)-\xi)\bigr)R_{0}(\xi)^{3}\non\\
& \qquad = \nabla f(\dGa^{\ex}(k)-\xi)\cdot \cdGa^{\ex}(j^t,[\un k, j^t]^\circ)R_{0}(\xi)^{3}+O(t^{-2}), \label{f-commutator-ex-one}\\
& R_{0}^{\ex}(\xi)^{3} \bigl(f(\dGa^{\ex}(k)-\xi)\cGa(j^t)-\cGa(j^t)f(\dGa(k)-\xi)\bigr)\non\\
& \qquad = R_{0}^\ex(\xi)^{3} \nabla f(\dGa^{\ex}(k)-\xi)\cdot \cdGa^{\ex}(j^t,[\un k, j^t]^\circ)+O(t^{-2}), \label{f-commutator-ex-two}
\end{align}
uniformly in $\xi\in\real^\nu$. 
Furthermore, all explicit terms on the r.h.s. of (\ref{f-commutator-ex-one}) and (\ref{f-commutator-ex-one}) above are bounded and $O(t^{-1})$, uniformly in $\xi\in\real^\nu$. 
\eel

\begin{remark}\label{Rem-MixedCommForm} Here $[\un k, j^t]^\circ$ is a $\nu$-vector 
with entries $[\un k_i, j^t]^\circ$, which is itself a $2$-vector
$([k_i,j_0^t]^\circ,[k_i,j_\infty^t]^\circ)$ with bounded operators as components, obtained through extension by continuity of the  form $[\un k_i, j^t]=\un k_i j^t-j^t k_i$
densely defined on $D(\un k)\times  D(k)$.  Recall that by 
Lemma~\ref{iterated-commutators}, we have $j^t_0,j_\infty^t\in C^1(k_i)$, for each $i$. 
The vector operator $[\un k, j^t]^\circ$ is bounded and $O(t^{-1})$ by Lemma~\ref{pseudo-lemma}.
\end{remark}

\begin{proof} 
 As in \cite{MR12}, we define 
\begin{equation}\label{AuxAuxGammas}
\cGa^{\ex}(q)=\cGa(q)P_{0} \quad \textup{and} \quad \cdGa^{\ex}(q,p)=\cdGa(q,p)P_{0},
\end{equation}
 where $P_{0}\colon\Fock^{\ex}\to \Fock$ is the natural restriction to the subspace $\Fock\otimes\Fock^{(0)}=\Fock$ in $\Fock^{\ex}$. (The notation $\cGa^{\ex}$ and $\cdGa^{\ex}$ is used
only in this proof).  After identifying $\Fock\otimes\Fock^{(0)}$ with $\Fock$, we can write
\begin{equation}
\begin{aligned}
& \bigl(f(\dGa^{\ex}(k)-\xi)\cGa(j^t)-\cGa(j^t)f(\dGa(k)-\xi)\bigr)R_{0}(\xi)^{3}\\
& \qquad =[f(\dGa^{\ex}(k)-\xi), \cGa^{\ex}(j^t)]R_{0}^{\ex}(\xi)^{3}.
\end{aligned}
\end{equation}
Next, we set $j:=j^t$ and write, making use of Lemma~\ref{many-A},
\begin{align}
& [f(\dGa^{\ex}(k)-\xi), \cGa^{\ex}(j)](1+N^{\ex})^{-3} \non \\ 
 & \qquad  =\nabla f(\dGa^{\ex}(k)-\xi)\cdot \cdGa^{\ex}(j,[\un k,  j]^\circ)(1+N^{\ex})^{-3} \\
\non & \qquad \quad +R(f, \dGa^{\ex}(k)-\xi,\cGa^{\ex}(j)(1+N^{\ex})^{-3}),
\end{align}
as a form identity on $D(H^\ex_0(\xi))$. Here $N^{\ex}$ is the number operator on $\Fock^{\ex}$.
We used Lemma~\ref{high-order-commutators}, the assumption that  $f \in S^{s_\Omega}(\real^{\nu})$ and that $\cGa^{\ex}(j)(1+N^\ex)^{-3}$  belongs to $C^3(\dGa^{\ex}(k)-\xi)$.
To justify this latter property, we note that by Lemma~\ref{rest-terms-double}
\beqa
\ad_{\dGa^{\ex}(k)-\xi}^{\al}(\cGa^{\ex}(j) )(1+N^{\ex})^{-3}=O(t^{-|\al|}),
\eeqa
uniformly in $\xi$, for $|\al|\leq 3$ (the expression is in fact $\xi$-independent).
Thus we obtain from Lemma~\ref{many-A} that
\begin{align}
& \|R(f, \dGa^{\ex}(k)-\xi,\cGa^{\ex}(j)(1+N^{\ex})^{-3})\|\non\\
& \qquad \leq \sum_{\al: 2\leq |\al| \leq 3 }\|\ad_{\dGa^{\ex}(k)-\xi}^{\al}(\cGa^{\ex}(j))(1+N^{\ex})^{-3}\|=O(t^{-2}),
\end{align}
uniformly in $\xi$.
Since $(1+N^{\ex})^\ell(1+H_{0}^{\ex}(\xi))^{-\ell}$ is bounded uniformly in $\xi$ for any $\ell\in\nat$, we have shown that uniformly in $\xi$ we have
\begin{align}
[f(\dGa^{\ex}(k)-\xi), \cGa^{\ex}(j)]R_{0}^{\ex}(\xi)^{3} &=\nabla f(\dGa^{\ex}( k)-\xi)\cdot \cdGa^{\ex}( j,[\un k,  j]^\circ)R_{0}^{\ex}(\xi)^{3}+O(t^{-2}),\label{cGammaResOnRight}\\
R_{0}^{\ex}(\xi)^{3}[f(\dGa^{\ex}(k)-\xi), \cGa^{\ex}(j)] &= R_{0}^{\ex}(\xi)^{3}\nabla f(\dGa^{\ex}( k)-\xi)\cdGa^{\ex}( j,[\un k,  j]^\circ)+O(t^{-2}).
\end{align}
To check that all the terms on the r.h.s. of the above relations  are $O(t^{-1})$ uniformly in $\xi$, we can assume that $s_\Omega\geq 1$ and write
\begin{align}
& (R_{0}^{\ex})^{3}\nabla f(\dGa^{\ex}(k))\cdot \cdGa^{\ex}(j,[\un k,  j]^\circ) \non\\
& \qquad =(1+N^{\ex})(R_{0}^{\ex})^{3}\nabla f(\dGa^{\ex}(k))(1+N^{\ex})^{-1}\cdot \cdGa^{\ex}(j,[\un k,  j]).
\end{align}
For \eqref{cGammaResOnRight} we argue as at the end of the proofs of Lemmata~\ref{comm-H-dGa-lemma}
and~\ref{f-commutator}, inserting
$I = (1+(\dGa(\un k)-\xi)^2)^{-1} (1+(\dGa(\un k)-\xi)^2)$ and commuting the second factor
onto the resolvent on the right to obtain 
\begin{align}
& \nabla f(\dGa^{\ex}(k)-\xi)\cdot \cdGa^{\ex}(j, [\un k,  j])R^{\ex}_{0}(\xi)^{3}\\
&  =-\nabla f(\dGa^{\ex}(k)-\xi)  (1+(\dGa(\un k)-\xi)^2)^{-1} \cdot  [\cdGa^{\ex}(j,[\un k,  j]^\circ),(\dGa(\un k)-\xi)^2] R^{\ex}_{0}(\xi)^{3} \non\\
\non &\quad +\nabla f(\dGa^{\ex}(k)-\xi)   (1+(\dGa(\un k)-\xi)^2)^{-1}  \cdot \cdGa^\ex(j,[\un k,  j]^\circ)  (1+(\dGa(\un k)-\xi)^2)R^{\ex}_{0}(\xi)^3.
\end{align}
Here $(\dGa(k)-\xi)^2 = (\dGa(k)-\xi)\cdot (\dGa(k)-\xi)$.
Recalling \eqref{AuxAuxGammas} and that  $|\pa_if(\eta)|\leq c\lan \eta \ran$, and making use of Lemma~\ref{rest-terms-double}, we conclude the proof. \end{proof}

\bep \label{H-cGa-comm} Let $j_0, j_{\infty}$ be as specified in Definition~\ref{j-definition-a} and s.t. $j_0^2+j_{\infty}^2\leq 1$.
Put $ j^t=(j_0^t,j_{\infty}^t)\colon \mfh\to\mfh\oplus\mfh$.
 Then, setting $R_{0}(\xi)=(1+H_{0}(\xi))^{-1}$ and  $R_{0}^{\ex}(\xi)=(1+H_{0}^{\ex}(\xi))^{-1}$, we obtain uniformly in $\xi\in\real^\nu$ the asymptotic expansions
\begin{align}
 &\bigl(H^{\ex}(\xi)\cGa(j^t)-\cGa(j^t)H(\xi)\bigr)R_{0}^3(\xi)\non\\
& \qquad =\bigl(-\nabla\Omeg(\xi-\dGa^{\ex}(k))\cdot \cdGa(j^t,[\un k, j^t]^\circ)+\cdGa(j^t, [\un \om, j^t]^\circ) \bigr)R_{0}^{3}(\xi)+O(t^{-2})
\label{Ga-check-comm-one}
\end{align}
and
\begin{align}
  &R_{0}^{\ex}(\xi)^3\bigl(H^{\ex}(\xi)\cGa(j^t)-\cGa(j^t)H(\xi)\bigr)\non\\
&\qquad =R_{0}^{\ex}(\xi)^3\bigl(-\nabla\Omeg(\xi-\dGa^{\ex}(k)) \cdot\cdGa(j^t,[\un k, j^t]^\circ)+\cdGa(j^t, [\un \om, j^t]^\circ) \bigr)+O(t^{-2}),
\label{Ga-check-comm-two}
\end{align}
and all explicit terms on the r.h.s. of relations (\ref{Ga-check-comm-one}) and (\ref{Ga-check-comm-two}) are bounded and $O(t^{-1})$ uniformly in $\xi\in\real^\nu$.
\eep
\begin{proof} We prove only relation~(\ref{Ga-check-comm-one}), as the proof of (\ref{Ga-check-comm-two}) is analogous.
 Observe again that by Lemma~\ref{iterated-commutators}, we have $\un j^t\in C^1(\un k)\cap C^1(\un \om)$. The operators $[\un k,j^t]^\circ$
and $[\un \om,  j^t]^\circ$ are bounded and $O(t^{-1})$ by Lemma~\ref{pseudo-lemma}.
See also Remark~\ref{Rem-MixedCommForm} for a more thorough
explanation of the notation.

 Lemma~\ref{f-commutator-ex} gives
\begin{align}
& \bigl(\Omeg(\xi-\dGa^{\ex}(k))\cGa(j^t)-\cGa(j^t)\Omeg(\xi-\dGa(k))\bigr)R_{0}(\xi)^3\non\\
& \qquad =-\nabla\Omeg(\xi-\dGa^{\ex}(k))\cdGa(j^t,[\un k_i, j^t]^\circ)R_{0}(\xi)^{3}+O(t^{-2}),
\end{align}
and all terms on the r.h.s. above are $O(t^{-1})$ uniformly in $\xi$. As for the second term in the extended Hamiltonian~(\ref{extended-fiber-hamiltonian}), 
we obtain
\beqa
(\dGa^{\ex}(\om)\cGa(j^t)-\cGa(j^t)\dGa(\om))R_{0}(\xi)^3=\cdGa(j^t,[\un \om, j^t]^\circ)R_{0}(\xi)^3,
\eeqa
where we made use of Lemma~\ref{high-order-commutators}. It is clear that this expression is $O(t^{-1})$ uniformly in $\xi$.

Finally, we consider the interaction term from Hamiltonian~(\ref{extended-fiber-hamiltonian}). There holds
\begin{align}
& \bigl((\phi(G)\otimes 1)\cGa(j^t)-\cGa(j^t)\phi(G)\bigr)R_{0}(\xi)^3\\
& \qquad =\bigl((a^*((1-j_0^t)G)\otimes 1+1\otimes a^*(j^t_{\infty}G))\cGa(j^t)+\cGa(j^t)a((j_0^t-1)G)\bigr)R_{0}(\xi)^3.\non
\end{align}
Since $j_0^t-1$ and $j^t_{\infty}$ are regular propagation observables,  this expression is $O(t^{-2})$ (uniformly in $\xi$) by Lemma~\ref{G-properties}.
This concludes the proof. \end{proof} 
\bel\label{Helffer} Let $j_0, j_{\infty}$ be as specified in Definition~\ref{j-definition-a} and s.t. $j_0^2+j_{\infty}^2\leq 1$.
Put $j^t=(j_0^t,j_{\infty}^t)\colon \mfh\to\mfh\oplus\mfh$.
There holds the relation
\begin{equation}
\bigl(\chi(H^{\ex}(\xi)+\lambda)\cGa(j^t)-\cGa(j^t)\chi(H(\xi)+\lambda)\bigr)  =O(t^{-1}), \label{Helffer-one}
\end{equation}
uniformly in $\xi\in\real^\nu$ and $\lambda\geq 0$.
\eel
\begin{proof} This lemma follows from Proposition~\ref{H-cGa-comm} by the method of almost analytic extensions. (Cf. the proof of Lemma~\ref{Helffer-dGa}). \end{proof}

The following corollary about domain invariance follows directly from
Propositions~\ref{comm-H-Ga} and~\ref{H-cGa-comm}.

\bec\label{Cor-DomInv-Gamma} Let $j_0, j_{\infty}$ be as specified in 
Definition~\ref{j-definition-a} and s.t. $j_0^2+j_{\infty}^2\leq 1$. Suppose $q\in C^\infty(\real)$, with $0\leq q\leq 1$, is bounded with bounded derivatives. Then
\begin{align}
& \Ga(q^t)\colon D(H(\xi)^3) \to D(H(\xi)),\\
& \cGa(j^t)\colon D(H(\xi)^3)\to D(H^\ex(\xi)),\\
& \cGa(j^t)^*\colon D(H^\ex(\xi)^3)\to D(H(\xi)).
\end{align}
\eec

We do not believe the third power in the corollary above is optimal.
It is however sufficient for our purpose.
A similar result was derived in \cite{MAHP} for localizations in configuration space.

\section{Auxiliary results for the proof of Proposition~\ref{a-dGa-Heisenberg}} \label{First-proposition-appendix}
\setcounter{equation}{0}

\bep\label{l=0} Let $\chi\in C_0^{\infty}(\real)_{\real}$. 
Let $q\in C^{\infty}(\real)$ be s.t. $q'\in C_0^{\infty}(\real)$ and $0\leq q\leq 1$, 
 and let $b$ be an admissible and regular propagation observable. Let $j_0$ be as specified in   
Definition~\ref{j-definition-a} and s.t. $\supp\, j_0\subset \De$, where $\De$ appeared in
Definition~\ref{admissible}.  Then
\begin{align}
& \chi\dGa(q^t,b)\chi\Ga(\jj_0^t)=O(\R^{-1}), \label{disjoint-supports-two}\\
& \chi[H(\xi),\dGa(b)]\chi\Ga(\jj_0^t)=O(\R^{-2}),\label{another-sandwiched-commutator}
\end{align}
where we set $\chi:=\chi(H(\xi))$.
\eep
\begin{proof}   To prove  (\ref{disjoint-supports-two})
we write
\begin{equation}
\chi\dGa(q^t,b)\chi\Ga(\jj^t_0)=\chi\dGa(q^t,b)[\chi, \Ga(\jj^t_0)]+O(t^{-2}),
\end{equation}
where we exploited regularity of $b$. The first term on the r.h.s.  is of order $O(\R^{-1})$ by Lemma~\ref{j_0Lemma}.

To verify (\ref{another-sandwiched-commutator}), we make use  of Proposition~\ref{comm-H-dGa}, which gives
\begin{align}
\chi[H(\xi),\dGa(b)]\chi\Ga(\jj^t_0) &= -\chi\nabla\Omeg(\xi-\dGa(k))\cdot \dGa([k,b]^\circ)  [\chi, \Ga(\jj^t_0)]\non\\
& \quad +\chi\dGa([\om,b]^\circ)[\chi,\Ga(\jj^t_0)] +O(t^{-2}), \label{double-j_0-commutator}
\end{align}
where we exploited regularity of $b$. The first two terms on the r.h.s. above are $O(t^{-2})$ by admissibility of $b$, cf.~Definition~\ref{admissible},
and Lemma~\ref{j_0Lemma}. \end{proof}

\bep \label{derivative-term} Let $q\in C^{\infty}(\real)$ be s.t. $q'\in C_0^{\infty}(\real)$ and $0\leq q\leq 1$.
Let $b$ be an admissible propagation observable and
 $j_0, j_{\infty}$ be as specified in Definition~\ref{j-definition-a}, with $j_0^2+j_{\infty}^2=1$. 
Then
\begin{align}
&\chi(H(\xi)+\lambda)\dGa(q^t,b)\chi(H(\xi)+\lambda)\non \\
& \qquad =\cGa(j^t)^*\chi^{\ex}(H^\ex(\xi)+\lambda)  \dGa^{\ex}(q^t,b)\chi^{\ex}(H^\ex(\xi)+\lambda)\cGa(j^t)+O(\R^{-1}),
\end{align}
uniformly in $\xi\in\real^\nu$ and $\lambda\geq 0$.

\eep
\begin{proof} We set $j:=j^t$, $q:=q^t$, $\chi:=\chi(H(\xi)+\lambda)$ and $\chi^{\ex}:=\chi(H^{\ex}(\xi)+\lambda)$. The reader is asked to keep in mind that $\chi$ and $\chi^\ex$ depend on both $\xi$ and $\lambda$. Write
\beqa
\chi\dGa(q, b )\chi=\cGa(j)^*\bigl(\cGa(j)\chi-\chi^{\ex}\cGa(j)\bigr)
\dGa(q, b)\chi
+\cGa(j)^*\chi^{\ex}\cGa(j)\dGa(q,b) \chi. \label{derivative-term-identity}
\eeqa
The first term on the r.h.s. above is of order $O(\R^{-1})$ uniformly in $\xi$ and $\lambda\geq 0$ by  Lemma~\ref{Helffer}. 
The last term on the r.h.s. of (\ref{derivative-term-identity}) can be rearranged as follows
\begin{align}
\cGa(j)^*\chi^{\ex}\cGa(j)\dGa(q,b) \chi &=\cGa(j)^*\chi^{\ex}\dGa^{\ex}(q,b)\cGa(j) \chi +\cGa(j)^*\chi^{\ex}\dcGa(jq,[j,\un b])\chi\non\\
&= \cGa(j)^*\chi^{\ex}\dGa^{\ex}(q,b)\cGa(j) \chi+O(t^{-1}), \label{derivative-term-second-step}
\end{align}
uniformly in $\xi$ and $\lambda\geq 0$.
Here  we made use of Lemma~\ref{high-order-commutators} and 
admissibility of $b$, cf.~Definition~\ref{admissible}. (Note that $[j,\un b] = jb-\un b j$ and $\un b=\diag(b,b)$ is a bounded propagation observable on $\mfh\oplus\mfh$). To exchange $\cGa(j)\chi$ with
$\chi^{\ex}\cGa(j)$ in (\ref{derivative-term-second-step}) we use again Lemma~\ref{Helffer}. This concludes the proof. \end{proof} 

\bel\label{Helffer-second-component-lemma} Let $q\in C^{\infty}(\real)_{\real}$ be s.t.  $q'\in C_0^{\infty}(\real)$. Let $\chi\in C_0^{\infty}(\real)_{\real}$.
Then
\beq
[1\otimes q^t,\chi(H^{(1)}(\xi))]=O(t^{-1}). \label{Helffer-second-component-formula}
\eeq
\eel
\begin{proof}   We follow the strategy explained in Remark~\ref{Remark-AuxHamSuffices}. By setting in  Lemma~\ref{Helffer-dGa} 
$\un q^t=\diag(1,1)$, $b=(0,q^t)$, conjugating formula~(\ref{Helffer-dGa-formula})
with the unitary $U$ and applying the projection  $P_{1}$ on the subspace $\Fock\otimes\Fock^{(1)}\subset\Fock^{\ex}$
we obtain (\ref{Helffer-second-component-formula}). \end{proof}

\bec\label{resolvents} Let $q,p\in C^{\infty}(\real)$ be s.t. $q',p'\in C_0^{\infty}(\real)$ and $0\leq q\leq 1$.
Let $b$ be  an admissible and regular propagation observable. 
Let $\chi\in C_0^{\infty}(\real)_{\real}$ be supported in $(-\infty, \Sigma_0^{(2)}(\xi))$. 
Then
\beq
\chi(H^{(1)}(\xi))(\dGa(q^t,b)\otimes p^t)\chi(H^{(1)}(\xi))=O(t^{-1}).\label{resolvents-formula}
\eeq
\eec
\begin{proof} We set $q:=q^t$, $p=p^t$ and  choose $\chi_0\in C_0^{\infty}(\real)$, supported in  $(-\infty, \Sigma_0^{(2)}(\xi))$
and s.t. $\chi=\chi_0\chi$. Then, abbreviating $\chi^{(1)}:= \chi(H^{(1)}(\xi))$ and  making use of the fact that $\chi^{(1)}(\dGa(q,b)\otimes 1)=O(1)$, we obtain
from Lemma~\ref{Helffer-second-component-lemma} that
\beqa
\chi^{(1)}(\dGa(q,b)\otimes p)\chi^{(1)}=\chi^{(1)}(\dGa(q,b)\otimes 1)\chi^{(1)}(1\otimes p) \chi_0^{(1)}+O(t^{-1}).
\eeqa 
Thus it suffices to prove~(\ref{resolvents-formula}) with $p=1$.
We rewrite this expression  as a direct integral
\beqa
\int^{\oplus} dk\, \chi(H(\xi-k)+\om(k))\,\dGa(q,b)\,\chi(H(\xi-k)+\om(k)).
\eeqa
In order to establish that the above expression is $O(t^{-1})$, it suffices to argue that
\begin{equation}\label{SuffBelowThreshold}
\chi(H(\xi-k)+\omega(k))\,\dGa(q,b)\,\chi(H(\xi-k)+\omega(k)) = O(t^{-1}),
\end{equation}
uniformly in $k\in\real^\nu$.

For $k\in\real^\nu$, define the function
$\chi_k(s):=\chi(s+\om(k))$.
It is easily seen that $\chi_k \in  C_0^{\infty}(\real)$ 
is supported in $(-\infty, \Sigma_0^{(1)}(\xi-k))$. Indeed, If $s+\omega(k)\in \supp\,\chi$,
then $s+\omega(k) < \Sigma_0^{(2)}(\xi)\leq \Sigma_0^{(1)}(\xi-k)+\omega(k)$.

Now let $j_0, j_{\infty}$ be as specified in Definition~\ref{j-definition-a}, s.t. $j_0^2+j_{\infty}^2=1$ and $j_0$ is supported in the
set $\De$ specified in Definition~\ref{admissible}. We set $j:=j^t$ below. Then,
\begin{align}
& \chi_k(H(\xi-k))\,\dGa(q,b)\, \chi_k(H(\xi-k))\non\\
& \qquad =\cGa(j)^* \chi_k(H^{\ex}(\xi-k))\,\dGa^{\ex}(q,b) \,\chi_k(H^{\ex}(\xi-k))\cGa(j)+O(t^{-1})\non\\
& \qquad =\Ga(j_0)\chi_k(H(\xi-k))\,\dGa(q,b)\,\chi_k(H(\xi-k))\Ga(j_0)+O(t^{-1}).
\end{align} 
Here in the first step we made use of Proposition~\ref{derivative-term}, which in particular ensures that the asymptotic expansion above is uniform in $k\in\real^\nu$. In the second step we applied the decomposition~(\ref{Hamiltonian-decomp})
of $H^{\ex}(\xi-k)$ and observed that, due to the support property of $\chi_k$, only the $l=0$ term is non-zero. 
Next, making use of Lemma~\ref{high-order-commutators}, we get
\begin{align}
\Ga(j_0)\chi_k(H(\xi-k))\dGa(q,b)\chi_k(H(\xi-k))
&= [\Ga(j_0),\chi_k(H(\xi-k))] \dGa(q,b)\chi_k(H(\xi-k))\non\\
&\quad + \chi_k(H(\xi-k))\dGa(j_0q,j_0b)\chi_k(H(\xi-k)).
\end{align}
This expression is of order $O(t^{-1})$, uniformly in $k$, by Lemma~\ref{j_0Lemma}  and regularity of $b$. This concludes the verification of \eqref{SuffBelowThreshold}, and hence we have established the corollary.\end{proof}

\bep \label{gamma-check}  Let $b$ be an admissible and regular propagation observable. 
Let $\chi\in C_0^{\infty}(\real)_{\real}$ and  $j_0, j_{\infty}$ be as specified in Definition~\ref{j-definition-a} and s.t. $j_0^2+j_{\infty}^2=1$. Then
\beqa
 \chi[H(\xi),\dGa(b)]\chi
=\cGa(j^t)^*\chi^{\ex}[H^{\ex}(\xi), \dGa^{\ex}(b)]\chi^{\ex}\cGa(j^t)+O(t^{-2}),
\label{main-equation}
\eeqa
where we set $\chi:=\chi(H(\xi))$ and $\chi^{\ex}:=\chi(H^{\ex}(\xi))$.
\eep
\begin{proof} We set $j:=j^t$ and write
\beqa
\chi[H(\xi),\dGa(b)]\chi=\cGa(j)^*\chi^{\ex}\cGa(j)[H(\xi),\dGa(b)]\chi+O(\R^{-2}),
\eeqa 
where we made use of the fact that, by Proposition~\ref{comm-H-dGa}, $[H(\xi),\dGa(b)]\chi=O(t^{-1})$.
Next, we will show that 
\beqa
\chi^{\ex}\bigl(  \cGa(j) [H(\xi),\dGa(b)]-[H^{\ex}(\xi), \dGa^{\ex}(b) ]\cGa(j) \bigr)\chi=O(t^{-2}).
 \label{long-commutator}
\eeqa
In view of Proposition~\ref{comm-H-dGa}, it is enough to check that
\begin{align}
& \chi^{\ex}\cGa(j)\pa_i\Omeg(\xi-\dGa(k))\dGa([k_i,b]^\circ)\chi \non\\
& \qquad =\chi^{\ex}\pa_i\Omeg(\xi-\dGa^{\ex}(k))\dGa^{\ex}([\un k_i,\un b]^\circ) \cGa(j)\chi+O(t^{-2}),\label{comm-rel-one} 
\end{align}
and
\beqa
\chi^{\ex}\cGa(j)\dGa([\om,b]^\circ)\chi= \chi^{\ex}\dGa^{\ex}([\om, b]^\circ)\cGa(j)+O(t^{-2}). \label{comm-rel-two}
\eeqa
We prove only (\ref{comm-rel-one}), since the proof of (\ref{comm-rel-two}) is analogous (and simpler). First, we note that,
\begin{align}
& \chi^{\ex}\cGa(j)\pa_i\Omeg(\xi-\dGa(k))\dGa([k_i,b]^\circ)\chi\non\\
& \qquad =\chi^{\ex}\pa_i\Omeg(\xi-\dGa^{\ex}(k))\cGa(j)\dGa([k_i,b]^\circ)\chi+O(t^{-2}),
\end{align}
where we exploited Lemma~\ref{f-commutator-ex}. Next, making use of Lemma~\ref{high-order-commutators}, we obtain
\begin{align}
\cGa(j)\dGa([k_i,b]^\circ)\chi&= \dGa^{\ex}([k_i,b]^\circ)\cGa(j)\chi+\cGa(j, [j,[\un k_i,\un b]^\circ] )\chi\non\\
&= \dGa^{\ex}([k_i,b]^\circ)\cGa(j)\chi+O(t^{-2}),
\end{align}
where we used admissibility of $b$, cf.~Definition~\ref{admissible}. Thus we have justified~\eqref{comm-rel-one}.

To conclude the proof,  it suffices  to show that
\beqa
\chi^{\ex} [H^{\ex}(\xi), \dGa^{\ex}(b)]\bigl(\cGa(j)\chi-\chi^{\ex}\cGa(j)\bigr)
=O(t^{-2}). \label{last-step-comm}
\eeqa
This follows from  Proposition~\ref{comm-H-dGa} and Lemma~\ref{Helffer}. \end{proof}

\bel\label{decay-of-dGa-term} Let $\chi\in C_0^{\infty}(\real)_{\real}$ be supported in $(-\infty, \Sigma_0^{(2)}(\xi))$. Let   $b_0$  be an admissible and regular propagation observable. Then 
\beq
\chi(H^{(1)}(\xi))[H^{(1)}(\xi), \dGa(b_0)\otimes 1]\chi(H^{(1)}(\xi))=O(t^{-2}).
\eeq
\eel
\begin{proof}  By conjugating  formula~(\ref{dGa-comm-double}) with the unitary $U$,  setting $\un q^t=\diag(1,1)$ and $\un b=\diag(b_0,0)$, we obtain 
\begin{align}
& [H^{\ex}(\xi),\dGa^{\ex}(\un b)]\\
&\non  \quad =\Bigl(\dGa^{\ex}([\un \om, \un b]^\circ)-\sum_{i=1}^\nu \pa_i\Omeg(\xi-\dGa^{\ex}(k))\dGa^{\ex}([\un k_i, \un b]^\circ)\Bigr)+O(t^{-2})(H^\ex_0(\xi)+1)^4,
\label{dGa-comm-double-one}
\end{align}
as a form identity on $D(H^\ex_0(\xi)^4)$.
Now we apply the projection $P_{1}$ on the subspace $\Fock\otimes\Fock^{(1)}\subset \Fock^{\ex}$ to both sides of this equality
and insert  both sides between the operators $\chi^{(1)}:= \chi(H^{(1)}(\xi))$. We get
\begin{align}
&\chi^{(1)}[H^{(1)}(\xi), \dGa(b_0)\otimes 1]\chi^{(1)}\\
\non&\quad =\chi^{(1)}\Bigl(-\nabla\Omeg(\xi-\dGa^{(1)}(k))\cdot (\dGa([k, b_0]^\circ)\otimes 1)+(\dGa([\om, b_0]^\circ)\otimes 1)\Bigr)\chi^{(1)}+O(t^{-2}),
\end{align}
where we used the higher order domain result in Corollary~\ref{Cor-technical2} for $H^\ex(\xi)$.
In view of Lemma~\ref{Helffer-dGa}, it suffices to show that
\beqa
\chi^{(1)}(\dGa([g, b_0]^\circ)\otimes 1)\chi^{(1)}=O(t^{-2})
\eeqa
for any $g\in C^{\infty}(\real^{\nu})$, whose derivatives (of non-zero order) are bounded.  This follows from
Corollary~\ref{resolvents} and the fact that $t\to t[g,b_0(t)]^\circ$ is admissible and regular. This latter fact is clear from Definition~\ref{admissible}.   \end{proof}

\section{Auxiliary results for the proof of Proposition~\ref{second-technical} } \label{appendix-second-technical}
\setcounter{equation}{0}

\bep\label{l=0-Ga-lemma} Let $\chi\in C_0^{\infty}(\real)_{\real}$. Let $ q\in C^{\infty}(\real)$ be s.t. $q'\in C_0^{\infty}(\real)$, $0 \leq q\leq 1$ and $q=1$ on a neighbourhood 
$\De$ of zero. Let $j_0$ be as specified in   Definition~\ref{j-definition-a} and s.t. $\supp\, j_0\subset \De$. 
Then
\begin{equation}\label{j-zero-Ga}
\chi\dGa(q^t,\pa_t q^t)\chi\Ga(j_0^t)=O(t^{-2})\quad \textup{and}\quad 
\chi[H(\xi), \Ga(\qd^t)]\chi \Ga(j_0^t)=O(t^{-2}), 
\end{equation}
where $\chi:=\chi(H(\xi))$.
\eep
\begin{proof} We set $q:=q^t$, $j_0:=j_0^t$ and write
\beq
\chi\dGa(q,\pa_t q)\chi\Ga(j_0)= \chi\dGa(q,\pa_t q)[\chi,\Ga(j_0)]=O(t^{-2}),
\eeq
due to the support property of $j_0$, Lemma~\ref{j_0Lemma} and the fact that $\chi\dGa(q,\pa_t q)=O(t^{-1})$.

Proceeding to the proof of the second part of (\ref{j-zero-Ga}), we write
\beqa
\chi [H(\xi), \Ga(\qd)]\chi\Ga(j_0)=\chi[H(\xi), \Ga(\qd)][\chi,\Ga(j_0)]+\chi [H(\xi), \Ga(\qd)]\Ga(j_0)\chi.
\eeqa
Here we used Corollary~\ref{Cor-DomInv-Gamma} to justify the formal computation.
The first term on the r.h.s. above is $O(t^{-2})$ by Lemma~\ref{j_0Lemma} and Proposition~\ref{comm-H-Ga}. 
As for the second term, we apply Proposition~\ref{comm-H-Ga} again:
\begin{align}\label{comm_H-Ga-rel-three}
& \chi[H(\xi),\Ga(q)]\Ga(j_0)\chi\\
\non &\qquad =\chi \bigl(-\nabla\Omeg(\xi-\dGa(k))\cdot \dGa(q,[k,q]^\circ)\Ga(j_0)+\dGa(q,[\om,q]^\circ)\Ga(j_0)\bigr)\chi+O(t^{-2}).
\end{align}
We note that $\dGa(q,[\om,q]^\circ)\Ga(j_0)\chi=O(t^{-2})$ and $\dGa(q,[k,q]^\circ)\Ga(j_0)\chi=O(t^{-2})$,  since $1-q$ is regular with the regularity region $\De$. 
See \eqref{regularity-property} and Lemma~\ref{concrete-admissible-obs}. This concludes the proof. \end{proof}

\bep\label{Ga-Heisenberg} Let $\chi\in C_0^{\infty}(\real)_{\real}$. Let $ q\in C^{\infty}(\real)$ be s.t. $q'\in C_0^{\infty}(\real)$, $0 \leq q\leq 1$ and $q=1$ on a neighbourhood  of zero. Let $j_0, j_{\infty}$ be as specified in   Definition~\ref{j-definition-a} and s.t. $j_0^2+j_{\infty}^2=1$. Then
\begin{align}
\chi \dGa(q^t,\pa_t q^t)\chi
&=\cGa(j^t)^*\chi^{\ex}\dGa^{\ex}(q^t, \pa_t q^t )\chi^{\ex}\cGa(j^t)+O(t^{-2}),   \label{q-derivative-term}\\
\chi[H(\xi), \Ga(q^t)]\chi&=\cGa(j^t)^*\chi^{\ex}[H^{\ex}(\xi),  \Ga^{\ex}(q^t)] \chi^{\ex}\cGa(j^t)+O(t^{-2}), \label{q-commutator}
\end{align}
where $\chi:=\chi( H(\xi))$ and $\chi^{\ex}:=\chi( H^{\ex}(\xi))$.
\eep
\begin{proof}   As for  (\ref{q-derivative-term}), it follows from  Proposition~\ref{derivative-term}, Lemma~\ref{concrete-admissible-obs} and~\eqref{regularity-property} applied with $b(t) = t \pa_t q^t$. 
Proceeding to the proof of formula~(\ref{q-commutator}),  we set $j:=j^t$, $q:=q^t$, and note that
\beq
\cGa(j)\chi[H(\xi), \Ga(q)]\chi=\chi^{\ex}\cGa(j)[H(\xi), \Ga(q)]\chi+O(t^{-2})
\eeq
by Lemma~\ref{Helffer} and  Proposition~\ref{comm-H-Ga}. 
Note that Corollary~\ref{Cor-DomInv-Gamma} ensures the validity of the computation above, as well as those to follow. Next, we will show that
\beqa
\chi^{\ex}\big( \cGa(j)[H(\xi), \Ga(q)] - [H^{\ex}(\xi),  \Ga^{\ex}(q)] \cGa(j) \big)\chi=O(t^{-2}).
\eeqa
By Proposition~\ref{comm-H-Ga}, it suffices to check that
\begin{align}
 &\chi^{\ex} \cGa(j)\nabla\Omeg(\xi-\dGa(k))\cdot \dGa(q,[k,q]^\circ)\chi\non\\ 
 &\qquad = \chi^{\ex}\nabla\Omeg(\xi-\dGa^{\ex}(k)) \cdot\dGa^{\ex}(q,[k,q]^\circ) \cGa(j) \chi+O(t^{-2}) \label{commutator-proof-one}
\end{align}
and
\beqa
\chi^{\ex} \cGa(j)\dGa(q,[\om,q]^\circ)\chi= \chi^{\ex} \dGa^{\ex}(q,[\om,q]^\circ) \cGa(j) \chi+O(t^{-2}). \label{commutator-proof-two} 
\eeqa
We show only (\ref{commutator-proof-one}), as the proof of (\ref{commutator-proof-two}) is analogous. Making use
of Lemma~\ref{f-commutator-ex}, and of the fact that $\dGa(q,[k,q]^\circ)\chi=O(t^{-1})$, we can write
\begin{align}
& \chi^{\ex} \cGa(j)\nabla\Omeg(\xi-\dGa(k))\cdot \dGa(q,[k,q]^\circ)\chi\non\\
& \qquad =\chi^{\ex}\nabla\Omeg(\xi-\dGa^{\ex}(k))\cdot \cGa(j) \dGa(q,[k,q]^\circ)\chi+O(t^{-2}).
\end{align}
Next, by exploiting the fact that $\chi^{\ex}\nabla\Omeg(\xi-\dGa^{\ex}(k))$ is bounded, and that Lemma~\ref{high-order-commutators} gives
\beqa
\cGa(j)\dGa(q,[k,q]^\circ)\chi=\dGa^{\ex}(q,[k,q]^\circ)\cGa(j)\chi+O(t^{-2}),
\eeqa
we conclude the proof of (\ref{commutator-proof-one}).

 We still have to show that
\beqa
\chi^{\ex}[H^{\ex}(\xi), \Ga^{\ex}(q)](\cGa(j)\chi-\chi^{\ex}\cGa(j))=O(t^{-2}).
\eeqa
This follows from Lemma~\ref{Helffer} and Proposition~\ref{comm-H-Ga}. \end{proof}

\bep\label{new-comm-prop} Let $q, \ti q\in C^{\infty}(\real)$ be s.t. $q', \ti q' \in C_0^{\infty}(\real)$, $0 \leq q\leq 1$, $q=1$ on a neighbourhood  of zero.
Let $\chi\in C_0^{\infty}(\real)_{\real}$ be supported in $(-\infty,\Sigma_0^{(2)}(\xi))$. Then
\beq
\chi^{(1)} [H^{(1)}(\xi),\Ga(q^t)\otimes 1](1\otimes \ti q^t)\chi^{(1)}=O(t^{-2}),
\eeq
where $\chi^{(1)}=\chi(H^{(1)}(\xi))$.
\eep
\begin{proof} Let us set in Proposition~\ref{comm-H-Ga} $j_0=q$, $j_{\infty}=1$ and conjugate equation~(\ref{comm_H-Ga-rel-one}) with the unitary $U$.
We obtain, as a form identity on $D(H_0^\ex(\xi)^3)$,
\begin{align}
& [H^{\ex}(\xi),\Ga^{\ex}(\un \jj)]\\
\non & \qquad = -\nabla\Omeg(\xi-\dGa^{\ex}(\un k))\cdot \dGa^{\ex}(\un j,[\un k,\un j]^\circ)+\dGa^{\ex}(\un j,[\un \om,\un j]^\circ)
+O(t^{-2})(H_0^\ex(\xi)+1)^3,\quad
\end{align}
where we set $j:=j^t$. Let us now apply the projection $P_{1}$ on $\Fock\otimes\Fock^{(1)}$ to both sides of this equality. 
We get, as a form identity on $D(H_0^{(1)}(\xi)^3)$,
\begin{align}
[H^{(1)}(\xi),\Ga(q)\otimes 1] & = -\nabla\Omeg(\xi-\dGa^{(1)}(k))\cdot (\dGa(q,[k,q]^\circ)\otimes 1)\non\\
&  \quad +(\dGa(q, [\om,q]^\circ)\otimes 1)+O(t^{-2})(H_0^{(1)}(\xi)+1)^3,  
\end{align}
where we abbreviated $q:=q^t$ and made use of relation~(\ref{dGa-q-p-def-double}).
Thus we can write
\begin{align}
 \chi^{(1)}[H^{(1)}(\xi),\Ga(q)\otimes 1](1\otimes \ti q)\chi^{(1)} &
=-\chi^{(1)}\nabla\Omeg(\xi-\dGa^{(1)}(k))\cdot (\dGa(q,[k,q]^\circ)\otimes \ti q)\chi^{(1)} \non\\
& \quad+\chi^{(1)}(\dGa(q, [\om,q]^\circ)\otimes \ti q)\chi^{(1)} +O(t^{-2}),    \label{reduced-commutator-equation}
\end{align}
where we set $\ti q:=\ti q^t$. Here we used Corollary~\ref{Cor-technical2}, with $n=3$.
Let us consider the first term on the r.h.s. above. We choose a function $\ti\chi\in C_0^{\infty}(\real)_{\real}$, supported in $(-\infty,\Sigma_0^{(2)}(\xi))$
and s.t. $\ti\chi\chi=\chi$. Then we get
\beqa
(1\otimes \ti q)\chi^{(1)}=(\ti\chi^{(1)})^2(1\otimes \ti q)\chi^{(1)}+O(t^{-1})
\eeqa
by Lemma~\ref{Helffer-second-component-lemma}. Next, we note that
\beqa
(\dGa(q,[k,q]^\circ)\otimes 1)(\ti\chi^{(1)})^2=\ti\chi^{(1)} (\dGa(q,[k,q]^\circ)\otimes 1)   \ti\chi^{(1)}+O(t^{-2}).
\eeqa 
Here we made use of Lemma~\ref{Helffer-dGa}, (after conjugating  expression~(\ref{Helffer-dGa-formula}) with $U$ and applying $P_{1}$ as above) and
of the fact that $t\to t[k,q]^\circ$ is an admissible and regular propagation observable. This is a consequence of the fact that $1-q$ is admissible
and regular by Lemma~\ref{concrete-admissible-obs}.

Thus making use of the fact that $\chi^{(1)}\nabla\Omeg(\xi-\dGa^{(1)}( k))$ is bounded and 
\beqa
\chi^{(1)}\nabla\Omeg(\xi-\dGa^{(1)}( k))\cdot (\dGa(q,[k,q]^\circ)\otimes 1)=O(t^{-1}),
\eeqa
we obtain
\begin{align}
& \chi^{(1)}\nabla\Omeg(\xi-\dGa^{(1)}( k))\cdot (\dGa(q,[k,q]^\circ)\otimes \ti q)\chi^{(1)}\non\\
& \quad =\chi^{(1)}\nabla\Omeg(\xi-\dGa^{(1)}( k))\cdot \ti\chi^{(1)}(\dGa(q,[k,q]^\circ)\otimes 1)\ti\chi^{(1)} (1\otimes \ti q)\chi^{(1)}+O(t^{-2}).
\end{align}
Exploiting again the fact that $t\to t[k,q]^\circ$ is admissible and regular, we obtain from Corollary~\ref{resolvents} that the first term on
the r.h.s. above is $O(t^{-2})$. The term involving $\dGa(q, [\om,q]^\circ)$ on the r.h.s. of (\ref{reduced-commutator-equation}) is treated
analogously. \end{proof}

\section{Auxiliary results for the proof of Propositions~\ref{main-propagation-estimate} and~\ref{minimal-velocity-estimate}}
\setcounter{equation}{0}

\bep\label{q-a-commutator} Let $\chi\in C_0^{\infty}(\real)_{\real}$ and  let $q\in C^{\infty}(\real)_{\real}$ be  s.t. $q'\in C_0^{\infty}(\real)$. Then
\beqa
\chi^{(1)}\i[H^{(1)}(\xi), 1\otimes q^t]\chi^{(1)}=\fr{1}{t}\chi^{(1)}C (1\otimes (q')^t)\chi^{(1)}+O(t^{-2}),
\eeqa
where $\chi^{(1)}:=\chi(H^{(1)}(\xi))$ and $C$ is a bounded operator on $\Fock\otimes\Fock^{(1)}$, which satisfies $[C,1\otimes p^t]=O(t^{-1})$ for any $p\in C^{\infty}(\real)_{\real}$ s.t. $p'\in C_0^{\infty}(\real)$.

If, in addition, $q'$ is positive and $\sqrt{q'}\in C_0^{\infty}(\real)$, then
\begin{align}
& \chi^{(1)}\i [H^{(1)}(\xi), 1\otimes q^t]\chi^{(1)}\non\\
& \qquad =\fr{1}{t}(1\otimes\sqrt{(q')^t})\chi^{(1)}\i[H^{(1)}(\xi), 1\otimes a]^\circ\chi^{(1)}(1\otimes\sqrt{(q')^t})+O(t^{-2}). 
\label{sandwitched-commutator-positive-q}
\end{align}
\eep
\begin{proof} We write $q:=q^t$ and set in Proposition~\ref{comm-H-dGa} $b_1=0$ and $b_2=q^t$. Clearly, $b_2$ is admissible.
By conjugating formula~(\ref{dGa-comm-double}) with the unitary $U$, we obtain 
\begin{align}
& [H^{\ex}(\xi),\dGa(\un b)]\non\\
&\qquad = -\nabla\Omeg(\xi-\dGa^{\ex}(k))\cdot \dGa^{\ex}([\un k,\un b]^\circ)+\dGa^{\ex}([\un \om,\un b]^\circ)+O(t^{-2})(H_0^\ex(\xi)+1)^4,
\label{dGa-comm-double-one-one}
\end{align}
as a form identity on $D(H_0^\ex(\xi)^4)$.

Now we apply to both sides of this equality the projection $P_{1}$ on the subspace $\Fock\otimes\Fock^{(1)}\subset\Fock^{\ex}$ and 
multiply by the operators $\chi^{(1)}$. We get
\begin{align}
\non & \chi^{(1)}\i [H^{(1)}(\xi),1\otimes q]\chi^{(1)}\\
&\qquad =\chi^{(1)}\bigl(-\nabla\Omeg(\xi-\dGa^{(1)}(k))\cdot (1\otimes \i[k, q]^\circ)+1\otimes \i[\om,q]^\circ\bigr)\chi^{(1)}+O(t^{-2}). \label{expansion-one-otimes-q}
\end{align}
Here we used  Corollary~\ref{Cor-technical2}, with $n=4$.
Now we obtain from Lemma~\ref{pseudo-lemma}, that
\begin{equation}
\i[k, q]^\circ=\fr{1}{t}\i[k,a]^\circ q'+O(t^{-2}) \quad \textup{and} \quad
\i[\om,q]^\circ=\fr{1}{t}\i[\om,a]^\circ q'+O(t^{-2}),
\end{equation}
where $\i[k,a]^\circ =v$ and $\i[\om,a]^\circ=\nabla\om\cdot v$.
Thus we get from relation~(\ref{expansion-one-otimes-q}) that
\begin{align}
& \chi^{(1)}\i[H^{(1)}(\xi),1\otimes q]\chi^{(1)}\non\\\label{AppI-Symm0}
& \qquad =\fr{1}{t}\chi^{(1)}\bigl(-\nabla\Omeg(\xi-\dGa^{(1)}(k))\cdot (1\otimes v)+1\otimes \nabla\om\cdot v\bigr)(1\otimes q')\chi^{(1)}+O(t^{-2}).
\end{align}
We choose $\ti\chi\in C_0^{\infty}(\real)_{\real}$ s.t. $\chi\ti\chi=\chi$ and set 
\beq
C:=\ti\chi^{(1)}\bigl(-\nabla\Omeg(\xi-\dGa^{(1)}(k))\cdot (1\otimes v)+1\otimes \nabla\om\cdot v\bigr).
\eeq
It is clear that $C$ is bounded. The property $[C,1\otimes p^t]=O(t^{-1})$ follows from Lemmas~\ref{pseudo-lemma}, 
\ref{comm-H-dGa-lemma} and \ref{Helffer-second-component-lemma}.
This concludes the proof of the first part of the proposition.

Proceeding to the proof of (\ref{sandwitched-commutator-positive-q}),  we note that by Lemma~\ref{pseudo-lemma}
\begin{equation}\label{AppI-Symm1}
[v, \sqrt{q'}]=O(t^{-1}) \qquad \textup{and}\qquad 
[\nabla\om\cdot v, \sqrt{q'}]=O(t^{-1}).
\end{equation}
There also holds by Lemma~\ref{comm-H-dGa-lemma} (after conjugating it with $U$ and applying the projection $P_{1}$)
\beqa\label{AppI-Symm2}
\chi^{(1)}[\nabla\Omeg(\xi-\dGa^{(1)}(k)), 1\otimes \sqrt{q'}]=O(t^{-1}). 
\eeqa
On the other hand, Lemma~\ref{Helffer-second-component-lemma} gives $[1\otimes \sqrt{q'},\chi^{(1)}]=O(t^{-1})$.
Observing that 
\begin{equation}
\i[H^{(1)}(\xi),1\otimes a]^\circ = -\nabla\Omega(\xi-\dGa^{(1)}(k))\cdot (1\otimes v) 
+ 1\otimes \nabla\om\cdot v,
\end{equation}
we conclude \eqref{sandwitched-commutator-positive-q} by symmetrizing \eqref{AppI-Symm0}, with the aid of \eqref{AppI-Symm1} and \eqref{AppI-Symm2}.
 \end{proof} 

\bel \label{minimal-lemma} Let $a_i:=\h (v_i(k)\cdot\i\nabla_k+\i\nabla_k \cdot v_i(k))$ for some 
$v_{i}\in C_0^{\infty}(\real^{\nu}\backslash\{0\};\real^{\nu})$, for $i\in \{1,2\}$.
Let $\un a=\diag(a_1,a_2)$ be an operator on (a domain in) $\mfh\oplus\mfh$. Then $H_1(\xi)$ is of class $C^1(\dGa(\un a))$ and
\begin{equation}
[H_1(\xi), \dGa(\un a)]^\circ \in B(D(N H_{1,0}(\xi));\hil) \subset B(D(H_{1,0}(\xi)^2);\hil)
\end{equation}
In particular $\chi(H_1(\xi))\in C^1(\dGa(\un a))$, for any $\chi\in C_0^\infty(\real)_{\real}$.
\eel
\begin{proof} From \cite[Prop.~2.8]{MR12}, and a conjugation by the unitary $U$, we learn that $H_1(\xi)$ is of class $C^1(\dGa(\un a))$. We furthermore find that
\beqa
\i[H_1(\xi), \dGa(\un a)]^\circ =\dGa(\un v\cdot \un{\nabla\om})-\dGa(\un v)\cdot \nabla\Omeg(\xi-\dGa(\un k))-\phi(\i a_1 G,0),
\eeqa
where $\un v:=\diag(v_{1},v_{2})$ is a $\nu$-tuple of operators  on $\mfh\oplus\mfh$ and  the expression on the r.h.s. above is manifestly $N H_{1,0}(\xi)$-bounded. 
The remaining part of the lemma follows analogously as in the proof of Lemma~\ref{Helffer-dGa}. \end{proof}

\bel\label{minimal-auxiliary}  Let $\chi\in C_0^{\infty}(\real)_{\real}$, $\ti q\in C_0^{\infty}(\real)$, $q\in C^{\infty}(\real)$ be  s.t. $q'\in C^{\infty}_0(\real)$, 
$0\leq q \leq 1$ and $q= 1$ in some neighbourhood  of zero.  Let $t\to b(t)=(a/t)\ti q(a/t)$.  Let $\un q^t=(q^t, q^t)$ and $\un b(t)=(b(t),b(t))$
be propagation observables on $\mfh\oplus\mfh$. Then
\beq
 [ \dGa(\un q^t,\un b), \chi(H_1(\xi))] =O(t^{-1}).
\eeq 
\eel
\begin{proof} We set $\un q:=\un q^t$ and $R_{1,0} = (1+H_{0,1}(\xi))^{-1}$. We note that $b$ is admissible by Lemma~\ref{pseudo-lemma}.  Let us first estimate
the commutator of $\dGa(\un q^t,\un b)$ with $H_1(\xi)$. As for the first term from the free auxiliary Hamiltonian, cf.~\eqref{AuxHam}, 
Lemma~\ref{comm-H-dGa-lemma} gives
\beqa
[\Omeg(\xi-\dGa(\un k)), \dGa(\un q, \un b)]R_{1,0}^4=O(t^{-1}).
\label{dGaqp-H-repeated}
\eeqa
Concerning the second term from the Hamiltonian, we obtain from Lemma~\ref{rest-terms} 
\beqa
[\dGa(\un\om),\dGa(\un q, \un b)] R_{1,0}^4 
=\bigl(\dGa(\un q,[\un \om, \un q]^\circ, \un b)+\dGa(\un q, [\un \om, \un b]^\circ)\bigr)R_{1,0}^4 =O(t^{-1}).
\eeqa
 The  interaction term from the Hamiltonian gives 
\begin{align}
& [\phi(G,0),\dGa(\un q, \un b)]R_{1,0}^{4}=\bigl(a^*((1-q)G,0)\dGa(\un q, \un b)-a^*(b G,0)\Ga(\un q)\non\\
& \qquad +\Ga(\un q)a(b^*G,0)+\dGa( \un q, \un b)a((q-1)G,0)\bigr) R_{1,0}^{4}=O(t^{-1}),
\end{align}
where we made use of Lemma~\ref{q-p-lemma}. In the last step we exploited the fact that $\|(1-q)G\|_2\leq C/t^2$, since $1-q$ is regular,
and the bound $\|bG\|_2=\fr{1}{t}\|qaG\|_2\leq c/t$, which follows from the fact that $G$ is in the domain of $a$. Thus we have shown that
\beqa
[H_1(\xi), \dGa(\un q, \un b)]R_{1,0}^4=O(t^{-1}).
\eeqa 
Now one concludes the proof using the method of almost analytic extensions as in the proof of Lemma~\ref{Helffer-dGa}. \end{proof}

\section{Auxiliary results for the proof of Theorem~\ref{wave-operators-theorem} }
\setcounter{equation}{0}

In the present appendix we ask the reader to keep
Corollary~\ref{Cor-DomInv-Gamma} in mind. It ensures that the statements of 
results and manipulations in proofs are meaningful.
\bel\label{commutator-exchange} Let $\chi\in C_0^{\infty}(\real)_{\real}$, $j_0, j_{\infty}$ be as specified in Definition~\ref{j-definition-a} and s.t.
$j_0^2+j_{\infty}^2=1$, and let $q=(q_0,q_{\infty}):=(j_0^2,j_{\infty}^2)$.
Then
\beqa
\chi^{\ex}\big(H^{\ex}(\xi)\cGa(q^t)-\cGa(q^t)H(\xi)\big)\chi
=2\chi^{\ex}[H^{\ex}(\xi),\Ga^{\ex}(\un j^t)]\cGa(j^t)\chi+O(t^{-2}), \label{wave-op-commutator}
\eeqa
where $\un j^t:=\diag(j_0^t,j_{\infty}^t)$ is a propagation observable on $\mfh\oplus\mfh$  
and we set $\chi:=\chi(H(\xi))$ and $\chi^{\ex}:=\chi(H^{\ex}(\xi))$.
\eel
\begin{proof} We set $q:=q^t$, $j:=j^t$.  We note that, by Proposition~\ref{H-cGa-comm}, 
\begin{align}
&\chi^{\ex}\bigl(H^{\ex}(\xi)\cGa(q)-\cGa(q)H(\xi)\bigr)\chi\non\\
&\qquad =\chi^{\ex}\bigl(-\nabla\Omeg(\xi-\dGa^{\ex}(k))\cdot \cdGa(q,[\un k, q]^\circ)+\cdGa(q, [\un \om, q]^\circ) \bigr)\chi+O(t^{-2}),
\label{Ga-check-comm-one-new}
\end{align}
where $[\un \om, q]^\circ$ is the extension by continuity of the form 
$\un\om q-q\om$,  a priori defined on $(D(\om)\oplus D(\om))\times D(\om)$.
The same remark goes for $[\un k, q]^\circ$.
On the other hand, Proposition~\ref{comm-H-Ga} gives
\begin{align}
&\chi^{\ex}[H^{\ex}(\xi),\Ga^{\ex}(\un j)]\cGa(j)\chi \non\\
&\qquad =\chi^{\ex}\big(-\nabla\Omeg(\xi-\dGa^{\ex}( k)) \cdot \dGa^{\ex}(\un j,[\un k,\un j]^\circ)+\dGa^{\ex}(\un j,[\un \om,\un j]^\circ)\big)\cGa(j)\chi  +O(t^{-2}),
\label{comm_H-Ga-rel-one-new}
\end{align}
where we made use of Lemma~\ref{Helffer} to observe $\cGa(j)\chi=\chi^{\ex}\cGa(j)+O(t^{-1})$.

In view of (\ref{Ga-check-comm-one-new}) and (\ref{comm_H-Ga-rel-one-new}), to complete the proof of the lemma,
it suffices to show that
\begin{align}
& \chi^{\ex}\nabla\Omeg(\xi-\dGa^{\ex}(k))\cdot \cdGa(q,[\un k, q]^\circ)\chi\non\\
& \qquad =2\chi^{\ex}\nabla\Omeg(\xi-\dGa^{\ex}( k)) \cdot \dGa^{\ex}(\un j,[\un k,\un j]^\circ)\cGa(j)\chi+O(t^{-2}) 
\label{first-relation-wave-op}
\end{align}
and
\beqa
\chi^{\ex}\cdGa(q, [\un \om, q]^\circ) \chi=2\chi^{\ex}\dGa^{\ex}(\un j,[\un \om,\un j]^\circ)\cGa(j)\chi+O(t^{-2}).
\label{second-relation-wave-op}
\eeqa
Both relations are a consequence of the following fact: Let $g\in C^{\infty}(\real)$ be  s.t. all its derivatives of
non-zero order are bounded. Let $\un g:=\diag(g,g)$ be the corresponding operator on $\mfh\oplus\mfh$.
Then 
\begin{align}
2\dGa^{\ex}(\un j,[\un g,\un j]^\circ)\cGa(j)\chi=\cdGa(\un j j, 2[\un g,\un j]^\circ j)\chi
&=\cdGa(q,[\un g, q]^\circ)\chi-\cdGa(q,  [[\un g,\un j]^\circ, j])\chi\non\\
&=\cdGa(q,[\un g, q]^\circ)\chi+O(t^{-2}),
\end{align}
where in the last step we made use of Lemma~\ref{pseudo-lemma}. This concludes the proof. \end{proof}
\bel\label{Ga-H-Ga} Let $\chi\in C_0^{\infty}(\real)_{\real}$ and  let $j_0$ be as specified in Definition~\ref{j-definition-a}.
Then there holds 
\beq
\chi[\Ga(j_0^t),[H(\xi),\Ga(j_0^t)]]\chi=O(t^{-2}),
\eeq
where we set $\chi:=\chi(H(\xi))$.
\eel
\begin{proof}  We set  $j_0:=j_0^t$ and recall that, by Proposition~\ref{comm-H-Ga},
\begin{align}\label{comm_H-Ga-rel-two-new}
& [H(\xi),\Ga(\jj_0)]\\
&\non\qquad =\bigl(-\nabla\Omeg(\xi-\dGa(k))\cdot \dGa(j_0,[k,j_0]^\circ)+\dGa(j_0,[\om,j_0]^\circ)\bigr)+O(t^{-2})(H_{1,0}(\xi)+1)^3,
\end{align}
in the sense of forms on $D(H_{1,0}(\xi)^3)$. In view of this relation and the fact that $\Ga(j_0)\chi=\chi\Ga(j_0)+O(t^{-1})$
(Lemma~\ref{j_0Lemma}), it is enough to check that
\beqa
\chi [\Ga(j_0),\nabla\Omeg(\xi-\dGa(k))\cdot \dGa(j_0,[k,j_0]^\circ)]\chi=O(t^{-2}) \label{double-commutator-first}
\eeqa
and
\beqa
\chi [\Ga(j_0),  \dGa(j_0,[\om,j_0]^\circ)]\chi=O(t^{-2}). \label{double-commutator-second}
\eeqa
Relation~\eqref{double-commutator-second} follows immediately from formula~(\ref{Gamma-dGamma-comm}) which gives
\beqa
\chi [\Ga(j_0),  \dGa(j_0,[\om,j_0]^\circ)]\chi=\chi\dGa(j_0^2,[j_0,[\om,j_0]^\circ])\chi=O(t^{-2}), \label{simple-part-double-commutator}
\eeqa
where in the last step we applied Lemma~\ref{pseudo-lemma}. Formula~(\ref{double-commutator-first}) follows from
\beqa
[\Ga(j_0),  \dGa(j_0,[k,j_0]^\circ)]\chi=O(t^{-2}),
\eeqa
which is justified as (\ref{simple-part-double-commutator}), from the fact that $\dGa(j_0,[k,j_0]^\circ)\chi=O(t^{-1})$ and from Lemma~\ref{f-commutator}, which gives
\beqa
\chi [\Ga(j_0),\nabla\Omeg(\xi-\dGa(k))]=O(t^{-1}).
\eeqa
This concludes the proof. \end{proof}

\section{Negative spectrum of the conjugate operator}
\setcounter{equation}{0}
\bel\label{negative-spectrum-of-a} Let $\chi\in C_0^{\infty}(\real)_{\real}$ and $\Psi\in D(\dGa(a))$. There exists a constant $c>0$ such that the following holds true:  
For any pair of functions $q, q_R\in C^{\infty}(\real)$, with  $0\leq q,q_R\leq 1$,  $\supp\, q \subset (-\infty,\eps)$
for some $\eps>0$, $q(s)=q_R(s)$ for $s>-R$ and $q_R(s)=0$ for $s<-R-1$; we have
\beqa
\sup_{t\geq 1}\|(\Ga(q^t)-\Ga(q_R^t))\e^{-\i tH(\xi)} \chi(H(\xi)) \Psi\|\leq \frac{c}{R}.
\eeqa
\eel
\begin{proof} Let us denote by $\mathbf{1}_{\{A\leq -R\}}$ the spectral projection of a self-adjoint operator $A$ on the interval $(-\infty,-R]$.
We set $q:=q^t$ and $q_R:=q_R^t$ and recall that $\Psi\in D(\dGa(a))$. As a consequence $\chi(H(\xi))\Psi\in D(\dGa(a))$ for any  
$\chi\in C_0^{\infty}(\real)_{\real}$, since $H(\xi)$ is of class $C^{1}(\dGa(a))$~\cite[Proposition~2.5]{MR12}.

Making use of the subsequent Lemma~\ref{Large-R}, and abbreviating $\chi = \chi(H(\xi))$, we obtain
\begin{align}
& \bigl\| \bigl(\Ga(q(a/t))- \Ga(q_{R}(a/t)) \bigr)\e^{-\i tH(\xi)}\chi\Psi\bigr\|\non\\
& \qquad =\bigl\| \mathbf{1}_{\{\dGa(a/t)\leq -R+\eps N\}}\bigl(\Ga(q)- \Ga(q_{R}) \bigr)\e^{-\i tH(\xi)}\chi\Psi\bigr\| \non\\
& \qquad =\bigl\| \mathbf{1}_{\{\dGa(a/t)\leq -R+\eps N\}}\mathbf{1}_{\{ N\leq R \}}\bigl(\Ga(q)- \Ga(q_R) \bigr)\e^{-\i tH(\xi)}\chi\Psi\bigr\|\non\\
& \quad\qquad +\bigl\| \mathbf{1}_{\{\dGa(a/t)\leq -R+\eps N\}}\mathbf{1}_{\{ N\geq R \}}\bigl(\Ga(q)- \Ga(q_{R}) \bigr)\e^{-\i tH(\xi)}\chi\Psi\bigr\|.\label{Control-of-large-R}
\end{align}
The first term on the r.h.s. above can be estimated by
\begin{align}
& \bigl\| \mathbf{1}_{\{\dGa(a/t)\leq -R(1-\eps)\}}  \bigl(\Ga(q)- \Ga(q_{R}) \bigr)\e^{-\i tH(\xi)}\chi\Psi\bigr\|\non\\
& \qquad =\bigl\| \mathbf{1}_{\{\dGa(a/t)\leq -R(1-\eps)\}} \dGa(a/t)^{-1} \bigl(\Ga(q)- \Ga(q_{R}) \bigr)\dGa(a/t)\e^{-\i tH(\xi)}\chi\Psi\bigr\|\non\\
& \qquad \leq \fr{2}{(1-\eps)Rt}\bigl\|\e^{\i tH(\xi)}\dGa(a)\e^{-\i tH(\xi)}\chi\Psi\bigr\|. \label{expression-to-be-integrated}
\end{align}
To estimate the expression on the r.h.s. of (\ref{expression-to-be-integrated}), we proceed similarly as in the proof of \cite[Lemma~44]{FGSch3}:
\begin{align}
\bigl\|\e^{\i tH(\xi)}\dGa(a)\e^{-\i tH(\xi)}\chi\Psi\bigr\| & \leq  \bigl\| \int_0^tdt' \e^{\i t'H(\xi)}[H(\xi), \dGa(a)]^\circ\e^{-\i t'H(\xi)}\chi\Psi\bigr\|
+\bigl\|\dGa(a)\chi\Psi\bigr\|\non\\
& \leq tc'\bigl\|\Psi\|+\|\dGa(a)\chi\Psi\bigr\|,
\end{align}
where we made use of the fact that $c':=\|[H(\xi), \dGa(a)]^\circ\chi\|<\infty$, by Lemma~\ref{minimal-lemma}. Thus we obtain that
\beqa
\| \mathbf{1}_{\{\dGa(a/t)\leq -R(1-\eps)\}}  \bigl(\Ga(q)- \Ga(q_{R}) \bigr)\e^{-\i tH(\xi)}\chi\Psi\|\leq \fr{c''}{R}\|\Psi\|+\fr{c''}{Rt}\|\dGa(a)\chi\Psi\|,
\eeqa
where the constant $c''$ does not depend on the choice of $t$, $q$ and $q_R$.
 As for the second term on the r.h.s. of (\ref{Control-of-large-R}), it is bounded by
\beqa
2\|\mathbf{1}_{\{ N\geq R \}}\chi\Psi\| \leq 2\|\mathbf{1}_{\{ N\geq R \}} (1+N)^{-1}\| \,
\|(1+N)\chi\Psi\|\leq \fr{c'''}{R}\|\Psi\|,
\eeqa
where $c''' :=  2\|(N+1)\chi(H(\xi))\|<\infty$. 
Altogether, we get that
\beqa
\| \bigl(\Ga(q(a/t))- \Ga(q_{R}(a/t)) \bigr)\e^{-\i tH(\xi)}\chi\Psi\|\leq \fr{c}{R}\|\Psi\|+\fr{c}{Rt}\|\dGa(a)\chi\Psi\|,
\eeqa
where $c$ is independent of $t$, $q$ and $q_R$. This concludes the proof. \end{proof}
\bel\label{Large-R} Let $q, q_R\in C^{\infty}(\real)$ be s.t.  $0\leq q,q_R\leq 1$,  $\supp\, q \subset (-\infty,\eps)$
for some $\eps>0$, $q(s)=q_R(s)$ for $s>-R$ and $q_R(s)=0$ for $s<-R-1$. Then, for $\Psi\in\hil$,
\begin{equation}
\bigl(\Ga(q(a/t))- \Ga(q_{R}(a/t) )\bigr)\Psi
 =\mathbf{1}_{\{\dGa(a/t)\leq -R+\eps N\}}\bigl(\Ga(q(a/t))-\Ga(q_{R}(a/t))\bigr)\Psi.\label{Large-R-eq}
\end{equation}
\eel
\begin{proof} As all the operators involved commute with the number operator, it is enough to consider the problem in some $n$-particle subspace.
We embed $\Fock^{(n)}$ into the non-symmetrized tensor product of single-particle spaces $\mfh^{\otimes n}$.
We note that $a\otimes 1\otimes \cdots\otimes 1, 1\otimes a \otimes 1\otimes\cdots \otimes 1, \ldots, 1\otimes\cdots\otimes 1\otimes a$ 
is a family of $n$ commuting operators on $\mfh^{\otimes n}$. We denote their joint spectral projection-valued measure by $F$. Thus the $n$-particle component of the vector on the l.h.s. of (\ref{Large-R-eq}) is a sum of terms of the form
\beqa
\int q(a_1/t)\ldots q(a_{i-1}/t)\bigl(q(a_i/t)-q_R(a_i/t)\bigr)q_R(a_{i+1}/t)\ldots q_R(a_n/t) dF(a)\Psi_n,
\eeqa
where $\Psi_n$ is the $n$-particle component of $\Psi$. Now, by the assumed properties of $q$ and $q_R$, we obtain that the above
expression is equal to 
\begin{align}
&\int \mathbf{1}\bigl( (a_1/t+\cdots+a_n/t)\leq -R+\eps n \bigr)\non\\
 &\quad \times q(a_1/t)\ldots q(a_{i-1}/t)\bigl(q(a_i/t)-q_R(a_i/t)\bigr)q_R(a_{i+1}/t)\ldots q_R(a_n/t) dF(a)\Psi_n.
\end{align}
This proves (\ref{Large-R-eq}) on $\mfh^{\otimes n}$. Since both sides of (\ref{Large-R-eq}) leave $\Fock^{(n)}$ invariant, this completes the proof. \end{proof} 

\section{Structure of the isolated spectrum}
\setcounter{equation}{0}

We begin by recalling some analytic perturbation theory for isolated eigenvalues following Kato.  Suppose that $D\subset \complex$ is an open set which intersects with the real line
and $D\ni \kappa\to T(\kappa)$ is a holomorphic family of Type A in the sense of Kato.
We assume that $T(\kappa)$ is a self-adjoint operator when $\kappa\in D\cap\real$.
Suppose $\lambda_0\in\real$ is an isolated eigenvalue of the self-adjoint operator $T(\kappa_0)$, with $\kappa_0\in\real\cap D$. Denote by $n_0$ its multiplicity, which we assume to be finite. Let $e>0$ be such that 
$\sigma(T(\kappa_0))\cap J_{2e} = \{\lambda_0\}$, where
$J_e :=  [\lambda_0-e,\lambda_0+e]$. 

Abbreviate $\sigma_e(\kappa):= \sigma(T(\kappa))\cap B_e(\lambda_0)$.
There exists $r>0$ such that $B_r(\kappa_0)\subset D$ and for all $\kappa\in B_r(\kappa_0)$ we have $\sigma_{2e}(\kappa) = \sigma_e(\kappa)$.
Such an $r$ exists because the set $\{(\kappa,\lambda)| \lambda\in\sigma(T(\kappa))\}$ is a closed subset of $D\times \complex$.

Denote by $C$ the circle in $\complex$ encircling $\lambda_0$ with radius $3e/2$.
Then $\sigma_e(\kappa)$ is enclosed by the circle for all $\kappa\in B_r(\kappa_0)$, and accounts for all the spectrum of $T(\kappa)$ inside (or on) the circle.
We can thus compute the Riesz projection:
\[
P(\kappa) =  \fr{1}{2\pi i}\oint_C dz\fr{1}{z-T(\kappa)}.
\]
For real $\kappa$ the bounded operator $P(\kappa)$ is the spectral projection
onto the spectral subspace  pertaining to the spectrum of $T(\kappa)$ inside the cluster $\sigma_e(\kappa)$.
In particular $P(\kappa_0) = P_{\lambda_0}(\kappa_0)$, the orthogonal projection onto the  $n_0$-dimensional eigenspace of $T(\kappa_0)$, pertaining to 
the eigenvalue $\lambda_0$. Due to norm-continuity
of $\kappa\to P(\kappa)$ we conclude that the set
$\sigma_e(\kappa)$ has cardinality at most $n_0$, corresponding to eigenvalues with (algebraic) multiplicities summing up to $n_0$. 

 Denote by $v_1^0,\dots,v_{n_0}^0$ an ONB for the range of $P_{\lambda_0}(\kappa_0)$. Then, possibly choosing $r$ smaller, we may assume
 that $v_j(\kappa) = P(\kappa)v_j^0$ forms a linearly independent analytic set of vectors spanning $\Ran(P(\kappa))$. Using the Gram-Schmidt procedure we can pass to an analytic ONB $u_1(\kappa),\dotsc,u_{n_0}(\kappa)$ for $\Ran(P(\kappa))$. 
Such a basis defines an analytic family of unitary maps $\Pi_\kappa\colon \Ran(P(\kappa))\to \complex^{n_0}$, defining $\Pi_\kappa(u_j(\kappa)) = e_j$, the $j$'th standard basis vector.  We can now construct an analytic family
of $n_0\times n_0$ matrices $A(\kappa) = \Pi_\kappa T(\kappa) \Pi_\kappa^*$. 
By construction $A(\kappa)$ is self-adjoint for $\kappa\in B_r(\kappa_0)\cap\real$
and $\sigma(A(\kappa))= \sigma_e(\kappa)$. 

By a result of Kato \cite[Theorem~6.1]{Ka}, we can identify a number $m_0\leq n_0$ of real analytic functions $\mu_j\colon B_r(\kappa_0)\cap \real\to \real$,
such that $\sigma_e(\kappa)=\{\mu_1(\kappa),\dotsc,\mu_{m_0}(\kappa)\}$.
They all coincide with $\lambda_0$ if $\kappa=\kappa_0$ and are otherwise distinct. 

The above discussion implies the following result 
on analytic continuation of shells through crossings.
\bep\label{analytic-continuation} Let $\mcXo$ be a  level crossing, which is a sphere of radius $R>0$. 
Let $(\mcA_m^-,S_m^-)$, $m\in J^-$ and $(\mcA_n^+,S_n^+)$, $n\in J^+$ be shells approaching
this crossing from the inside and outside, respectively. 
Then, (after suitable identification of the  index sets $J_\pm=:J$) one can find analytic functions
\beqa
\mcA_n^+\cup \mcXo\cup   \mcA_n^-\ni \xi\to S_n(\xi),
\eeqa
such that $S_n(\xi)=S^\pm_n(\xi)$, $\xi\in \mcA_n^{\pm}$.
\eep
\begin{proof} \rm  Put $T(\kappa) = H(\kappa,0,\dotsc,0)$, where we exploit the rotation invariance of the spectrum to conclude the proposition from the preceding discussion.
\end{proof}

Let $\mcA(r;R):=\set{\xi\in\real^{\nu}}{ r<|\xi|<R }$ for some $0\leq r<R\leq \infty$.  
Keeping in mind the possibility that the inner or outer  boundary of a shell is a subset of the essential spectrum,
we obtain from the above proposition that $\Sigma_\iso\backslash (\{0\}\times\real)$ is a union of graphs of an at most countable family of rotation invariant analytic functions
\beqa
\mcA(r_n; R_n)\ni \xi\to S_n(\xi),
\eeqa
where $n\in J$. (The zero total momentum fiber has been cut out since 
one may in principle have shells like graphs of the two functions $\xi\to(\xi\pm\xi_0)^2$ crossing analytically at $\xi = 0$ but not naturally occurring as a single-valued rotation invariant function.) These considerations enable a splitting of the isolated bound states $\hil_\iso = E(\Sigma_\iso)\hil$ into 
dressed electron subspaces:
\begin{align}\label{QuasiParticleSpace}
\hil_\iso & = \bigoplus_{n} \hil_{\iso,n}, \qquad \textup{where}\quad \hil_{\iso,n} = \ov{\widetilde{\hil}_{\iso,n}},\\
\label{PreQuasiParticleSpace}
\widetilde{\hil}_{\iso,n} &=\Bigset{ I_{\LLP}^* \int^\oplus d\xi\, \Psi_\xi }{  \Psi\in C^0_0(\mcA(r_n;R_n);\Fock),\ H(\xi)\Psi_\xi = S_n(\xi)\Psi_\xi},
\end{align}
where by $\Psi\in C^0_0(\mcA(r_n;R_n);\Fock)$ it is understood that $\xi \to \Psi_\xi\in\Fock$ is a continuous function,
 compactly supported in $\mcA(r_n;R_n)$.

After this preparation we state and prove the following  corollary of Proposition~\ref{analytic-continuation}:
\bec\label{Non-constant-corollary} Let $\om$ be the boson dispersion relation. Then
\beq
S^{(1)}_n(\xi; k):=S_n(\xi-k)+\om(k), \label{constant-function-one}
\eeq 
defined for $k\in \xi-\mcA(r_n;R_n)$, is a constant function at most for $\xi$ from some countable set. 
\eec
\begin{proof} Let us first assume that $\om$ is not a constant function.
Suppose that 
\beq
\xi-\mcA(r_n;R_n)\ni k\to S^{(1)}_n(\xi; k) 
\eeq
is constant  for $\xi=\xi_0$ and $\xi=\xi_0+k'$ for 
some  $k'\neq 0$. (For $\nu>1$ it is enough to assume that there is one such $k'$ to arrive at a contradiction.
For $\nu=1$ we assume that there are uncountably many).
Then
\begin{equation}
\begin{aligned}
&k\in \xi_0-\mcA(r_n;R_n): & & S_n(\xi_0-k)+\om(k)  =c_{\xi_0},\\
&k\in \xi_0+k'-\mcA(r_n;R_n):& &  S_n(\xi_0-k+k')+\om(k)=c_{\xi_0+k'}.
\end{aligned}
\end{equation}
But the latter condition means that $k-k'\in \xi_0-\mcA(r_n;R_n)$, so we can substitute it into the first equality,
obtaining the equations

\begin{equation}
k\in \xi_0+k'-\mcA(r_n;R_n): \qquad \begin{aligned}
& S_n(\xi_0-k+k')+\om(k-k') =c_{\xi_0},\\
& S_n(\xi_0-k+k')+\om(k)=c_{\xi_0+k'}.
\end{aligned}
\end{equation}
Consequently,
\beq
\om(k)-\om(k-k')=c_{\xi_0+k'}-c_{\xi_0}. \label{periodicity-of-omega}
\eeq
Since this equality holds on an open set, it extends to all $k\in\real^{\nu}$ by analyticity. 
Now let us assume that $\nu>1$. Then, making use
of rotation invariance of $\om$, we obtain for any $O\in\mrO(\nu)$ 
\beqa
\om(k)-\om(k-Ok')=c_{\xi_0+k'}-c_{\xi_0}.
\eeqa
By differentiating this relation w.r.t. one-parameter families of rotations, we obtain $\nabla\om(k)\cdot Lk'=0$,
for any element $L$ of the Lie algebra of the group of rotations. Recalling that
such $L$ are antisymmetric matrices and choosing coordinates so that $k'=(c,0,\ldots,0)$, 
we obtain that $\pa_i\om(k)=0$ for all $2\leq i\leq \nu$. By rotation invariance, this
is only possible if $\om$ is  constant, which is a contradiction.

Let us now go back to formula~(\ref{periodicity-of-omega}) and assume that $\nu=1$.
By differentiating this relation  w.r.t. $k$, we obtain that 
\beqa
\nabla\om(k)=\nabla\om(k-k')
\eeqa
i.e. $\nabla\om$ is a continuous function  which has uncountably many periods $k'$. But this
is only possible if $\nabla\om$ is a constant function \cite{CK66}. This implies that $\om(k)=c_1k+c_2$.
We note that $c_1=0$ by reflection invariance. Thus we 
obtain again that $\om$ is a constant function contradicting our assumption.

Finally, let us suppose that $\om$ is a constant function.  Then $S^{(1)}_n(\xi; k)=S_n(\xi-k)+\om(k)$
can only be  constant  if $S_n$ is  constant. 
But this is excluded by the following property
\beqa
\lim_{|\xi|\to\infty}( \Sigma_0^{(1)}(\xi)-\Sigma_0(\xi))=0,
\eeqa
proven in \cite[Theorem~2.4]{MRMP}, and the fact that for a constant dispersion relation
\beq
\Sigma_{0}^{(1)}(\xi)=\inf_{k\in\real^{\nu}}(\Sigma_0(\xi-k)+\om(k))=\inf_{k\in\real^{\nu}}\Sigma_0(k)+m=\mathrm{const.}
\eeq
In the above reasoning  we made use of Proposition~\ref{analytic-continuation} to show that
any shell $(\mcA, S)$ s.t. $S$ is constant extends to  a constant shell $S_n$ on $\mcA(0,\infty)$. \end{proof}

\section{Structure of the spectrum of the extended Hamiltonian}
\setcounter{equation}{0}

For a Borel set $O\subset \real\times \real^\nu$ we recall the notion of $O$-compatibility
from Subsection~\ref{results-subsection}.  A state $\Psi\in\hil_\bnd$ and a boson wave packet $h\in\mfh$ are called
$O$-compatible if  there exists a  Borel subset $S\subset \real^{\nu+1}$ such that:
$\Psi\in E(S)\hil$ and for any $k$ in the essential support of $h$ and $(\xi,\mu)\in S$,
we have $(\xi+k,\mu+\omega(k))\in O$. As shown in Lemma~\ref{energetic-considerations} below, this property ensures that
the simple tensor $\Psi\otimes a^*(h)\vac$ is an element of $E^{(1)}(O)(\hil_\bnd\otimes \mfh)$. 

Recall that $E^{(1)}$ denotes the joint spectral resolution for the pair $P^{(1)},H^{(1)}$, cf.~\eqref{Hl-And-Pl}, as well as the notation $\hil_\iso = E(\Sigma_\iso)\hil$ for the
subspace of $\hil_\bnd$, consisting of isolated bound states \eqref{Hiso}. 
Finally, we remind the reader of the notation $\mcR\subset \real^{\nu+1}$ for the set of points $(\xi,\lambda)$, with $\lambda < \Sigma^{(1)}_0(\xi)$, i.e. the energy-momentum set below the two-boson threshold. For the purpose of this appendix we write $\mcC(O)\subset \hil_\bnd\otimes \mfh$ for the set of $O$-compatible pairs $(\Psi,h)$.
The following lemma characterizes the incoming and outgoing states below the two-boson threshold.
It is similar to \cite[Lemma~30]{FGSch2}.

\bel\label{energetic-considerations} Let  $O\subset\mcR$ be an open set. 
Then   
\begin{align}\label{PreFormOfAsymptoticStates}
& E^{\ex}(O)\hil^{\ex} = E(O)\hil\oplus  E^{(1)}(O)(\hil_\iso\otimes\mfh),\\
\label{FormOfAsymptoticStates}
& E^{(1)}(O)(\hil_\iso\otimes\mfh) = \ov{\Span\bigset{\Psi\otimes a^*(h)\Om }{ (\Psi,h)\in \mcC(O)}},\\
\label{AsympOnePartAreBound}
& E^{(1)}(O)(\hil\otimes\mfh)\subset \hil_\iso\otimes \mfh.
\end{align}
\eel
\begin{proof}  
 Let $\mathbf{1}_O$ be the characteristic function of $O$. Making use of the decomposition~(\ref{Hamiltonian-decomp}), 
 we compute
\begin{equation}
 \mathbf{1}_O\bigl(P^{\ex},H^{\ex}\bigr)=\mathbf{1}_O(P,H)\oplus\Bigl( \bigoplus_{\ell=1}^\infty \mathbf{1}_O\bigl(P^{(\ell)},H^{(\ell)}\bigr)\Bigr).
\label{chi-decomposition-formula}
\end{equation}
Since $O$ is located below the $2$-boson threshold $\Sigma^{(2)}_0$, the contributions from asymptotic particle
sectors, with $\ell\geq 2$, are zero. The range of the $0$'th summand is $E(O)\hil$ and the range of the $1$'st
summand is $E^{(1)}(O)(\hil\otimes\mfh)$. We are thus reduced to establishing the identity \eqref{FormOfAsymptoticStates} and the inclusion \eqref{AsympOnePartAreBound}.

Abbreviate
\begin{equation}
V := \overline{\Span\bigset{ \Psi\otimes a^*(h)\Om }{ (\Psi,h) \in \mcC(O)}}.
\end{equation}
Clearly, $V \subset \hil_\iso\otimes\mfh$. In order to prove \eqref{FormOfAsymptoticStates} we need to verify
$E^{(1)}(O)( \hil\otimes\mfh) = V$.

In the following we will make repeated use of the direct integral representation
\begin{equation}
I_\LLP^{(1)}(\Psi\otimes a^*(h)\vac) = \int^\oplus d\xi \int^\oplus  dk \, h(k) \Psi_{\xi-k}
\end{equation}
for simple tensors, with $\Psi\in\hil$ and $h\in\mfh$. This decomposition is the same as the one in Subsection~\ref{Extended-Hamiltonian}, cf.~\eqref{FirstExFib}, \eqref{MomentumFib}
and~\eqref{DoubleFib}.
If $\Psi\in\widetilde{\hil}_{\iso,n}$, cf.~\eqref{PreQuasiParticleSpace}, we can in particular compute:
\begin{align}
\non I_\LLP^{(1)} E^{(1)}(O)(\Psi\otimes a^*(h)\vac)& = \int^\oplus  d\xi \int^\oplus dk\, h(k) \mathbf{1}_O(\xi,H^{(1)}(\xi;k))\Psi_{\xi-k}\\
&  = \int^\oplus  d\xi  \int^\oplus dk \,h(k) \mathbf{1}_O(\xi,S_n(\xi-k)+\omega(k))\Psi_{\xi-k}.
\end{align}
If $\Psi$ and $h$ are $O$-compatible we see that for $k\in\supp\, h$ and $\xi$ such that $\Psi_{\xi-k}\neq 0$,
we must have $(\xi,S_n(\xi-k)+\omega(k))\in O$ and hence, by \eqref{QuasiParticleSpace} and a density argument, we have established that $V\subset E^{(1)}(O)(\hil\otimes\mfh)$. 

We proceed to show that $E^{(1)}(O)(\hil\otimes\mfh)\subset \hil_\iso\otimes\mfh$. For this it suffices to argue that
for any Borel set $U\subset\set{(\xi,\lambda)}{ \la\geq \Sigma^{(1)}_0(\xi)}$ and state $\varphi = E^{(1)}(O)(\Psi\otimes a^*(h)\Om)$, with $\Psi\in\hil$ and $h\in\mfh$,
we must have $(E(U)\otimes 1) \varphi = 0$. For this we compute
\begin{align}
\non & (E(U)\otimes 1) \varphi = (E(U)\otimes 1)I_\LLP^{(1)*} 
\int^\oplus d\xi \int^\oplus  dk\, h(k) \mathbf{1}_O(\xi,H^{(1)}(\xi;k))\Psi_{\xi-k}\\
& \quad =I_\LLP^{(1)*}  \int^\oplus d\xi \int^\oplus  dk \,h(k)\mathbf{1}_U(\xi-k,H(\xi-k)) \mathbf{1}_O(\xi,H(\xi-k)+\omega(k))\Psi_{\xi-k}.
\end{align}
For a point $(\xi-k,\mu)$ to be in $U$ we must have $\mu \geq \Sigma^{(1)}_0(\xi-k)$.
Hence, $\mu + \omega(k) \geq \Sigma^{(1)}_0(\xi-k)+\omega(k)\geq \Sigma^{(2)}(\xi)$.
Conversely, for a point $(\xi,\mu+\omega(k))$ to be in $O$, we must have $\mu + \omega(k) < \Sigma^{(2)}(\xi)$.
Since these two situations cannot occur simultaneously we conclude that 
$\mathbf{1}_U(\xi-k,H(\xi-k)) \mathbf{1}_O(\xi,H(\xi-k)+\omega(k)) = 0$. This concludes the proof of \eqref{AsympOnePartAreBound}. 

Consider a state of the form $\chi(P^{(1)},H^{(1)})(\Psi\otimes a^*(h)\vac)$, $\chi\in C_0^{\infty}(O)_{\real}$,
$\Psi\in\widetilde{\hil}_{\iso,n}$, and $h\in\mfh$ with compact essential support. To conclude the proof it suffices to show that 
such states can  be approximated by elements from $V$. Note that by \eqref{QuasiParticleSpace}, the spectral theorem, 
and the inclusion $E^{(1)}(O)(\hil\otimes\mfh)\subset \hil_\iso\otimes\mfh$ just proved: any state
in $E^{(1)}(O)(\hil\otimes\mfh)$ can be approximated using states of the considered form.

Put $r := d(\supp\,\chi,\real^{\nu+1}\backslash O)>0$. 
Let $\epsilon>0$ be given. We may assume $2\epsilon<r$. Using that $\chi$ is uniformly continuous
we get a $\delta'$, such that $|\chi(\xi',\mu')-\chi(\xi'',\mu'')|\leq \epsilon$, for $(\xi',\mu'),(\xi'',\mu'')$
 with $|\mu'-\mu''| <\delta'$ and $|\xi'-\xi''|<\delta'$. We may take $\delta' < \epsilon$.
Let $R>0$ be so large that $\supp\,h \subset \set{k\in\real^\nu}{|k| \leq R}$. 
Using that $\omega$ is also uniformly continuous on the ball of radius $R$,
we get a $0<\delta \leq \delta'$ such that $|\omega(k')-\omega(k'')|<\delta'$ if $|k'-k''|<\delta$.
 
Cover $B_R(0)$ with finitely many pairwise disjoint Borel sets $B_\ell$ such that $B_\ell\subset B_\delta(k_\ell)$, $\ell=1,\dotsc, L$, for some collection of momenta $k_1,\dotsc,k_L$.
Write
\begin{align}
\non & I_\LLP^{(1)}\chi(P^{(1)},H^{(1)})(\Psi\otimes a^*(h)\vac) = 
 \int^\oplus d\xi \int^\oplus  dk\, h(k) \chi(\xi,S_n(\xi-k)+\omega(k))\Psi_{\xi-k}\\
 &\qquad  = \sum_{\ell=1}^L \int^\oplus d\xi \int^\oplus  dk\, h(k) \mathbf{1}_{B_\ell}(k)\chi(\xi,S_n(\xi-k)+\omega(k))\Psi_{\xi-k}.
\end{align}
For $k\in B_\delta(k_\ell)$ we have
$|\chi(\xi;S_n(\xi-k)+\omega(k))- \chi(\xi-k+k_\ell;S_n(\xi-k)+\omega(k_\ell))|<\epsilon$. Define
\begin{equation}
\psi_\ell := (\chi_\ell(P,H)\Psi)\otimes a^*(\mathbf{1}_{B_\ell} h)\vac,
\end{equation}
with $\chi_\ell(\xi,\lambda) := \chi(\xi+k_\ell,\lambda+\omega(k_\ell))$. Then
$K_\ell:=\supp\,\chi_\ell = \supp\,\chi - (k_\ell,\omega(k_\ell))$. 
Note that $K_\ell\cap \Sigma \subset \Sigma_\iso$, and hence; $\psi_\ell\in \hil_{\iso}\otimes \mfh$.

Estimate
\begin{align}
\non &\Bigl\| I^{(1)}_\LLP\psi_\ell -  \int^\oplus d\xi \int^\oplus  dk\, h(k) \mathbf{1}_{B_\ell}(k)\chi(\xi,S_n(\xi-k)+\omega(k))\Psi_{\xi-k}\Bigr\|^2\\
\non &\quad = \int d\xi \int  dk\, \mathbf{1}_{B_\ell}(k)|h(k)|^2 \bigl|\chi_\ell(\xi-k,S_n(\xi-k))-
\chi(\xi,S_n(\xi-k)+\omega(k)) \bigr|^2 \|\Psi_{\xi-k}\|^2\\
& \quad \leq \epsilon^2 \|\Psi\|^2 \int dk\, \mathbf{1}_{B_\ell}(k)|h(k)|^2.
\end{align}
Due to  the fact that $B_{\ell}\cap B_{\ell'}=\emptyset$, summing up over $\ell$ yields
\begin{equation}
\bigl\|\chi(P^{(1)},H^{(1)})(\Psi\otimes a^*(h)\Om)-\sum_{\ell=1}^L\psi_\ell\bigr\|\leq \epsilon\|h\|\|\Psi\|.
\end{equation}

It remains to verify that $\chi_\ell(P,H)\Psi$ and $\mathbf{1}_{B_\ell} h$ are 
$O$-compatible, such that we in fact have $\psi_\ell\in V$.
Let $k\in B_\ell\subset B_\delta(k_\ell)$ and $(\xi,\mu)\in K_\ell$.
Then 
\begin{align}
\non (\xi + k, \mu + \omega(k)) & = (\xi+k_\ell,\mu+\omega(k_\ell)) +(k-k_\ell,\omega(k)-\omega(k_\ell))\\
& \in \supp\,\chi + (k-k_\ell,\omega(k)-\omega(k_\ell)).
\end{align}
By the choice of $\delta$ we conclude that $|(k-k_\ell,\omega(k)-\omega(k_\ell))| <\epsilon<r$ and hence
we have $(\xi + k, \mu + \omega(k))\in O$. This means that $\chi_\ell(P,H)\Psi$ and $\mathbf{1}_{B_\ell} h$ are 
$O$-compatible, which concludes the proof.
\end{proof}

\bel\label{Extended-spectrum-lemma} Let
\beqa
\Sigma^{(1)}_{\pp}:=\bigset{(\xi,\la)\in\real^{\nu+1} }{\la\in \si_{\pp}(H^{(1)}(\xi))}. 
\eeqa
Then $E^{(1)}(\Sigma^{(1)}_\pp\cap \mcE^{(1)}) =0$, hence the set 
$\set{ \xi\in \real^{\nu}}{ \sigma_{\pp}(H^{(1)}(\xi))\cap \mcE^{(1)}(\xi)\neq \emptyset }$
has zero Lebesgue measure. 
\eel
\begin{proof} Let us consider a  vector $\Psi\in  E^{(1)}(O)(\hil\otimes\Fock^{(1)})$, where $O\subset \mcE^{(1)}$ is some Borel subset.
Let $\mathbf{1}_O$ be the characteristic function of $O$. Then, making use of the expansion~(\ref{Hamiltonian-decomp}), we can write 
 \beqa
\Psi=
I_{\mathrm{LLP}}^{(1)*}\int^{\oplus} d\xi \int^{\oplus}dk\,\mathbf{1}_O(\xi, H^{(1)}(\xi;k)) \Psi_{\xi-k}. \label{vector-decomposition}
\eeqa
 Now suppose that $\Psi\in E^{(1)}(\Sigma^{(1)}_{\pp})(\hil\otimes\Fock^{(1)})$.
We note that  $\Sigma^{(1)}_{\pp}(\xi)\cap \mcE^{(1)}(\xi)$ can be at most countable due to the separability of $\Fock\otimes\Fock^{(1)}$.  Then,
by \cite[Th{\'e}or{\`e}me~21]{Retal}, $\Sigma^{(1)}_{\pp}\cap \mcE^{(1)}$ is a countable union of graphs of Borel functions 
from Borel subsets of $\real^{\nu}$ to $\real$. Thus, without loss of generality, we can assume that
  there exists a  Borel function $p\colon N \to\real$, defined on a Borel set $N$, s.t.
 $\Psi\in  \mathbf{1}_N(P^{(1)})(\hil\otimes\Fock^{(1)})$ and 
\beqa
H^{(1)}\Psi=p(P^{(1)})\Psi. \label{eigenvector-of-ext}
\eeqa
Suppose, by contradiction, that $\Psi\neq 0$ and satisfies  (\ref{eigenvector-of-ext}). Since $\mathbf{1}_O$ 
is supported  below the two-boson threshold, it is easy to see that 
\beqa
\Psi_{\xi}\in E_{\xi}((-\infty,\Sigma_0^{(1)}(\xi)))\Fock,
\eeqa
where $E_{\xi}$ is the spectral measure of $H(\xi)$. Consequently, $\Psi\in\hil_\iso\otimes\mfh$.
Hence,  there exists a shell $(\mcA, S)$ in $\Sigma_{\iso}$ s.t.
\beqa
\Psi':=(\mathbf{1}_{\mcG_S}(P,H)\otimes 1)\Psi\neq 0,
\eeqa
where $\mcG_S$ is the graph of $S$. 
Since $\mathbf{1}_{\mcG_S}(P,H)\otimes 1$ commutes with $H^{(1)}$, $P^{(1)}$, we obtain that $\Psi'$ also satisfies 
(\ref{eigenvector-of-ext}). Thus we obtain
\beqa
\int_N d\xi\, \int dk\, \bigl(S(\xi-k)+\om(k)-p(\xi)\bigr)^2\|\Psi'_{\xi-k}\|^2=0.
\eeqa
Hence the set of $\xi$ for which 
\begin{equation}
 \int dk\, \bigl(S(\xi-k)+\om(k)-p(\xi)\bigr)^2\|\Psi'_{\xi-k}\|^2\neq 0
\end{equation}
has zero Lebesgue measure. Conversely, the set of $\xi$ for which the real analytic function 
$k\to S(\xi-k)+\omega(k)$ is constant
also has zero Lebesgue measure by Corollary~\ref{Non-constant-corollary}.  
Since $k\to \Psi'_{\xi-k}$ has essential support of
positive Lebesgue measure, we conclude that the above integral can only vanish for a set of $\xi$'s having
zero Lebesgue measure.
This is a contradiction, which concludes the proof. \end{proof}



\begin{thebibliography}{99}
\providecommand{\bysame}{\leavevmode\hbox to3em{\hrulefill}\thinspace}

\bibitem{AT}
T.~Adachi and H.~Tamura,
\emph{Asymptotic completeness for long-range many-particle systems with {S}tark effect. {II}},
Commun. Math. Phys., \textbf{174} (1996), 537--559.

\bibitem{Am} Z.~Ammari,
\emph{Asymptotic completeness for a renormalized nonrelativistic {H}amiltonian in quantum field theory: the {N}elson model},
Math. Phys. Anal. Geom., \textbf{3} (2000), 217--285.


\bibitem{AMZ} N.~Angelescu, R.~A.~Minlos and V.~A.~Zagrebnov,
\emph{Lower spectral branches of a particle coupled to a Bose field},
Rev. Math. Phys., \textbf{17} (2005), 1111--1142.








\bibitem{CD82}   M. Combescure and F. Dunlop, \emph{Three-body asymptotic completeness for $P(\phi)_2$ models}, 
Commun. Math. Phys., \bf 85 \rm (1982), 381--418.

\bibitem{CK66}   R.~H. Cox and L.~C. Kurtz,  \emph{Real periodic functions}, 
Amer. Math. Monthly, \textbf{73}  (1966), 761--762.



\bibitem{DK11} W. De Roeck and A. Kupiainen, \emph{Approach to ground state and time-independent photon bound for massless spin-boson models},
To appear in Ann. Henri Poincar{\'e}.


\bibitem{D} J.~Derezi\'nski,
\emph{Asymptotic completeness of long-range $N$-body quantum systems},
Ann. of Math.~(2), \textbf{138} (1993), 427--476.


\bibitem{DGBook}   J.~Derezi\'nski and C.~G\'erard,
  \emph{Scattering theory of classical and quantum $N$-particle systems},
   Texts and Monographs in Physics, Springer-Verlag, Berlin, 1997.


\bibitem{DGe1} \bysame,
 \emph{Asymptotic completeness in quantum
  field theory. Massive {Pauli-Fierz} {Hamiltonians}}, Rev. Math. Phys.,
  \textbf{11} (1999), 383--450.

\bibitem{DGe2} \bysame,
\emph{Spectral and scattering theory of spatially cut-off {$P(\phi)_2$} {H}amiltonians},
Commun. Math. Phys., \textbf{213} (2000), 39--125.




\bibitem{Dy05} W. Dybalski, \emph{Haag-Ruelle scattering theory in presence of
massless particles}, Lett. Math. Phys., \bf 72 \rm  (2005), 27--38.

\bibitem{DT10} W.\! Dybalski and Y.\! Tanimoto, \emph{Asymptotic completeness in a class of
massless relativistic quantum field theories},  Commun. Math. Phys., \bf 305 \rm (2011), 427--440.  

\bibitem{DG12} W. Dybalski and C. G\'erard,\emph{Towards asymptotic completeness of two-particle scattering in local relativistic QFT}, 
ArXiv:1211.3393v3.
	

\bibitem{En78} V. Enss, \emph{Asymptotic completeness for quantum mechanical potential scattering},
Comm. Math. Phys. \textbf{61} (1978), 285--291. 




\bibitem{FMS11}  J.~Faupin, J.~S.~M{\o}ller and E.~Skibsted, \emph{Regularity of bound states}, 
Rev.  Math. Phys., \textbf{23}  (2011),   453--530.


\bibitem{FS12} J. Faupin and I.~M. Sigal, \emph{On quantum Huygens principle and Rayleigh scattering}, ArXiv: 1202.6151v1.  

\bibitem{FS12b} \bysame,
\emph{Comment on the photon number bound and Rayleigh scattering},
ArXiv:1207.4735v2.


\bibitem{FrH} H.~Fr\"ohlich,
\emph{Electrons in lattice fields},
Adv. in Phys., \textbf{3} (1954), 325--362.

\bibitem{FrJ2} J.~Fr\"ohlich,
\emph{Existence of dressed one-electron states in a class of persistent models},
Fortschr. Phys. \textbf{22} (1974), 159--198.



\bibitem{FGSch2} J.~Fr\"ohlich, M.~Griesemer and B.~Schlein,
  \emph{Asymptotic completeness for {R}ayleigh scattering},
  Ann. Henri Poincar{\'e}, \textbf{3} (2002), 107--170.

\bibitem{FGSch3} \bysame,
  \emph{Asymptotic completeness for {C}ompton scattering},
  Commun. Math. Phys., \textbf{252} (2004), 415--476.

\bibitem{FGSch4} \bysame,
\emph{Rayleigh scattering at atoms with dynamical nuclei},
Commun. Math. Phys., \textbf{271} (2007), 387--430.



\bibitem{Ge98} C.~G\'erard, \emph{Mourre estmate for regular dispersive systems},
Ann. Inst. H.~Poincar\'e, Phys. Th{\'e}o., \textbf{54} (1991),  59--88.

\bibitem{Ge1} \bysame,
\emph{On the scattering theory of massless Nelson models},
Rev. Math. Phys., \textbf{14} (2002), 1165--1280.


\bibitem{GeLaBook} \bysame,
\emph{Multiparticle quantum scattering in constant magnetic fields},
    Mathematical Surveys and Monographs \textbf{90},
   American Mathematical Society, Providence, RI, 2002.



\bibitem{GMR} C.~G\'erard, J.~S.~M{\o}ller and M.~G.~Rasmussen,
\emph{Asymptotic completeness in quantum field theory: Translation invariant Nelson type models restricted to the vacuum and one-particle sectors},
Lett. Math. Phys., \textbf{95} (2011), 109--134.

\bibitem{Gr} G.~M.~Graf,
\emph{Asymptotic completeness for {$N$}-body short-range quantum systems: a new proof},
Commun. Math. Phys., \textbf{132} (1990), 73--101.


\bibitem{GS97} G.~M. Graf and D.~Schenker, \emph{$2$-Magnon scattering in the Heisenberg model},
Ann. Inst. H.~Poincar\'e, Phys. Th{\'e}o., \textbf{67} (1997), 91--107.




\bibitem{HMS} I.~Herbst, J.~S.~M{\o}ller and E.~Skibsted,
\emph{Asymptotic completeness for $N$-body {S}tark {H}amiltonians},
Commun. Math. Phys., \textbf{174} (1996), 509--535.

\bibitem{HK68} R. H\o egh-Krohn, 
\emph{Asymptotic fields in some models of quantum field theory. I},
J. Math. Phys., \textbf{9} (1968), 2075-2079.



 \bibitem{HuSp1} M.~H\"ubner and H.~Spohn,
\emph{Spectral properties of the spin-boson Hamiltonian},
Ann. Inst. Henri Poincar\'e, \textbf{62} (1995), 289--323.

\bibitem{HuSp2} \bysame,
\emph{Radiative decay: nonperturbative approaches},
Rev. Math. Phys., \textbf{7} (1995), 363--387.



\bibitem{Ka} T. Kato, \emph{Perturbation theory for linear operators},
(second edition) Grundlehren der Mathematischen Wissenschaften, 132. Springer-Verlag, Berlin, 1976.

\bibitem{Le08} G.\! Lechner, \emph{Construction of quantum field theories with factorizing S-matrices},
Commun. Math. Phys., \bf 277 \rm  (2008), 821--860.


\bibitem{LLP}
T.~D. Lee, F.~E. Low, and D.~Pines, \emph{The motion of slow electrons in a
  polar crystal}, Phys. Rev., \textbf{90} (1953), 297--302.




\bibitem{MAHP} J.~S. M{\o}ller, \emph{The translation invariant massive {Nelson} model: {I}.
  The bottom of the spectrum}, Ann. Henri Poincar\'e \textbf{6} (2005),
  1091--1135.



\bibitem{MRMP} \bysame,
\emph{The {P}olaron revisited}, Rev. Math. Phys.,
  \textbf{18} (2006), 485--517.

\bibitem{MR12} J.~S. M{\o}ller and M~.G. Rasmussen, \emph{The translation invariant massive Nelson model: II. The continuous spectrum below the two-boson threshold}. To appear in Ann. Henri Poincar\'e.

\bibitem{Mo81} E. Mourre, \emph{Absence of singular continuous spectrum for certain self-adjoint operators}, 
Commun. Math. Phys., \textbf{78} (1981), 391--408. 


\bibitem{Ne} E.~Nelson,
\emph{Interaction of nonrelativistic particles with a quantized scalar field},
J.~Math. Phys., \textbf{5} (1964), 1190--1197.

\bibitem{Ra10} M.~G.~Rasmussen, \emph{A Taylor-like expansion of a commutator with a function of self-adjoint, pairwise commuting operators},
Math. Scand., \textbf{111} (2012), 107--117.


\bibitem{Retal}
C.~A.~Rogers et al., \emph{Analytic Sets} Academic Press, 1980.




\bibitem{SiSo}
I.~M.~Sigal, and A.~Soffer,
\emph{The {$N$}-particle scattering problem: asymptotic completeness for short-range systems},
Ann. of Math.~(2),  \textbf{126} (1987), 35--108.


\bibitem{SZ76} T. Spencer and F. Zirilli, \emph{Scattering states and bound states in $\la P(\phi)_2$},  Commun. Math. Phys., \bf 49 \rm (1976), 1--16.



\bibitem{Sp2} H.~Spohn,
\emph{The polaron at large total momentum}, J.~Phys.~A, \textbf{21} (1988), 1199--1211.

\bibitem{Sp97} \bysame,
\emph{Asymptotic completeness for Rayleigh scattering}, J.~Math. Phys., \textbf{38}  (1997), 2281--2296 .


\bibitem{Ya92} D.~R. Yafaev, \emph{Mathematical scattering theory. General theory}, American Mathematical Society, Providence, 1992. 





\end{thebibliography}
\end{document}